\documentclass[11pt,a4paper]{article}

\usepackage{hyperref}
\usepackage{cite}

\usepackage{amsmath}
\usepackage{amsfonts,amsthm,amssymb}

\numberwithin{equation}{section}

\newcommand{\beqa}{\begin{eqnarray}}
\newcommand{\eeqa}{\end{eqnarray}}

\newtheorem{theorem}{Theorem}[section]
\newtheorem{proposition}{Proposition}[section]

\newtheorem{corollary}{Corollary}[section]

{\theoremstyle{remark}
\newtheorem{rem}{Remark}[section]}
\newtheorem{identity}{Identity}
\newcommand{\tr}{\operatorname{tr}}

\newcommand{\End}{\operatorname{End}}

\newcommand{\bra}[1]{\langle\,#1\,|}
\newcommand{\ket}[1]{|\,#1\,\rangle}

\newcommand{\moy}[1]{\langle\,#1\,\rangle}

\def\sul{\sum\limits}
\def\pl{\prod\limits}

\newcommand{\la}{\lambda}

\begin{document}

\begin{flushright}
LPENSL-TH-07/18
\end{flushright}

\bigskip

\bigskip

\begin{center}
\textbf{{\Large The open XXZ  spin chain in the SoV framework:}}

\textbf{{\Large scalar product of separate states}}

\vspace{25pt}

\begin{large}

{\bf N.~Kitanine}\footnote[1]{IMB UMR5584, CNRS, Univ. Bourgogne Franche-Comt\'e, F-21000 Dijon, France; Nikolai.Kitanine@u-bourgogne.fr},~~
{\bf J.~M.~Maillet}\footnote[2]{Univ Lyon, Ens de Lyon, Univ Claude Bernard, CNRS, 
Laboratoire de Physique, F-69342 Lyon, France;
 maillet@ens-lyon.fr},~~
{\bf G. Niccoli}\footnote[3]{Univ Lyon, Ens de Lyon, Univ Claude Bernard, CNRS, 
Laboratoire de Physique, F-69342 Lyon, France; giuliano.niccoli@ens-lyon.fr},

{\bf V.~Terras}\footnote[4]{LPTMS, CNRS, Univ. Paris-Sud, Universit\'e Paris-Saclay, 91405 Orsay, France;\\ veronique.terras@lptms.u-psud.fr}
\end{large}

\vspace{45pt}


\end{center}


\begin{abstract}
In our previous paper \cite{KitMNT17}  we have obtained,  for the XXX spin-1/2 Heisenberg open chain, new  determinant representations for the scalar products of separate states  in the quantum  separation of variables (SoV) framework. 
 In this article we perform a similar study in a more complicated case: the XXZ open spin-1/2 chain with the most general integrable boundary terms.  
To solve this model by means of SoV we use an algebraic Vertex-IRF gauge transformation reducing one of the boundary $K$-matrices to a diagonal form.  
As usual within the SoV approach, the scalar products of separate states are computed in terms of dressed Vandermonde determinants having an intricate dependency on the inhomogeneity parameters. We show that these determinants can be transformed into different ones in which the homogeneous limit can be taken straightforwardly. These representations generalize in a non-trivial manner to the trigonometric case the expressions found previously in the rational case. 
 We also show that generically all scalar products can be expressed in a form which is similar to --- although more cumbersome than ---  the well-known Slavnov determinant representation for the scalar products of the Bethe states of the periodic chain. Considering a special choice of the boundary parameters relevant in the thermodynamic limit to describe the half infinite chain with a general boundary, we particularize these representations to the case of one of the two states being an eigenstate. We obtain simplified formulas that should be of direct use to compute the form factors and correlation functions of this model.
\end{abstract}

\newpage



\newpage

\section{Introduction}
 The resolution of quantum lattice models with the most general boundary conditions preserving integrability properties has become along the years a subject of intense studies \cite{Gau71a,Bar79,Bar80,Bar80a,Sch85,AlcBBBQ87,Che84,Skl85,Skl88,Bar88,KulS91,MezNR90,MezN91,PasS90,BatMNR90,deVG93,deVG94,GhoZ94, JimKKKM95, JimKKMW95, KitKMNST07,KitKMNST08,ArtMN95,Doi03,Doi06,CraR12, Nep04, BasK07, Nep02, NepR03, MurNS06, Gal08, YanNZ06, FraNR07, CraRS10, FilK10,FilK11,Fil11,CaoYSW13b, XuHYCYS16, DerKM03b, FraSW08, FraGSW11, AmiFOR07, AmiFOW10, Nic12, FalN14, FalKN14, KitMN14, FanHSY96, CaoLSW03, YanZ07, ArnCDFR05,ArnCDFR06, RagS07, BelR09,BasB13,BasB17,BelC13,BelCR13,Bel15,BelP15,AvaBGP15, Zho96, Zho97,AsaS96, GuaWY97,Gua00, ShiW97}. This interest is due, in particular, to their potential relevance for the investigation of the non-equilibrium and transport properties of quantum integrable systems \cite{KinWW06,HofLFSS07,BloDZ08,TroCFMSEB12,SchHRWBBBDMRR12,EisFG15,RonSBHLMHBS13,Pro11}.  
 Among several different approaches to this problem, the separation of variables (SoV) method \cite{Skl85,Skl85a,Skl92,Skl95}  has proven to produce cutting edge answers to give the possibility to construct  the full set of eigenvalues and eigenstates of the Hamiltonians, while providing the first steps towards the computation of the form factors and correlation functions \cite{BabBS96,BabBS97,Smi98a,DerKM03,BytT06,vonGIPS06,vonGIPS09,vonGIPST07, vonGIPST08, NicT10,Nic10a,GroN12,GroMN12,GroMN14,Nic13,Nic13b,Nic13a,NicT15,NicT16,LevNT16,KitMNT16,KitMNT17,MaiNP17,MaiNP18}.

 We would like to mention that numerous  other methods were used in this context, including algebraic Bethe ansatz and its modifications, analytic Bethe ansatz, Baxter $T$-$Q$ equation and $q$-Onsager algebras. 
 In our recent articles \cite{MaiNP17,KitMNT17,KitMNT16}
  we extensively discussed their connection to our  approach and their possible relevance for the computation of form factors and correlation functions.

Looking for the exact description of the dynamics of quantum integrable lattice models, the determinant representations for the scalar products of states, including eigenstates  of the Hamiltonian, plays a prominent role, as it was clearly demonstrated  in the Algebraic Bethe ansatz framework, see e.g., \cite{KitMT99,IzeKMT99,KitMT00} using the determinant formula \cite{Sla89,KitMT99}. It was shown there that they are expected to provide, together with the necessary resolution of the quantum inverse scattering problem \cite{KitMT99,MaiT00}, the key ingredients to compute form factors of local operators \cite{KitMT99}. These results in their turn lead to efficient expressions for the correlation functions \cite{KitMT00,KitMST05a,KitMST05b} that can be computed either numerically from their exact formulas \cite{CauM05,CauHM05,PerSCHMWA06,PerSCHMWA07,Mai07}, leading to direct predictions for experimental measurements, or from refined analysis, in particular of the form factor series \cite{KitMST05a,KitMST05b}, to their asymptotic behavior \cite{KitKMST09b,KozMS11a,KozMS11b,KozT11,KitKMST11b,KitKMST11a,KitKMST12,KitKMT14} leading to an exact derivation of the CFT predictions \cite{KozM15}. 

In the SoV framework it was shown already in  \cite{GroMN12} that determinant representations for scalar products of separate states, which include all eigenstates, can be obtained rather straightforwardly. They are given in terms of Vandermonde determinants dressed by the separate wave functions of the considered states. Hence for eigenstates these scalar products directly involve the $Q$-functions solving the Baxter $T$-$Q$ relation. It should be noted however that to apply the SoV method one needs to construct a separate basis. In Sklyanin's approach it is identified with the eigenstate basis of a distinguished operator of the Yang-Baxter or reflection algebra that is diagonalizable with simple spectrum. This is made possible by considering integrable lattice models with generic representations in each lattice site, meaning the presence of inhomogeneity parameters in generic positions. 

However, to be fully successful, the program of computing scalar products, form factors and then correlation functions \cite{KitMST05a,KitMST05b,KitKMST09b,KozMS11a,KozMS11b,KozT11,KitKMST11b,KitKMST11a,KitKMST12,KitKMT14}, needs determinant formulas for which the homogeneous and then the large volume limit can be tackled explicitly as in \cite{KitKMST11a}. Unfortunately, it became clear from the results of  \cite{GroMN12}, that although the obtained dressed Vandermonde determinants were easy to derive and appear quite universally from the SoV method, their homogeneous limit was very involved. This is due to the fact that the inhomogeneity parameters are somehow spread non-locally in the whole determinant thus preventing a clear extraction of the homogeneous limit. This motivated us to dedicate our main efforts to a resolution of this problem that appeared to be a crucial point for the applicability of the SoV method itself. Indeed, while SoV has proven to be efficient quite widely for determining the full spectrum of integrable quantum systems \cite{Skl85,Skl85a,Skl92,Skl95,BabBS96,BabBS97,Smi98,DerKM03,BytT06,vonGIPS06,vonGIPS09,vonGIPST07,vonGIPST08,NicT10,Nic10a,GroN12,GroMN12,NieWF09,Nic13,Nic13b,Nic13a,NicT16,NicT15,LevNT16,KitMNT16,KitMNT17,MaiNP17,MaiNP18}, even in cases where other methods fail, its use for computing the dynamics, in particular in the infinite volume limit,  was strongly dependent of the resolution of this homogeneous limit question. 

It turns out that we have been able recently to find an answer to this question for the Heisenberg XXX spin-1/2 chain both in the closed \cite{KitMNT16} and open \cite{KitMNT17} cases.  It is based on a rather beautiful correspondence between Vandermonde, Izergin \cite{Ize87,Tsu98} and generalized Slavnov type \cite{Sla89,FodW12} determinants. This approach is similar to the one used for the study of the semi-classical limit of scalar products \cite{Kos12,Kos12a,KosM12}. Using such a rewriting the homogeneous limit can be taken quite easily. What is more, these alternative determinants, written in terms of completely different matrices, sometimes of different dimensions, can be obtained using purely algebraic identities. Hence, it gave hopes that similar rewritings can be derived for other models of interest, the XXZ spin-1/2 Heisenberg models with general boundaries being one of the first on such a wish list. The aim of the present article is to present the corresponding results for this model. As such, it can be considered as the continuation of \cite{KitMNT17}. It should be stressed however that going from rational to trigonometric case is not as trivial as one could think at first sight. The reduction of the global symmetry together with the change of the basis of functions involved  have direct consequences on the transformations one can use to solve the problem. Hence part of the methods developed in \cite{KitMNT17} have to be adapted and generalized. Nevertheless the results we obtain in the present article are quite similar to the rational case, although slightly more complex. In particular the determinant formulas relevant for describing an half-infinite chain with a general boundary can be tackled very nicely.
  
This article is organized as follows. In Section~\ref{sec-model} we recall the definition of the Heisenberg XXZ spin-1/2 open chain with the most general boundary conditions preserving integrability together with the associated reflection algebra and associated boundary $K$-matrices. In Section~\ref{sec-gauged} we use the vertex-IRF gauge transformation \cite{Bax73a,FadT79,FanHSY96} in its algebraic form~\cite{FelV96b} to put into correspondence this model having the most general $K$-matrices with a model of SOS type for which one of the $K$-matrices has a diagonal form. This connection to the generators of the SOS monodromy matrix allows us to compute for the first time explicitly the normalization of the elements of the left and right SoV basis, as described in Appendix~\ref{Normalization-Factor}. In Section~\ref{sec-diag} and in the first part of Section~\ref{sec-sp}, we recall and present in a uniformized manner results previously derived in \cite{Nic12,FalKN14,KitMN14}. More precisely, in Section~\ref{sec-diag} we recall the implementation of the separation of variables method for this model and characterize the complete spectrum, i.e., eigenvalues and eigenstates, of the transfer matrix. The result is also written in terms of a functional $T$-$Q$ equation with an inhomogeneous term, and we specify the constraints on the boundary parameters for which this inhomogeneous term vanishes. We also describe the set of separate states, namely states having a separate wave function in the SoV basis, that include all eigenstates. In Section~\ref{sec-sp} we first recall the scalar products of the separate states in terms of dressed Vandermonde determinants, as it is usual in integrable systems solved by SoV. Then we give the main result of this article: the rewriting of these scalar products in terms of new dressed  Vandermonde determinants with one modified column for which the homogeneous limit can be taken easily, and ultimately in terms of some generalized version of the Slavnov determinant \cite{Sla89,FodW12}. All technical details and proofs are gathered in a set of five appendices. 
\section{The open spin-1/2 XXZ quantum chain}
\label{sec-model}

The Hamiltonian of the open spin-1/2 quantum XXZ chain with the most general non-diagonal integrable boundary terms can be written in the following form:
\begin{align}
    H &=\sum_{n=1}^{N-1} \Big[ \sigma _n^{x} \sigma _{n+1}^{x}
    +\sigma_n^{y}\sigma _{n+1}^{y}
    +\cosh \eta \, \sigma _n^{z} \sigma _{n+1}^{z}\Big]
    \nonumber\\
    &\hspace{2cm}
    +\frac{\sinh \eta }{\sinh \varsigma _{-}} \Big[ \sigma _{1}^{z}\cosh \varsigma_{-}
    +2\kappa _{-} \big(\sigma _{1}^{x}\cosh \tau _{-}
    +i\sigma _{1}^{y}\sinh \tau_{-} \big)\Big] 
     \nonumber\\
    &\hspace{2cm}
    +\frac{\sinh \eta }{\sinh \varsigma _{+}} \Big[\sigma _{\mathsf{N}}^{z}\cosh \varsigma_{+}
    +2\kappa _{+} \big(\sigma _{N}^{x}\cosh \tau _{+}+i\sigma _{N}^{y}\sinh \tau _{+}\big)\Big]. 
     \label{H-XXZ-Non-D}
\end{align}
This is an operator acting on the quantum space of states  $\mathcal{H}=\otimes _{n=1}^{\mathsf{N}}\mathcal{H}_{n}$ of the chain, where $\mathcal{H}_n\simeq \mathbb{C}^{2}$ is the 
bidimensional local quantum spin space at site $n$, on which the operators $\sigma _{n}^{\alpha},\ \alpha\in\{x,y,z\}$, act as the corresponding Pauli matrices.
In \eqref{H-XXZ-Non-D}, $\Delta =\cosh \eta $ is the anisotropy parameter, and $\varsigma _{\pm}$, $\kappa _{\pm }$, $\tau _{\pm }$ parametrize the most general non-diagonal integrable boundary interactions.
It may sometimes be convenient to use different sets of boundary parameters $\alpha_\pm, \beta_\pm$ instead of $\varsigma_\pm,\kappa_\pm$, by using the following reparametrization:
\begin{equation}
   \sinh\alpha_\pm\, \cosh \beta_\pm =\frac{\sinh\varsigma_\pm}{2\kappa_\pm},
   \qquad
   \cosh \alpha_\pm\, \sinh\beta_\pm =\frac{\cosh\varsigma_\pm}{2\kappa_\pm}.
\label{reparam-bords}
\end{equation}

The open spin-1/2 XXZ chain can be studied in the framework of the representation theory of the reflection algebra, by considering monodromy matrices $\mathcal{U}(\lambda )\in\End(\mathbb{C}^2\otimes \mathcal{H})$ satisfying the following reflection equation, on $\mathbb{C}^2\otimes\mathbb{C}^2\otimes\mathcal{H}$:
\begin{equation}
R_{21}(\lambda -\mu )\,\mathcal{U}_{1}(\lambda )\,R_{12}(\lambda +\mu -\eta)\,\mathcal{U}_{2}(\mu )
=\mathcal{U}_{2}(\mu )\,R_{21}(\lambda +\mu -\eta )\,\mathcal{U}_{1}(\lambda )\,R_{12}(\lambda -\mu ).  \label{bYB}
\end{equation}
In this relation, the subscripts parameterize the subspaces of $\mathbb{C}^2\otimes\mathbb{C}^2$ on which the corresponding operator acts non-trivially.
The $R$-matrix,
\begin{equation}\label{R-6V}
R_{12}(\lambda )=\begin{pmatrix}
\sinh (\lambda +\eta ) & 0 & 0 & 0 \\ 
0 & \sinh \lambda & \sinh \eta & 0 \\ 
0 & \sinh \eta & \sinh \lambda & 0 \\ 
0 & 0 & 0 & \sinh (\lambda +\eta )
\end{pmatrix}
 \in \text{End}(\mathbb{C}^2\otimes \mathbb{C}^2),
\end{equation}
is the 6-vertex trigonometric solution of the Yang-Baxter equation, and $R_{21}(\lambda)=P_{12}\, R_{12}(\lambda)\, P_{12}$, where $P_{12}$ is the permutation operator on $\mathbb{C}^2\otimes \mathbb{C}^2$.
Note that, in the case \eqref{R-6V}, $R_{21}(\lambda)=R_{12}(\lambda)$.

Following Sklyanin~\cite{Skl88}, we define two classes of solutions $\mathcal{V}_-(\lambda)\equiv\mathcal{U}_{-}(\lambda )$ and $\mathcal{V}_{+}(\lambda )\equiv\mathcal{U}_{+}^{t_{0}}(-\lambda )$ of \eqref{bYB} by considering the operators
\begin{align}
   &\mathcal{U}_{-}(\lambda )
   = M(\lambda)\, K_-(\lambda)\, \hat{M}(\lambda)
   = \begin{pmatrix} \mathcal{A}_-(\lambda) & \mathcal{B}_-(\lambda) \\
                               \mathcal{C}_-(\lambda) & \mathcal{D}_-(\lambda) \end{pmatrix},
   \label{def-U-}\\
    &\mathcal{U}_{+}^{t_0}(\lambda )
   = M^{t_0}(\lambda)\, K_+^{t_0}(\lambda)\, \hat{M}^{t_0}(\lambda)
   = \begin{pmatrix} \mathcal{A}_+(\lambda) & \mathcal{C}_+(\lambda) \\
                               \mathcal{B}_+(\lambda) & \mathcal{D}_+(\lambda) \end{pmatrix}.
   \label{def-U+}
\end{align}
Both operators are acting on $\mathcal{H}_0\otimes \mathcal{H}$, where $\mathcal{H}_0=\mathbb{C}^2$ is called auxiliary space. They are defined in terms of
\begin{equation}
M(\lambda )=R_{0 N}(\lambda -\xi _N-\eta /2)\dots R_{01}(\lambda -\xi _{1}-\eta /2)
=
\begin{pmatrix}
A(\lambda ) & B(\lambda ) \\ 
C(\lambda ) & D(\lambda )%
\end{pmatrix} ,  \label{bulk-mon}
\end{equation}
which corresponds to the bulk monodromy matrix of a chain of length $N$ with some inhomogeneity parameters $\xi_{1},\ldots,\xi_{N}$, and of
\begin{equation}
\hat{M}(\lambda )=(-1)^N\,\sigma _{0}^{y}\,M^{t_{0}}(-\lambda
)\,\sigma _{0}^{y}.  \label{Mhat}
\end{equation}
The boundary matrices $K_\pm(\lambda)\in\End(\mathcal{H}_0)$ are here defined as
\begin{equation}
K_{-}(\lambda )=K(\lambda ;\varsigma _{-},\kappa _{-},\tau _{-}),
\qquad
K_{+}(\lambda )=K(\lambda +\eta ;\varsigma _{+},\kappa _{+},\tau _{+}),
\end{equation}
where
\begin{equation}
K(\lambda ;\varsigma ,\kappa ,\tau )=\frac{1}{\sinh \varsigma }\,
\begin{pmatrix}
\sinh (\lambda -\eta /2+\varsigma ) & \kappa e^{\tau }\sinh (2\lambda -\eta ) \\ 
\kappa e^{-\tau }\sinh (2\lambda -\eta ) & \sinh (\varsigma -\lambda +\eta /2)
\end{pmatrix}, 
 \label{mat-K}
\end{equation}
is the most general scalar solution \cite{deVG93,deVG94,GhoZ94} of the reflection equation \eqref{bYB} for general
values of the parameters $\varsigma ,$ $\kappa $ and $\tau $.

Sklyanin \cite{Skl88} has shown that the transfer matrices,
\begin{align}
\mathcal{T}(\lambda )
&=\tr_{0}\{K_{+}(\lambda )\,M(\lambda)\,K_{-}(\lambda )\, \hat{M}(\lambda )\}
   \nonumber\\
&=\tr_{0}\{K_{+}(\lambda )\, \mathcal{U}_{-}(\lambda )\}
  =\tr_{0}\{K_{-}(\lambda )\, \mathcal{U}_{+}(\lambda)\},
  \label{transfer}
\end{align}
form a one-parameter family of commuting operators on $\mathcal{H}$.
In the homogeneous limit in which $\xi _{m}=0$, $m=1,\ldots ,N$, the Hamiltonian \eqref{H-XXZ-Non-D} of the spin-1/2 open chain can be obtained as
\begin{equation}
H=\frac{2\, (\sinh \eta )^{1-2 N}}{\tr \{K_{+}(\eta /2)\}\,\tr\{K_{-}(\eta /2)\}}\,
\frac{d}{d\lambda }\mathcal{T}(\lambda )_{\,\vrule height13ptdepth1pt\>{\lambda =\eta /2}\!}
+\text{constant.}  \label{Ht}
\end{equation}

To conclude this section, let us recall some useful properties of the 6-vertex reflection algebra. The inversion relation for $\mathcal{U}_-(\lambda)$ can be written as
\begin{equation}\label{inv-U-}
   \mathcal{U}_-(\lambda+\eta/2)\, \mathcal{U}_-(-\lambda+\eta/2)
   =\frac{\det_q\mathcal{U}_-(\lambda)}{\sinh(2\lambda-2\eta)},
\end{equation}
in terms of the quantum determinant $\det_q\mathcal{U}_-(\lambda)$,
\begin{align}
  \frac{\det_q\mathcal{U}_-(\lambda)}{\sinh(2\lambda-2\eta)}
   &=\mathcal{A}_-(\eta/2\pm\lambda)\,\mathcal{A}_-(\eta/2\mp\lambda)
        +\mathcal{B}_-(\eta/2\pm\lambda)\,\mathcal{C}_-(\eta/2\mp\lambda)
        \nonumber\\
    &=\mathcal{D}_-(\eta/2\pm\lambda)\,\mathcal{D}_-(\eta/2\mp\lambda)
        +\mathcal{C}_-(\eta/2\pm\lambda)\,\mathcal{B}_-(\eta/2\mp\lambda),
   \label{det-U-}
\end{align}
which is a central element of the reflection algebra: $[\det_q\mathcal{U}_-(\lambda),\mathcal{U}_-(\lambda)]=0$.
It is obtained as the product
\begin{equation}\label{det-prod}
   \mathrm{det}_q \,\mathcal{U}_-(\lambda)=\mathrm{\det}_q M(\lambda)\, \mathrm{det}_q M(-\lambda)\, \mathrm{det}_q K_-(\lambda),
\end{equation}
of the bulk quantum determinant
\begin{equation}\label{det-M}
  \mathrm{det}_q M(\lambda)=a(\lambda+\eta/2)\, d(\lambda-\eta/2),
\end{equation}
where
\begin{equation}\label{a-d}
   a(\lambda)=\prod_{n=1}^N\sinh(\lambda-\xi_n+\eta/2),
   \qquad
   d(\lambda)=\prod_{n=1}^N\sinh(\lambda-\xi_n-\eta/2),
\end{equation}
and of the quantum determinant of the scalar boundary matrix $K_-(\lambda)$.
Similar results can be obtained for $\mathcal{U}_+(\lambda)$ using the fact that $\mathcal{V}_+(\lambda)\equiv \mathcal{U}_+^{t_0}(-\lambda)$ satisfies the same algebra as $\mathcal{U}_-(\lambda)$.
The quantum determinant of the scalar boundary matrices $K_\mp(\lambda)$ can be expressed as
\begin{align}
  \frac{ \det_q K_\mp(\lambda)}{\sinh(2\lambda\mp2\eta)}
  &= \mp\frac{\big(\sinh^2\lambda-\sinh^2\alpha_\mp\big)\big(\sinh^2\lambda+\cosh^2\beta_\mp\big)}{\sinh^2\alpha_\mp\,\cosh^2\beta_\mp}.
  \label{det-K-}
\end{align}
%
%

\section{Gauge transformation of the model}
\label{sec-gauged}

It is possible to solve the model, i.e. to diagonalize the boundary transfer matrices \eqref{transfer}, by means of the quantum version of the Separation of Variable (SoV) approach \cite{Skl85,Skl85a,Skl92,Skl95}. This has been done in \cite{FalKN14}. The idea is, as in the XXX case \cite{KitMNT17}, to gauge transform the model into an effective one in which at least one of the boundary matrices becomes triangular \cite{Nic12}. The XXZ case is however much more complicated than the XXX case, since the involved gauged transformation is a generalized (or dynamical) one. In this section, we reformulate the generalized gauge transformation used in \cite{FalKN14} in a more usual way\footnote{The generalized gauge transformation presented here is equivalent to the one of \cite{FalKN14}. However, the notations and the objects that are considered may sometimes be slightly different.}, by using the trigonometric version of the Vertex-IRF transformation \cite{Bax73a} in its algebraic form \cite{FelV96b}.

\subsection{Vertex-IRF transformation}

The Vertex-IRF transformation relates the 6-vertex $R$-matrix \eqref{R-6V} to the $R$-matrix of the trigonometric solid-on-solid (SOS) model:
\begin{equation}\label{Vertex-IRF}
    R_{12}(\la-\mu)\, S_1(\la|\beta) \, S_2(\mu|\beta+\sigma_1^z)
    =S_2(\mu|\beta) \, S_1(\la|\beta+\sigma_2^z) \, R^{\mathrm{SOS}}_{12}(\la-\mu|\beta),
\end{equation}
where the trigonometric SOS (or dynamical) $R$-matrix reads:
\begin{equation}\label{R-SOS}
     R^{\mathrm{SOS}}(\lambda|\beta)
     =\begin{pmatrix}
       \sinh(\la+\eta) &0&0&0\\
       0& \frac{\sinh(\eta(\beta+1))}{\sinh (\eta\beta)}\, \sinh \la & 
                        \frac{\sinh(\la+\eta\beta)}{\sinh (\eta\beta)}\,\sinh \eta & 0\\
       0& \frac{\sinh(\eta\beta-\la)}{\sinh (\eta\beta)}\, \sinh \eta & 
                        \frac{\sinh(\eta(\beta-1))}{\sinh (\eta\beta)}\,\sinh \la &0\\
       0&0&0&\sinh(\la+\eta)
        \end{pmatrix}.
\end{equation}
In this context, the parameter $\beta$ is usually called dynamical parameter.
The corresponding Vertex-IRF transformation matrix can be written as
\begin{equation}\label{mat-S}
     S(\lambda|\beta)
     =\begin{pmatrix} e^{\la-\eta(\beta+\alpha)} & e^{\la+\eta(\beta-\alpha)}  \\
                                1&1\end{pmatrix},
\end{equation}
where the parameter $\alpha$ corresponds to an arbitrary shift of the spectral parameter.
Note that the relation \eqref{Vertex-IRF} can equivalently be written as
\begin{equation}\label{Vertex-IRF2}
       R_{12}(\la-\mu)\, S_2(-\mu|\beta) \, S_1(-\la|\beta+\sigma_2^z)
       =S_1(-\la|\beta) \, S_2(-\mu|\beta+\sigma_1^z) \, R^{\mathrm{SOS}}_{21}(\la-\mu|\beta).
\end{equation}
%

\subsection{Gauge transformed monodromy matrices}

The Vertex-IRF transformation can easily be extended to a transformation between bulk monodromy matrices:
\begin{multline}\label{V-IRF-bulk1}
      M (\la)\,S_{1\dots N}  (\{\xi\}|\beta)\, S_0(-\la+\eta/ 2|\beta+\mathbf{S}^z)
      \\
     = S_0(-\la+\eta/ 2|\beta)\,S_{1\dots N}(\{\xi\}|\beta+\sigma^z_0)\,M^{\mathrm{SOS}}(\la|\beta),
\end{multline}
where $M(\lambda)\in\End(\mathcal{H}_0\otimes\mathcal{H})$ is the bulk monodromy matrix \eqref{bulk-mon}, $\mathbf{S}^z=\sum_{j=1}^N\sigma_j^z$, $S_0(\lambda|\beta)$ denotes the Vertex-IRF transformation matrix \eqref{mat-S} acting on the auxiliary space $\mathcal{H}_0$, whereas $S_{1\ldots N}(\{\xi\}|\beta)$ is the following product of local gauge matrices \eqref{mat-S} acting on the tensor product $\mathcal{H}=\otimes_{n=1}^N\mathcal{H}_n$ of the $N$ local quantum spaces:
\begin{align}
      S_{1\dots N}(\{\xi\}|\beta)
      &= S_N(-\xi_n|\beta)\, S_{N-1}(-\xi_{N-1}|\beta+\sigma_N^z)\ldots S_1(-\xi_1|\beta+\sigma_2^z +\ldots+\sigma_N^z)
      \nonumber\\
      &= \pl_{n=N\rightarrow 1}S_n\Bigg(-\xi_n \, \Big| \, \beta+  \sum_{j=n+1}^N\sigma_j^z\Bigg).
      \label{Sq}
\end{align}
The resulting gauged transformed bulk monodromy matrix $M^\mathrm{SOS}(\lambda|\beta)$ is  defined as
\begin{align}
M^{\mathrm{SOS}}(\la|\beta)
     &=\pl_{n=N\rightarrow 1}R^{\mathrm{SOS}}_{n0}\Bigg(\la-\xi_n-\frac \eta 2\, \Big|\, \beta+\sum_{j=n+1}^N\sigma_j^z\Bigg)
     \nonumber\\
     &= \begin{pmatrix}
          A^{\mathrm{SOS}}(\la|\beta)&B^{\mathrm{SOS}}(\la|\beta)\\
          C^{\mathrm{SOS}}(\la|\beta)&D^{\mathrm{SOS}}(\la|\beta)
          \end{pmatrix}.
    \label{M-SOS}
\end{align}
In these expressions, we have used the following notation concerning an ordered product of non-commuting operators:
\begin{equation}
    \pl_{n=N\rightarrow 1}X_n\equiv X_N\,X_{N-1}\dots X_1.
\end{equation}
In a similar way, we obtain the following transformation for the matrix $\hat{M}(\la)$ \eqref{Mhat}:
\begin{multline}
    \hat{M}(\la)\,S_0 (\la-\eta/ 2|\beta)\,S_{1\dots N}(\{\xi\}|\beta+\sigma^z_0)
    \\
    = S_{1\dots N} (\{\xi\}|\beta )\,S_0 (\la-\eta / 2|\beta+\mathbf{S}^z)\,\hat{M}^{\mathrm{SOS}}(\la|\beta),
\end{multline}
where
\begin{align}
   \hat{M}^{\mathrm{SOS}}(\la|\beta)
   &=\pl_{n=1\rightarrow N}
       R^{\mathrm{SOS}}_{0 n}\Bigg(\la+\xi_n-\frac \eta 2\, \Big|\, \beta+\sul_{j=n+1}^N\sigma_j^z\Bigg)
       \nonumber\\
    &=\begin{pmatrix}\hat{A}^{\mathrm{SOS}}(\la|\beta)&\hat{B}^{\mathrm{SOS}}(\la|\beta)\\
            \hat{C}^{\mathrm{SOS}}(\la|\beta)&\hat{D}^{\mathrm{SOS}}(\la|\beta)\end{pmatrix}.
    \label{Mhat-SOS}
\end{align}

Let us introduce the following gauged transformed versions of the boundary monodromy matrix $\mathcal{U}_-(\lambda)$:
\begin{align}
        \widetilde{\mathcal{U}}_-(\la|\beta)
        &=S_0^{-1}(-\la+\eta/2  | \beta )\ \mathcal{U}_-(\la)\ S_0(\la-\eta/2|\beta)
        \nonumber\\
        &=\begin{pmatrix}
         \widetilde{\mathcal{A}}_-(\lambda|\beta) & \widetilde{\mathcal{B}}_-(\lambda|\beta) \\
        \widetilde{\mathcal{C}}_-(\lambda|\beta) & \widetilde{\mathcal{D}}_-(\lambda|\beta) \end{pmatrix},
        \label{gauged-U}
\end{align}
and
\begin{align}
    \mathcal{U}^{\mathrm{SOS}}_-(\la|\beta)
         &= S^{-1}_{1 \dots N}(\{\xi\}|\beta+\sigma^z_0)\,\widetilde{\mathcal{U}}_-(\la|\beta)\,
              S_{1\dots N}(\{\xi\}|\beta+\sigma^z_0)
              \nonumber\\
         &=\begin{pmatrix}
             \mathcal{A}^{\mathrm{SOS}}_-(\la|\beta)&\mathcal{B}^{\mathrm{SOS}}_-(\la|\beta)\\
             \mathcal{C}^{\mathrm{SOS}}_-(\la|\beta)&\mathcal{D}^{\mathrm{SOS}}_-(\la|\beta)
             \end{pmatrix}.
        \label{U-SOS}
\end{align}
It is easy to see that both \eqref{gauged-U} and \eqref{U-SOS} satisfy the following dynamical reflection equation:
\begin{multline}\label{dyn_refl}
 R^{\mathrm{SOS}}_{21}(\la-\mu|\beta)\,
 \mathcal{U}_1(\la|\beta+\sigma^z_2)\,R^{\mathrm{SOS}}_{12}(\la+\mu-\eta|\beta)\,
 \mathcal{U}_2(\mu|\beta+\sigma^z_1)
 \\
 = \mathcal{U}_2(\mu|\beta+\sigma^z_1)\,R^{\mathrm{SOS}}_{21}(\la+\mu-\eta|\beta)\,
    \mathcal{U}_1(\la|\beta+\sigma^z_2)\, R^{\mathrm{SOS}}_{12}(\la-\mu|\beta).
\end{multline}
A few useful commutation relations issued from \eqref{dyn_refl} are specified in Appendix~\ref{app-prop-gauged}.
Moreover, the elements of the matrix \eqref{gauged-U} can easily be expressed in terms of the elements of the matrix \eqref{def-U-} (see \eqref{Atilde}-\eqref{Btilde}).
Instead, the matrix \eqref{U-SOS} can be expressed in terms of the SOS bulk monodromy matrices \eqref{M-SOS} and \eqref{Mhat-SOS} as
\begin{equation}
     \mathcal{U}^{\mathrm{SOS}}_-(\la|\beta)
     =	M^{\mathrm{SOS}}(\la|\beta)\,\mathcal{K}_-^{\mathrm{SOS}}(\la|\beta+\mathbf{S}^z)\,
        \hat{M}^{\mathrm{SOS}}(\la|\beta),
\label{boundary-bulk}
\end{equation}
where 
\begin{align}
     \mathcal{K}_-^{\mathrm{SOS}}(\la|\beta)
     &=S^{-1}_0(-\la+\eta /2|\beta)\,K_-(\la)\,S_0(\la-\eta/ 2|\beta)
     \nonumber\\
     &=\begin{pmatrix}\mathsf{a}_-(\la|\beta)&\mathsf{b}_-(\la|\beta)\\
           \mathsf{c}_-(\la|\beta)&\mathsf{d}_-(\la|\beta)\end{pmatrix}.
     \label{K-SOS}
\end{align}
In particular, the expression for $\mathsf{b}_-(\la|\beta)$ and $\mathsf{c}_-(\la|\beta)$ in \eqref{K-SOS}  is given by
\begin{equation}\label{b-lambda}
  \mathsf{b}_-(\la|\beta)
  = \mathsf{c}_-(\la|-\beta)=e^{\lambda-\eta/2}\sinh(2\lambda-\eta)\, \mathsf{b}_-(\beta),
\end{equation}
where $\mathsf{b}_-(\beta)$ depends on $\beta$, $\alpha$, and on the boundary parameters as
\begin{align}
 \mathsf{b}_-(\beta)
  &=\frac{e^{\eta\beta}}{2\sinh(\eta\beta)\,\sinh\varsigma_-}\Big[2\kappa_-\sinh(\eta(\beta-\alpha)-\tau_-)-e^{\varsigma_-}\Big]
   \nonumber \\
  &= \frac{\kappa_- \, e^{\eta\beta}  }{\sinh(\eta\beta)\, \sinh\varsigma_-}\Big[\sinh(\eta(\beta-\alpha)-\tau_-)-\sinh(\alpha_-+\beta_-)\Big].\label{b-}
\end{align}

The inversion relation for the matrix $\mathcal{U}^\mathrm{SOS}(\lambda|\beta)$ follows directly from the inversion relation \eqref{inv-U-} for $\mathcal{U}_-(\lambda)$:
\begin{equation}
   \mathcal{U}^{\mathrm{SOS}}_-(\lambda+\eta/2|\beta)\
   \mathcal{U}^{\mathrm{SOS}}_-(-\lambda+\eta/2|\beta)
   =\frac{\det_q\mathcal{U}_-(\lambda)}{\sinh(2\lambda-2\eta)},
\end{equation}
where $\det_q\mathcal{U}_-(\lambda)$ is the quantum determinant \eqref{det-U-}-\eqref{det-prod}.
Hence we also have
\begin{align}
  \frac{\det_q\mathcal{U}_-(\lambda)}{\sinh(2\lambda-2\eta)}
   &=\mathcal{A}^\mathrm{SOS}_-(\eta/2+\epsilon\lambda|\beta)\,\mathcal{A}^\mathrm{SOS}_-(\eta/2-\epsilon\lambda|\beta)
    \nonumber\\
    &\hspace{4cm}
        +\mathcal{B}^\mathrm{SOS}_-(\eta/2+\epsilon\lambda|\beta)\,\mathcal{C}^\mathrm{SOS}_-(\eta/2-\epsilon\lambda|\beta)
        \nonumber\\
    &=\mathcal{D}^\mathrm{SOS}_-(\eta/2+\epsilon\lambda|\beta)\,\mathcal{D}^\mathrm{SOS}_-(\eta/2-\epsilon\lambda|\beta)
      \nonumber\\
    &\hspace{4cm}
        +\mathcal{C}^\mathrm{SOS}_-(\eta/2+\epsilon\lambda|\beta)\,\mathcal{B}^\mathrm{SOS}_-(\eta/2-\epsilon\lambda|\beta).
   \label{det-U-SOS}
\end{align}
%

\subsection{Transfer matrix and gauge for $K_+(\la)$}

It remains to express the transfer matrix \eqref{transfer} $\mathcal{T}(\la)=\tr_0 \{ K_+(\la)\,\mathcal{U}_-(\la)\}$ in terms of the elements of the gauged monodromy matrix \eqref{gauged-U} or \eqref{U-SOS}.
A natural way to do this would be to apply the gauge transformation  \eqref{gauged-U} to $\mathcal{U}_-(\la)$ inside the trace, which would result into an expression for $\mathcal{T}(\la)$ in terms of $\widetilde{\mathcal{A}}_-(\la|\beta)$, $\widetilde{\mathcal{B}}_-(\la|\beta)$,
$\widetilde{\mathcal{C}}_-(\la|\beta)$ and $\widetilde{\mathcal{D}}_-(\la|\beta)$.
Such a representation would however not be so convenient, since the natural commutation relations issued from \eqref{dyn_refl} are  established between $\widetilde{\mathcal{A}}_-(\la|\beta-1)$ and $\widetilde{\mathcal{D}}_-(\la|\beta+1)$ (and not between $\widetilde{\mathcal{A}}_-(\la|\beta)$ and $\widetilde{\mathcal{D}}_-(\la|\beta)$, see \eqref{comm-AB}-\eqref{comm-BD}). The same is true for the expression of $\widetilde{\mathcal{D}}_-(\la|\beta+1)$ in terms of $\widetilde{\mathcal{A}}_-(\la|\beta-1)$ and $\widetilde{\mathcal{A}}_-(-\la|\beta-1)$ (see \eqref{parity-A}-\eqref{parity-D}).
It is therefore better to introduce a slightly modified gauge transformation which, when applied to  $\mathcal{U}_-(\la)$, produces $\mathcal{A}_-(\la|\beta-1)$ and $\mathcal{D}_-(\la|\beta+1)$.
To this aim, one can for instance consider the following modified gauged boundary monodromy matrix
\begin{align}
\widehat{\mathcal{U}}_-(\la|\beta)
= S_0^{-1}(\eta/2-\la|\beta,\alpha+1)\ \mathcal{U}_-(\la)\ S_0 (\la-\eta/ 2|\beta,\alpha-1 ),
\end{align}
in which we have specified the explicit dependency of \eqref{mat-S} and of its inverse on the parameter $\alpha$.
Indeed, by considering these different shifts of $\alpha$, we obtain that the diagonal entries of  $\widehat{\mathcal{U}}_-(\la|\beta)$ are\footnote{We do not specify here the expression of the other entries of $\widehat{\mathcal{U}}_-(\la|\beta)$ since we shall not use them.}
\begin{equation}
 \widehat{\mathcal{U}}_-(\la|\beta)
 =\begin{pmatrix}e^\eta\,\frac{\sinh(\eta(\beta-1))}{\sinh(\eta\beta)\,}\widetilde{\mathcal{A}}_-(\la|\beta-1)&\star\\
\star&e^\eta\,\frac{\sinh(\eta(\beta+1))}{\sinh(\eta\beta)}\,\widetilde{\mathcal{D}}_-(\la|\beta+1)\end{pmatrix}.
\end{equation}
 Then the transfer matrix is given as
 \begin{equation}
      \mathcal{T}(\la)=\tr_0\left\{\widehat{\mathcal{K}}_+(\la|\beta) \ \widehat{\mathcal{U}}_-(\la|\beta)\right\},
 \end{equation}
in which
\begin{align}
 \widehat{\mathcal{K}}_+(\la|\beta)
   &=S_0^{-1}(\la- \eta / 2|\beta,\alpha-1)\ K_+(\la)\ S_0(\eta /2-\la|\beta,\alpha+1)
   \nonumber\\
   &=\begin{pmatrix}\mathsf{a}_+(\la|\beta)&\mathsf{b}_+(\la|\beta)\\
       \mathsf{c}_+(\la|\beta)&\mathsf{d}_+(\la|\beta)\end{pmatrix}.
       \label{K+gauged}
\end{align}

The parameters $\alpha$ and $\beta$ can then be chosen adequately so as to simplify the expression of $\mathcal{T}(\la)$. In particular, we can impose the gauged boundary matrix \eqref{K+gauged} to be diagonal\footnote{In fact, for the use of the SoV approach, it is enough to choose $\widehat{\mathcal{K}}_+(\lambda|\beta)$ to be lower triangular, i.e. to impose $\mathsf{b}_+(\la|\beta)=0$ (see \cite{FalKN14}). The choice \eqref{cond-2}-\eqref{cond-2bis} is here just for convenience.}, i.e. choose $\mathsf{b}_+(\la|\beta)=\mathsf{c}_+(\la|\beta)=0$.
The condition $\mathsf{b}_+(\la|\beta)=0$ is equivalent to
\begin{equation}\label{cond-1}
   2\kappa_+\sinh\big(\eta(\beta-\alpha)-\tau_+\big)-e^{-\varsigma_+}=0,
\end{equation}
which can alternatively be rewritten in terms of the boundary parameters  $\alpha_+$ and $\beta_+$ as
\begin{equation}\label{cond-1bis}
    \sinh\big(\eta(\beta-\alpha)-\tau_+\big)=\sinh(\beta_+-\alpha_+),
\end{equation}
or equivalently as
\begin{equation}\label{cond-1ter}
   \exists \, \epsilon_+\in\{1,-1\},\qquad
   \eta (\beta-\alpha)=\tau_++\epsilon_+(\alpha_+-\beta_+)+\frac{1+\epsilon_+}{2} i\pi \mod 2i\pi.
\end{equation}
The condition $\mathsf{c}_+(\la|\beta)=0$ is equivalent to
\begin{equation}\label{cond-2}
   2\kappa_+\sinh\big(\eta(\beta+\alpha)+\tau_+\big)+e^{-\varsigma_+}=0,
\end{equation}
which can alternatively be rewritten in terms of the boundary parameters  $\alpha_+$ and $\beta_+$ as
\begin{equation}\label{cond-2bis}
    \sinh\big(\eta(\beta+\alpha)+\tau_+\big)=\sinh(\alpha_+-\beta_+),
\end{equation}
or equivalently as
\begin{equation}\label{cond-2ter}
   \exists \, \epsilon'_+\in\{1,-1\},\qquad
   \eta(\beta+\alpha)=-\tau_++\epsilon'_+(\alpha_+-\beta_+)+\frac{1-\epsilon'_+}{2} i\pi \mod 2i\pi.
\end{equation}
A choice for the gauged parameters $\alpha$ and $\beta$ compatible with these two conditions is therefore given by
\begin{align}
 &\eta\alpha= -\tau_+ +\frac{\epsilon'_+-\epsilon_+}{2}(\alpha_+-\beta_+)-\frac{\epsilon_++\epsilon'_+}{4}i\pi \mod i\pi , \\
 &\eta\beta=\frac{\epsilon_++\epsilon'_+}{2}(\alpha_+-\beta_+)+\frac{2+\epsilon_+-\epsilon'_+}{4}i\pi \mod i\pi,
\end{align}
for $\epsilon_+,\epsilon'_+\in\{1,-1\}$.

If these two conditions are fulfilled (which we shall suppose from now on), the transfer matrix is simply given by
\begin{multline}\label{T-gauged-U}
 \mathcal{T}(\la)= \frac{e^{\eta}}{\sinh(\eta\beta)}
     \Big\{ \mathsf{a}_+(\la|\beta) \, \sinh(\eta(\beta-1))\,\widetilde{\mathcal{A}}_-(\la|\beta-1)\\
    +\mathsf{d}_+(\la|\beta) \, \sinh(\eta(\beta+1))\,\widetilde{\mathcal{D}}_-(\la|\beta+1)\Big\},
\end{multline}
where
 \begin{align}
 \mathsf{a}_+(\la|\beta)
 &= \mathsf{d}_+(\la|-\beta)
     \nonumber\\
 &=\frac{e^{-\lambda-\frac{\eta}{2}}}{2\sinh(\eta \beta)\, \sinh\varsigma_+}
   \Big\{  e^{\varsigma_+}\sinh(\eta\beta)-e^{-\varsigma_+}\sinh(2\lambda+\eta+\eta\beta)
    \nonumber \\
 &\hspace{4.5cm} -2\kappa_+ \sinh(\eta\alpha+\tau_+)\,\sinh(2\lambda+\eta)\Big\}.
\end{align}
Note that, by using \eqref{cond-1ter} for a given choice of $\epsilon_+$, theses coefficients $\mathsf{a}_+(\la|\beta)$ and $\mathsf{d}_+(\la|\beta)$ can be rewritten 
as\footnote{These rewritings hold even if \eqref{cond-2bis} is not satisfied.}
\begin{align}
   &\mathsf{a}_+(\la|\beta)\equiv \mathsf{a}_+(\la)
       =\epsilon_+\,e^{-\lambda-\frac{\eta}{2}}\,\frac{\sinh(\lambda+\frac{\eta}{2}+\epsilon_+\alpha_+)\,\cosh(\lambda+\frac{\eta}{2}-\epsilon_+\beta_+)}{\sinh\alpha_+\,\cosh\beta_+},
       \label{a+}\\
   &\mathsf{d}_+(\la|\beta)\equiv \mathsf{d}_+(\la)
       =-\epsilon_+\,e^{-\lambda-\frac{\eta}{2}}\,
       \frac{\sinh(\lambda+\frac{\eta}{2}-\epsilon_+\alpha_+)\,\cosh(\lambda+\frac{\eta}{2}+\epsilon_+\beta_+)}{\sinh\alpha_+\,\cosh\beta_+} .
       \label{d+}
\end{align}
It may be useful to express the transfer matrix in the following form:
\begin{gather}
    \mathcal{T}(\la)=\mathcal{F}(\la|\beta)+\mathcal{F}(\la|-\beta),\\
\text{with} \quad
    \mathcal{F}(\la|\beta)= e^{\eta}\,\frac{\sinh(\eta(\beta-1))}{\sinh(\eta\beta)}\,\mathsf{a}_+(\la|\beta)\,\widetilde{\mathcal{A}}_-(\la|\beta-1).
\end{gather}

One can alternatively express the transfer matrix \eqref{T-gauged-U} in terms of the elements of SOS boundary monodromy matrix \eqref{U-SOS}, i.e. as
\begin{equation} \label{gauge-transfer}
  \mathcal{T}(\lambda)= S_{1\ldots N}(\{\xi\}|\beta)\ \mathcal{T}^\mathrm{SOS}(\lambda|\beta)\ S_{1\ldots N}^{-1}(\{\xi\}|\beta),
\end{equation}
in terms of the following SOS transfer matrix:
\begin{multline}\label{transfer-SOS}
   \mathcal{T}^{\mathrm{SOS}}(\la|\beta)=\frac{e^{\eta}}{\sinh(\eta\beta)}
     \Big\{ \mathsf{a}_+(\la|\beta) \, \sinh(\eta(\beta-1))\, \mathcal{A}_-^{\mathrm{SOS}}(\la|\beta-1)
     \\
     +\mathsf{d}_+(\la|\beta) \, \sinh(\eta(\beta+1))\,\mathcal{D}_-^{\mathrm{SOS}}(\la|\beta+1)\Big\}.
\end{multline}
Hence, we have reduced the problem of diagonalizing the transfer matrix \eqref{transfer} to the study of the eigenstates of the SOS transfer matrix \eqref{transfer-SOS}.

\section{Diagonalisation of the transfer matrix by SoV}
\label{sec-diag}

It is easy to see that, by construction, the transfer matrix $\mathcal{T}(\lambda)$ is a polynomial of degree $N+2$ in $\sinh^2\lambda$ (or equivalently in $\cosh(2\lambda)$), with leading asymptotic behavior 
\begin{equation}\label{T-asympt}
   \mathcal{T}(\lambda)\underset{\lambda\to\pm\infty}{\sim}
   \frac{\kappa_+\kappa_-\, \cosh(\tau_+-\tau_-)}{2^{2N+1}\, \sinh\varsigma_+\sinh\varsigma_-}\,
   e^{\pm2(N+2)\lambda}.
\end{equation}
Its value at $\frac{\eta}{2}$ (respectively at $\frac{\eta}{2}+i\frac{\pi}{2}$) can easily be computed from the fact that $\mathcal{U}_-(\frac{\eta}{2})=(-1)^N\det_qM(0)$ (respectively that $\mathcal{U}_-(\frac{\eta}{2}+i\frac{\pi}{2})=i\coth\varsigma_-\det_qM(i\frac{\pi}{2})\,\sigma^z$):
\begin{align}
   &\mathcal{T}(\eta/2)=2\,(-1)^N\cosh\eta\,{\det}_q M(0), \label{T-value1}\\
   &\mathcal{T}(\eta/2+i\pi/2)=-2\,\cosh\eta\,\coth\varsigma_+\coth\varsigma_-\,{\det}_qM(i\pi/2).\label{T-value2}
\end{align}
Finally, its spectrum and eigenstates are directly related to those of the SOS transfer matrix $\mathcal{T}^\mathrm{SOS}(\lambda)$ through the gauge transformation \eqref{gauge-transfer}. The latter being expressed in a simple form \eqref{transfer-SOS} in terms of $\mathcal{A}_-^{\mathrm{SOS}}(\la|\beta-1)$ and $\mathcal{D}_-^{\mathrm{SOS}}(\la|\beta+1)$, we can construct \cite{Nic12,FalKN14} a basis of the space of states (the SoV basis) which separates the variables for the spectral problem associated to $\mathcal{T}^\mathrm{SOS}(\lambda)$ at particular values (related to the inhomogeneity parameters of the model) of the spectral parameter $\lambda$.

\subsection{SoV basis of the space of states}

The construction of the basis which separates the spectral problem for $\mathcal{T}^\mathrm{SOS}(\lambda)$ relies, as usual, on the use of the (shifted) inhomogeneity parameters of the model. The latter have to be generic or, at least, to satisfy the following non-intersecting conditions:
\begin{equation}\label{cond-inh}
    \xi_j^{(h_j)}\pm\xi_k^{(h_k)}\notin i\pi\mathbb{Z},
    \quad \forall j,k\in\{1,\ldots,N\} \text{ with } j\not=k,\ \forall h_j,h_k\in\{0,1\}.
\end{equation}
Here we have used the notation:
\begin{equation}\label{xi-shift}
   \xi_n^{(h)}=\xi_n+\eta/2-h\eta, \qquad 1\le n\le N, \quad h\in\{0,1\}.
\end{equation}

Let $\bra{0}$ be the dual reference state  with all spins up and $\ket{\underline{0}}$ be the reference state with all spins down.
For each $N$-tuple $\mathbf{h}\equiv(h_1,\ldots,h_N)\in\{0,1\}^N$, we define the following states:
 \begin{align}
  &\ket{\mathbf{h},\beta+1}
   =\pl_{j=1}^N 
   \left(\frac{ \mathcal{D}_-^{\mathrm{SOS}}( \xi_j+\eta/2|\beta+1)}{k_j\,\mathsf{A}_-(\eta/2-\xi_j)}\right)^{h_j}\ket{\underline{0}},
   \label{SOVstate-R}
 \\
   &\bra{\beta-1,\mathbf{h}}
   =\bra{0}\pl_{j=1}^N 
   \left(\frac{ \mathcal{A}_-^{\mathrm{SOS}}(\eta/2-\xi_j|\beta-1)}{\mathsf{A}_-(\eta/2-\xi_j)}\right)^{1-h_j}.
    \label{SOVstate-L}
  \end{align}  
In \eqref{SOVstate-R}-\eqref{SOVstate-L}, the normalization coefficients are chosen in the form
\begin{equation}\label{A-}
    k_j=\frac{\sinh(2\xi_j+\eta)}{\sinh(2\xi_j-\eta)},
    \qquad
   \mathsf{A}_-(\lambda)=
   \mathsf{g}_-(\lambda)\,
   a(\lambda)\, d(-\lambda),
\end{equation}
in terms of a function $\mathsf{g}_-(\lambda)$ satisfying the relation
\begin{equation}\label{g-}
   \mathsf{g}_-(\lambda+\eta/2)\, \mathsf{g}_-(-\lambda+\eta/2)=\frac{\det_q K_-(\lambda)}{\sinh(2\lambda-2\eta)}.
\end{equation}
%

It is easy to see that these states  are right and left pseudo-eigenstates of the operator  $\mathcal{B}_-^{\mathrm{SOS}}(\la|\beta)$, i.e.
\begin{align}
   &\mathcal{B}_-^{\mathrm{SOS}}(\la|\beta-1) \,\ket{\mathbf{h},\beta-1}
   = (-1)^N a_{\mathbf{h}}(\la)\,a_{\mathbf{h}}(-\la)\,
         \nonumber\\
   &\hspace{3.5cm}\times\mathsf{b}_-(\la|\beta-N-1)\,
   \frac{\sinh(\eta(\beta-N-1))}{\sinh(\eta(\beta-1))} \,
   \ket{\mathbf{h},\beta+1},
   \label{act-BR}
   \\
   &\bra{\beta+1,\mathbf{h}}\, \mathcal{B}_-^{\mathrm{SOS}}(\la|\beta+1)
    = (-1)^N a_{\mathbf{h}}(\la)\,a_{\mathbf{h}}(-\la) \,
      \nonumber\\
   &\hspace{3.5cm}\times\mathsf{b}_-(\la|\beta+N+1)\,
       \frac{\sinh(\eta\beta)}{\sinh(\eta (\beta+N))}  \, \bra{\beta-1,\mathbf{h}},
       \label{act-BL}
\end{align}
where
\begin{equation}\label{a-h}
    a_\mathbf{h}(\lambda)=\prod_{n=1}^N\sinh(\lambda-\xi_n-\eta/2+h_n\eta).
\end{equation}
As usual, one can determine the action on \eqref{SOVstate-L} and \eqref{SOVstate-R} of the other operators of the SOS boundary algebra by polynomial interpolation (see Appendix~\ref{app-act-SOS}).

From these actions, we can easily derive the orthogonality property:
\begin{equation}\label{orthog-states}
  \moy{\beta-1,\mathbf{h}\, |\, \mathbf{k},\beta+1}=\delta_{\mathbf{h},\mathbf{k}}\,
  N(\{\xi\},\beta)\,
  \frac{e^{2\sum_{j=1}^Nh_j\xi_j}}{\widehat{V}(\xi_1^{(h_1)},\ldots,\xi_N^{(h_N)})}.
\end{equation}
Here and in the following, we define, for any $N$-tuple of variables $(x_1,\ldots,x_N)$, the quantity $\widehat{V}(x_1,\ldots,x_N)$ as
\begin{equation}\label{VDM}
   \widehat{V}(x_1,\ldots,x_N)=\det_{1\le i,j\le N}\left[\sinh^{2(j-1)}x_i\right]=\prod_{j<k}(\sinh^2x_k-\sinh^2 x_j).
\end{equation}
The normalization constant in \eqref{orthog-states}, as computed in Appendix \ref{Normalization-Factor}, is given by the following expression:
\begin{align}
   N(\{\xi\},\beta)
   &=\widehat{V}(\xi_1^{(0)},\ldots,\xi_N^{(0)})\, \bra{0}\pl_{j=1}^N 
   \frac{ \mathcal{A}_-^{\mathrm{SOS}}(\eta/2-\xi_j|\beta-1)}{\mathsf{A}_-(\eta/2-\xi_j)}\,
   \ket{\underline{0}}
   \nonumber\\
   &= (-1)^N\, \widehat{V}(\xi_1,\ldots,\xi_N)\,
    \frac{ \widehat{V}(\xi_1^{(0)},\ldots,\xi_N^{(0)}) }{ \widehat{V}(\xi_1^{(1)},\ldots,\xi_N^{(1)}) } 
   \nonumber\\
   &\qquad\times
   \prod_{j=1}^N\left[ \frac{ \mathsf{b}_-(\frac{\eta}{2}-\xi_j|\beta+1+N-2j)}{\mathsf{g}_-(\frac{\eta}{2}-\xi_j)}\,\frac{\sinh(\eta(\beta+1+N-2j))}{\sinh(\eta(\beta+N-j))}\right].
   \label{norm-states}
\end{align}
Note that, for $\beta\notin\mathbb{Z}$ and generic inhomogeneity parameters $\xi_j$, the condition for this normalization constant \eqref{norm-states} to be non-zero is
\begin{equation}\label{cond-3}
   \forall j\in\{1,\ldots , N\},\quad\sinh(\eta(\beta+1-\alpha+N-2j)-\tau_-)\not= \sinh(\alpha_-+\beta_-),
\end{equation}
which, taking into account the condition \eqref{cond-1ter} for a given choice of $\epsilon_+$, is equivalent to the following condition on the boundary parameters:
\begin{multline}\label{cond-3bis}
   \forall j\in\{1,\ldots N\},\ \forall\epsilon\in\{1,-1\},\\
   \tau_+-\tau_-+\eta(N-2j+1)
   \not= \epsilon(\alpha_-+\beta_-)-\epsilon_+(\alpha_+-\beta_+)-\frac{\epsilon_++\epsilon}{2}i\pi \mod 2i\pi.
\end{multline}
Hence, from now on, we shall suppose that the condition \eqref{cond-3bis} is satisfied, which means that, for generic values of the inhomogeneity parameters \eqref{cond-inh}, the states \eqref{SOVstate-R} (respectively \eqref{SOVstate-L}) form a basis of $\mathcal{H}$ (respectively of $\mathcal{H}^\star$).
In that case, we have the following resolution of the identity:
\begin{equation}\label{decomp-id}
   \mathbf{1}=\frac{1}{N(\{\xi\},\beta)}\sum_{\mathbf{h}\in\{0,1\}^N} \!\!\!\!
   e^{-2\sum_{j=1}^Nh_j\xi_j} \, \widehat{V}(\xi_1^{(h_1)},\ldots,\xi_N^{(h_N)})\
   \ket{\mathbf{h},\beta+1}\bra{\beta-1,\mathbf{h}}.
\end{equation}

\subsection{The transfer matrix spectrum and eigenstates}

From the actions \eqref{act-AL}, \eqref{act-DR} and the parity properties \eqref{parity-A}, \eqref{parity-D}, 
it is easy to compute the action of $\mathcal{T}^\mathrm{SOS}(\xi_n^{(h_n)})=\mathcal{T}^\mathrm{SOS}(-\xi_n^{(h_n)})$ on the states $\ket{\mathbf{h},\beta+1}$ and $\bra{\beta-1,\mathbf{h}}$.
Hence we see that the basis \eqref{SOVstate-R} of $\mathcal{H}$ (respectively the basis \eqref{SOVstate-L} of $\mathcal{H}^\star$) separates the variables for the spectral problem associated to $\mathcal{T}^\mathrm{SOS}(\lambda)$ at these points $\pm\xi_n^{(h_n)}$, $n\in\{1,\ldots,N\}$.
This fact, together with the aforementioned  algebraic properties of the transfer matrix, leads to the following characterization for the spectrum and eigenstates of $\mathcal{T}(\lambda)$:

\begin{theorem}\label{th-SoV-spectrum}
Let us suppose that the inhomogeneity parameters are generic \eqref{cond-inh} and that the condition \eqref{cond-3} is  satisfied.
Then the spectrum $\Sigma_\mathcal{T}$ of the transfer matrix $\mathcal{T}(\lambda)$ is simple and consists in the set of functions $\tau(\lambda)$ which satisfy the following properties:
\begin{enumerate}
   \item[(i)] $\tau(\lambda)$ is  a polynomial of degree $N+2$ in $\sinh^2\lambda$ (or equivalently in $\cosh(2\lambda)$),
   \item[(ii)] its leading asymptotic behavior when $\lambda\to\pm\infty$ is
\begin{equation}\label{T-asympt1}
   \tau(\lambda)\underset{\lambda\to\pm\infty}{\sim}
   \frac{\kappa_+\kappa_-\, \cosh(\tau_+-\tau_-)}{2^{2N+1}\, \sinh\varsigma_+\sinh\varsigma_-}\,
   e^{\pm2(N+2)\lambda} ,
\end{equation}
   \item[(iii)] its values at $\eta/2$ and at $\eta/2+i\pi/2$ are respectively given by
\begin{align}
   &\tau(\eta/2)=2\,(-1)^N\cosh\eta\,{\det}_q M(0), \label{t-value1}\\
   &\tau(\eta/2+i\pi/2)=-2\,\cosh\eta\,\coth\varsigma_+\coth\varsigma_-\,{\det}_qM(i\pi/2), \label{t-value2}
\end{align}
\item[(iv)] it satisfies the conditions
\begin{equation}\label{cond-SoV-spectrum}
   \tau(\xi_n+\eta/2)\,\tau(\xi_n-\eta/2)
   =-\frac{\det_q K_+(\xi_n)\, \det_q \mathcal{U}_-(\xi_n)}{\sinh(2\xi_n+\eta)\, \sinh(2\xi_n-\eta)},
   \quad
   \forall n\in\{1,\ldots,N\}.
\end{equation}
\end{enumerate}
The one-dimensional right and left $\mathcal{T}(\lambda)$-eigenstates associated with the eigenvalue $\tau(\lambda)$ are respectively generated by the following vectors
\begin{align}
   &\ket{\Psi_t}
   =\!\! \sum_{\mathbf{h}\in\{0,1\}^N}\prod_{n=1}^N\! Q_{\tau}(\xi_n^{(h_n)})\ e^{-\sum_jh_j\xi_j}\,\widehat{V}(\xi_1^{(h_1)},\ldots,\xi_N^{(h_N)})\ S_{1\ldots N}(\{\xi\}|\beta) \,\ket{\mathbf{h},\beta+1},
   \label{eigen-R}\\
   &\bra{\Psi_t}=\!\!\sum_{\mathbf{h}\in\{0,1\}^N}\prod_{n=1}^N
   \left[ \left(\frac{\sinh(2\xi_n-2\eta)}{\sinh(2\xi_n+2\eta)}\, 
   \frac{\mathbf{A}(\xi_n+\frac{\eta}{2})}
         {\mathbf{A}(-\xi_n+\frac{\eta}{2})}
   \right)^{\! h_n} Q_{\tau}(\xi_n^{(h_n)})\right]
   \nonumber\\
   &\hspace{3cm}\times
   e^{-\sum_jh_j\xi_j}\,\widehat{V}(\xi_1^{(h_1)},\ldots,\xi_N^{(h_N)})\ 
   \bra{\beta-1,\mathbf{h}}\,  S_{1\ldots N}^{-1}(\{\xi\}|\beta),
   \label{eigen-L}
\end{align}
expressed on the gauged transformed basis \eqref{SOVstate-R} (respectively \eqref{SOVstate-L}).
In these expressions, $Q_{\tau}$ denotes a function on the discrete set of values $\xi_n^{(h_n)}$, $n\in\{1,\ldots,N\}$, $h_n\in\{0,1\}$, which satisfies
\begin{equation}\label{eq-Q-dis}
    \frac{Q_{\tau}(\xi_n^{(1)})}{Q_{\tau}(\xi_n^{(0)})}
    =\frac{t(\xi_n^{(0)})}{\mathbf{A}(\xi_n^{(0)})}
    =\frac{\mathbf{A}(-\xi_n^{(1)})}{t(\xi_n^{(1)})},
\end{equation}
and $\mathbf{A}(\lambda)$ is defined in terms of \eqref{d+} and \eqref{A-} as
\begin{equation}\label{def-A}
  \mathbf{A}(\lambda)
  =e^{-\lambda+\frac{\eta}{2}}\,
       \frac{\sinh(2\lambda+\eta)}{\sinh(2\lambda)}\,\mathsf{d}_+(-\lambda)\,\mathsf{A}_-(\lambda).
\end{equation}
\end{theorem}

Note that the explicit expression of the function $\mathbf{A}(\lambda)$ \eqref{def-A} appearing in \eqref{eigen-R}-\eqref{eq-Q-dis} depends on the particular choices
that we make on one hand for $\epsilon_+\in\{1,-1\}$ in fixing the gauge parameters $\beta-\alpha$ in \eqref{cond-1ter}, and on the other hand for the function $\mathsf{g}_-(\lambda)$ involved in the normalization of the states \eqref{SOVstate-R}-\eqref{SOVstate-L}. In particular, since there is a large freedom in fixing $\mathsf{g}_-(\lambda)$ (it has only to satisfy the relation \eqref{g-}), the resulting function $\mathbf{A}(\lambda)$ may be any function satisfying the relation
\begin{equation}\label{eq-A}
   \mathbf{A}(\lambda+\eta/2)\, \mathbf{A}(-\lambda+\eta/2)
   =-\frac{\det_q K_+(\lambda)\, \det_q \mathcal{U}_-(\lambda)}{\sinh(2\lambda+\eta)\, \sinh(2\lambda-\eta)}.
\end{equation}
In the following, we shall focus on the particular solutions of \eqref{eq-A} which are given by the expressions
\begin{equation}
  \mathbf{A}_{\boldsymbol{\varepsilon}}(\lambda)
  = (-1)^N\, \frac{\sinh(2\lambda+\eta)}{\sinh(2\lambda)}\, \mathbf{a}_{\boldsymbol{\varepsilon}}(\lambda)\, a(\lambda)\, d(-\lambda),
   \label{expr-A}
\end{equation}
with
\begin{multline}\label{a-eps}
  \mathbf{a}_{\boldsymbol{\varepsilon}}(\lambda)
  = \frac{\sinh(\lambda-\frac{\eta}{2}+\epsilon_{\alpha_+}\alpha_+)\,\cosh(\lambda-\frac{\eta}{2}-\epsilon_{\beta_+}\beta_+)}{\sinh(\epsilon_{\alpha_+}\alpha_+)\,\cosh(\epsilon_{\beta_+}\beta_+)}
  \\
  \times
   \frac{\sinh(\lambda-\frac{\eta}{2}+\epsilon_{\alpha_-}\alpha_-)\,\cosh(\lambda-\frac{\eta}{2}+\epsilon_{\beta_-}\beta_-)}{\sinh(\epsilon_{\alpha_-}\alpha_-)\,\cosh(\epsilon_{\beta_-}\beta_-)}
\end{multline}
for any choice of $\boldsymbol{\varepsilon}\equiv(\epsilon_{\alpha_+}, \epsilon_{\alpha_-}, \epsilon_{\beta_+}, \epsilon_{\beta_-})\in\{-1,1\}^4$ such that $\epsilon_{\alpha_+} \epsilon_{\alpha_-}\epsilon_{\beta_+} \epsilon_{\beta_-}=1$.
Such solutions correspond, for a fixed choice of $\epsilon_+$ in \eqref{cond-1ter}, to a choice of $\mathsf{g}_-(\lambda)$ in \eqref{A-}-\eqref{g-} such that
\begin{multline}\label{g-eps}
   \mathsf{g}_-(\lambda+\eta/2)=\epsilon_+\,\epsilon_{\alpha_+}\,(-1)^N\,
   \frac{\sinh(\lambda+\epsilon_{\alpha_-}\alpha_-)\,\cosh(\lambda+\epsilon_{\beta_-}\beta_-)}{\sinh(\epsilon_{\alpha_-}\alpha_-)\,\cosh(\epsilon_{\beta_-}\beta_-)}
   \\
   \times
   \frac{\sinh(\lambda+\epsilon_{\alpha_+}\alpha_+)\,\cosh(\lambda-\epsilon_{\beta_+}\beta_+)}{\sinh(\lambda+\epsilon_{+}\alpha_+)\,\cosh(\lambda-\epsilon_{+}\beta_+)}.
\end{multline}
From now on, we shall also denote by
\begin{equation}\label{states-eps}
   \ket{\mathbf{h},\beta+1}_{\boldsymbol{\varepsilon}} \quad \text{and}\quad {}_{\boldsymbol{\varepsilon}}\bra{\beta-1,\mathbf{h}}
\end{equation}
the states \eqref{SOVstate-R} and \eqref{SOVstate-L} with normalization \eqref{g-eps} given by such a particular choice of $\boldsymbol{\varepsilon}$.

\subsection{Transfer matrix spectrum by $T$-$Q$ functional equation}

We now recall the results of \cite{KitMN14} concerning the rewriting of the above SoV discrete characterization of the transfer matrix spectrum and eigenstates in terms of particular classes of solutions of a functional equation of Baxter's type.

We denote by $\Sigma_{Q}^M$ the set of $Q(\lambda)$ polynomials in $\cosh(2\lambda)$ of degree $M$ of the form
\begin{equation}\label{Q-form1}
    Q(\lambda)=\prod_{j=1}^M\frac{\cosh(2\lambda)-\cosh(2\lambda_j)}{2}
                      =\prod_{j=1}^M\big(\sinh^2\lambda-\sinh^2\lambda_j\big) ,
\end{equation}
with
\begin{equation}
   \cosh(2\lambda_j)\not=\cosh(2\xi_n^{(h)}),
    \quad
    \forall\, (j,n,h)\in\{1,\ldots,M\}\times\{1,\ldots,N\}\times\{0,1\}.
\end{equation}
Moreover, we consider the following function:
\begin{align}
   \mathbf{F}_{\boldsymbol{\varepsilon}}(\lambda)
   &=  \mathfrak{f}^{(N)}_{\boldsymbol{\varepsilon}}\, a(\lambda)\, a(-\lambda)\, d(\lambda)\, d(-\lambda)\, \big(\cosh^2(2\lambda)-\cosh^2\eta\big)\nonumber\\
   &= \mathfrak{f}^{(N)}_{\boldsymbol{\varepsilon}}\, \big(\cosh^2(2\lambda)-\cosh^2\eta\big)
         \prod_{n=1}^N\prod_{h=0}^1\frac{\cosh(2\lambda)-\cosh(2\xi_n^{(h)})}{2},
\end{align}
where
\begin{multline}
  \mathfrak{f}^{(r)}_{\boldsymbol{\varepsilon}}
  \equiv \mathfrak{f}^{(r)}_{\boldsymbol{\varepsilon}}(\tau_+,\tau_-,\alpha_+,\alpha_-,\beta_+,\beta_-)
  =
  \frac{2\kappa_+\kappa_-}  {\sinh\varsigma_+\,\sinh\varsigma_-}\,
  \big[\cosh(\tau_+-\tau_-)\\
  -\epsilon_{\alpha_+}\epsilon_{\alpha_-}\cosh(\epsilon_{\alpha_+}\alpha_++\epsilon_{\alpha_-}\alpha_--\epsilon_{\beta_+}\beta_++\epsilon_{\beta_-}\beta_-+(N-1-2r)\eta)\big].
\end{multline}

\begin{theorem}\label{th-inhom}
Let us suppose that the inhomogeneity parameters are generic \eqref{cond-inh} and that the condition \eqref{cond-3bis} is satisfied. 
Suppose moreover that, for a given choice of $\boldsymbol{\varepsilon}\equiv( \epsilon_{\alpha_+}, \epsilon_{\alpha_-}, \epsilon_{\beta_+}, \epsilon_{\beta_-})\in\{-1,1\}^4$ such that $\epsilon_{\alpha_+} \epsilon_{\alpha_-}\epsilon_{\beta_+} \epsilon_{\beta_-}=1$,
\begin{equation}\label{cond-inhom-N}
  \forall r\in\{0,...,N-1\}, \qquad \mathfrak{f}^{(r)}_{\boldsymbol{\varepsilon}}(\tau_+,\tau_-,\alpha_+,\alpha_-,\beta_+,\beta_-)\not=0.
\end{equation}
Then, a function $\tau(\lambda)$ is an eigenvalue of the transfer matrix $\mathcal{T}(\lambda)$ (i.e. $\tau(\lambda)\in\Sigma_\mathcal{T}$) if and only if it is an entire function of $\lambda$ such that there exists a unique $Q(\lambda)\in\Sigma_{Q}^N$ satisfying
\begin{equation}\label{inhom-TQ}
   \tau(\lambda)\, Q(\lambda)=\mathbf{A}_{\boldsymbol{\varepsilon}}(\lambda)\, Q(\lambda-\eta)
   +\mathbf{A}_{\boldsymbol{\varepsilon}}(-\lambda)\, Q(\lambda+\eta)
   +\mathbf{F}_{\boldsymbol{\varepsilon}}(\lambda).
\end{equation}
\end{theorem}

An interesting particular case of the above theorem corresponds to the situation in which the function $\mathbf{F}_{\boldsymbol{\varepsilon}}(\lambda)$ cancels, so that we obtain a complete description of the transfer matrix spectrum in terms of solutions of a usual (homogeneous) T-Q functional equation. Denoting with $\Sigma_{\boldsymbol{\varepsilon},\mathcal{T}}^M$ the set of entire functions of $\lambda$ which can be expressed in the form 
\begin{equation}\label{ratio-Q}
   \frac{\mathbf{A}_{\boldsymbol{\varepsilon}}(\lambda)\, Q(\lambda-\eta)+\mathbf{A}_{\boldsymbol{\varepsilon}}(-\lambda)\, Q(\lambda+\eta)}{Q(\lambda)}
\end{equation}
in terms of some polynomial $Q(\lambda)\in\Sigma_{Q}^M$ for a given choice of $\boldsymbol{\varepsilon}\equiv(\epsilon_{\alpha_+}, \epsilon_{\alpha_-}, \epsilon_{\beta_+}, \epsilon_{\beta_-})\in\{-1,1\}^4$, we obtain in that case the following corollary:

\begin{corollary}
Let us suppose that the inhomogeneity parameters are generic \eqref{cond-inh} and that the conditions \eqref{cond-3bis} and \eqref{cond-inhom-N} are satisfied. Suppose moreover that
\begin{equation}\label{cond-hom-N}
   \mathfrak{f}^{(N)}_{\boldsymbol{\varepsilon}}(\tau_+,\tau_-,\alpha_+,\alpha_-,\beta_+,\beta_-)=0.
\end{equation}
Then, $\tau(\lambda)\in\Sigma_\mathcal{T}$ if and only if  it is an entire function of $\lambda$ such that there exists a unique $Q(\lambda)\in\Sigma_{Q}^N$ satisfying
%
\begin{equation}\label{hom-TQ}
   \tau(\lambda)\, Q(\lambda)=\mathbf{A}_{\boldsymbol{\varepsilon}}(\lambda)\, Q(\lambda-\eta)
   +\mathbf{A}_{\boldsymbol{\varepsilon}}(-\lambda)\, Q(\lambda+\eta).
\end{equation}
In other words, $\Sigma_{\boldsymbol{\varepsilon},\mathcal{T}}^N=\Sigma_\mathcal{T}$.
\end{corollary}

There are other situations in which the spectrum of the transfer matrix is partially given by solutions of the form \eqref{Q-form1} of a homogeneous T-Q equation, as described by the following theorem:

\begin{theorem}\label{th-hom-partiel}
Let us suppose that the inhomogeneity parameters are generic \eqref{cond-inh} and that the condition \eqref{cond-3bis} is satisfied. 
Suppose moreover that, for a given choice of $\boldsymbol{\varepsilon}\equiv( \epsilon_{\alpha_+}, \epsilon_{\alpha_-}, \epsilon_{\beta_+}, \epsilon_{\beta_-})\in\{-1,1\}^4$ such that $\epsilon_{\alpha_+} \epsilon_{\alpha_-}\epsilon_{\beta_+} \epsilon_{\beta_-}=1$, there exists $M\in\{0,\ldots,N-1\}$ such that
\begin{equation}\label{cond-hom-M}
   \mathfrak{f}^{(M)}_{\boldsymbol{\varepsilon}}(\tau_+,\tau_-,\alpha_+,\alpha_-,\beta_+,\beta_-)=0.
\end{equation}
Then  $\Sigma_{\boldsymbol{\varepsilon},\mathcal{T}}^M\subset\Sigma_\mathcal{T}$ and, for any $\tau(\lambda)\in\Sigma_{\boldsymbol{\varepsilon},\mathcal{T}}^M$, there exists one and only one $Q(\lambda)\in\Sigma_{\boldsymbol{\varepsilon},Q}^M$ such that $\tau(\lambda)$ and $Q(\lambda)$ satisfy the functional equation \eqref{hom-TQ}, whereas for each $\tau(\lambda)\in\Sigma_\mathcal{T}\setminus\Sigma_{\boldsymbol{\varepsilon},\mathcal{T}}^M$, there exists one and only one $Q(\lambda)\in\Sigma_{\boldsymbol{\varepsilon},Q}^N$ such that $\tau(\lambda)$ and $Q(\lambda)$ satisfy the functional equation \eqref{inhom-TQ}.
\end{theorem}

\subsection{Separate states and eigenstates}

For any polynomial $Q(\lambda)$ of the form \eqref{Q-form1} and for a given choice of $\boldsymbol{\varepsilon}$, let us consider the states
\begin{multline}\label{separate-R}
 \ket{Q}_{\boldsymbol{\varepsilon}}=\frac 1{N(\{\xi\},\beta)}\sum_{\mathbf{h}\in\{0,1\}^N}\prod_{n=1}^N\! Q(\xi_n^{(h_n)})\ e^{-\sum_jh_j\xi_j}\\ 
 \times\widehat{V}(\xi_1^{(h_1)},\ldots,\xi_N^{(h_N)})\ S_{1\ldots N}(\{\xi\}|\beta) \,\ket{\mathbf{h},\beta+1}_{\boldsymbol{\varepsilon}},
\end{multline}
and
\begin{multline}
{}_{\boldsymbol{\varepsilon}}\bra{Q}=\frac 1{N(\{\xi\},\beta)}\sum_{\mathbf{h}\in\{0,1\}^N}\prod_{n=1}^N
   \left[ ( \mathsf{u}_n\, \mathsf{v}_{n,\boldsymbol{\varepsilon}} )^{ h_n} \, Q(\xi_n^{(h_n)})\right] \,
   e^{-\sum_jh_j\xi_j}  \\
   \times\;
   \widehat{V}(\xi_1^{(h_1)},\ldots,\xi_N^{(h_N)})
   {}_{\boldsymbol{\varepsilon}}\bra{\beta-1,\mathbf{h}}\,  S_{1\ldots N}^{-1}(\{\xi\}|\beta),
   \label{separate-L}
\end{multline}
Here we have defined
\begin{align}
   \mathsf{u}_n 
   &= \frac{\sinh(2\xi_n-\eta)}{\sinh(2\xi_n+\eta)}\frac{a(\xi_n+\eta/2)\, d(-\xi_n-\eta/2)}{a(-\xi_n+\eta/2)\, d(\xi_n-\eta/2)} \nonumber\\
   &=-\prod_{j\not= n}\frac{\sinh(\xi_n-\xi_j+\eta)\, \sinh(\xi_n+\xi_j+\eta)}{\sinh(\xi_n+\xi_j-\eta)\,\sinh(\xi_n-\xi_j-\eta)},
   \label{fn}
\end{align}
and
\begin{equation}
  \mathsf{v}_{n,\boldsymbol{\varepsilon}}
  = \frac{\mathbf{a}_{\boldsymbol{\varepsilon}}(\xi_n+\frac{\eta}{2})}{\mathbf{a}_{\boldsymbol{\varepsilon}}(-\xi_n+\frac{\eta}{2})}
  = \frac{\mathbf{a}_{\boldsymbol{\varepsilon}}(\xi_n+\frac{\eta}{2})}{\mathbf{a}_{-\boldsymbol{\varepsilon}}(\xi_n+\frac{\eta}{2})},
\end{equation}
so that 
\begin{equation}
   \frac{\sinh(2\xi_n-2\eta)}{\sinh(2\xi_n+2\eta)}\, 
   \frac{\mathbf{A}_{\boldsymbol{\varepsilon}}(\xi_n+\frac{\eta}{2})}
         {\mathbf{A}_{\boldsymbol{\varepsilon}}(-\xi_n+\frac{\eta}{2})}
   = \mathsf{u}_n\, \mathsf{v}_{n,\boldsymbol{\varepsilon}}.
\end{equation}
States of the form \eqref{separate-R} and \eqref{separate-L} are called {\em separate states}. If moreover $Q(\lambda)$ satisfies the T-Q equation \eqref{hom-TQ} or \eqref{inhom-TQ} with some entire function $\tau(\lambda)$, they are eigenstates of the transfer matrix $\mathcal{T}(\lambda)$.

As in the XXX case \cite{KitMNT17}, it is simple to prove the following identity
\begin{equation}
   \prod_{n=1}^N(-\mathsf{u}_n)^{h_n}\ \widehat{V}(\xi_1^{(h_1)},\ldots,\xi_N^{(h_N)})
   = \frac{\widehat{V}(\xi_1^{(0)},\ldots,\xi_N^{(0)})}{\widehat{V}(\xi_1^{(1)},\ldots,\xi_N^{(1)})}\;
   \widehat{V}(\xi_1^{(1-h_1)},\ldots,\xi_N^{(1-h_N)}),
\end{equation}
so that the state \eqref{separate-L} can equivalently be rewritten as
\begin{multline}
{}_{\boldsymbol{\varepsilon}}\bra{Q}
   =\frac 1{N(\{\xi\},\beta)}\,\,\frac{\widehat{V}(\xi_1^{(0)},\ldots,\xi_N^{(0)})}{\widehat{V}(\xi_1^{(1)},\ldots,\xi_N^{(1)})}
   \sum_{\mathbf{h}\in\{0,1\}^N}\prod_{n=1}^N
   \left[ ( - \mathsf{v}_{n,\boldsymbol{\varepsilon}} )^{ h_n} \, Q(\xi_n^{(h_n)})\right] \\
   \times
   e^{-\sum_jh_j\xi_j}\,\widehat{V}(\xi_1^{(1-h_1)},\ldots,\xi_N^{(1-h_N)})
   \
   {}_{\boldsymbol{\varepsilon}}\bra{\beta-1,\mathbf{h}}\,  S_{1\ldots N}^{-1}(\{\xi\}|\beta).
   \label{separate-L-bis}
\end{multline}

As usual, it is possible to rewrite the separate states in a Bethe-type form, using that
\begin{equation}
  \prod_{n=1}^N Q(\xi_n^{(h_n)})= \prod_{j=1}^M  a_{\mathbf{h}}(\la_j)\,a_{\mathbf{h}}(-\la_j),
\end{equation}
and the formulas \eqref{act-BR}, \eqref{act-BL} for the action of the operators $\mathcal{B}_-^\mathrm{SOS}(\lambda|\beta-1)$ and $\mathcal{B}_-^\mathrm{SOS}(\lambda|\beta+1)$ on the states \eqref{SOVstate-R} and \eqref{SOVstate-L} respectively. We obtain that
\begin{multline}\label{Bethe-state-R}
    \ket{Q}_{\boldsymbol{\varepsilon}}= S_{1\ldots N}(\{\xi\}|\beta)\,
    \widehat{\mathcal{B}}^R_-(\lambda_1|\beta-1)\, \widehat{\mathcal{B}}^R_-(\lambda_2|\beta-3)
    \ldots\\
    \ldots
    \widehat{\mathcal{B}}^R_-(\lambda_M|\beta+1-2M)\,
    \ket{\Omega_{\beta+1-2M}}_{\boldsymbol{\varepsilon}},
\end{multline}
and
\begin{multline}\label{Bethe-state-L}
  {}_{\boldsymbol{\varepsilon}}\bra{Q}
  ={}_{\boldsymbol{\varepsilon}}\bra{\Omega_{\beta-1+2M}}\,
  \widehat{\mathcal{B}}^L_-(\lambda_M|\beta-1+2M)\ldots\\
  \ldots \widehat{\mathcal{B}}^L_-(\lambda_2|\beta+2)\, 
  \widehat{\mathcal{B}}^L_-(\lambda_1|\beta+1)\,
  S_{1\ldots N}^{-1}(\{\xi\}|\beta).
\end{multline}
Here we have defined the renormalized operators $\widehat{\mathcal{B}}^R_-(\lambda |\beta)$ and $\widehat{\mathcal{B}}^L_-(\lambda |\beta)$ in terms of $\mathcal{B}_-^{\mathrm{SOS}}(\la|\beta)$ as
\begin{align}
   &\widehat{\mathcal{B}}^R_-(\lambda |\beta)
   =\frac{(-1)^N}{\mathsf{b}_-(\la|\beta-N)}\,
   \frac{\sinh(\eta\beta)}{\sinh(\eta(\beta-N))}\, \mathcal{B}_-^{\mathrm{SOS}}(\la|\beta),
   \label{B-R}\\
    &\widehat{\mathcal{B}}^L_-(\lambda |\beta)
   =\frac{(-1)^N}{\mathsf{b}_-(\la|\beta+N)}\,
   \frac{\sinh(\eta(\beta+N-1))}{\sinh(\eta(\beta-1))}\, \mathcal{B}_-^{\mathrm{SOS}}(\la|\beta),
   \label{B-L}
\end{align}
and the right and left reference states as
\begin{equation}\label{ref-state-R}
  \ket{\Omega_{\beta+1}}_{\boldsymbol{\varepsilon}}
  = \frac 1{N(\{\xi\},\beta)}\sum_{\mathbf{h}\in\{0,1\}^N} e^{-\sum_jh_j\xi_j}\,\widehat{V}(\xi_1^{(h_1)},\ldots,\xi_N^{(h_N)}) \,\ket{\mathbf{h},\beta+1}_{\boldsymbol{\varepsilon}},
\end{equation}
and
\begin{align}
  {}_{\boldsymbol{\varepsilon}}\bra{\Omega_{\beta-1}}
  &=\frac 1{N(\{\xi\},\beta)}\sum_{\mathbf{h}\in\{0,1\}^N}\prod_{n=1}^N 
  ( \mathsf{u}_n\, \mathsf{v}_{n,\boldsymbol{\varepsilon}} )^{ h_n} \,
   e^{-\sum_jh_j\xi_j}\,\widehat{V}(\xi_1^{(h_1)},\ldots,\xi_N^{(h_N)})\
   {}_{\boldsymbol{\varepsilon}}\bra{\beta-1,\mathbf{h}}
   \\
   &=\frac 1{N(\{\xi\},\beta)}\,\,\frac{\widehat{V}(\xi_1^{(0)},\ldots,\xi_N^{(0)})}{\widehat{V}(\xi_1^{(1)},\ldots,\xi_N^{(1)})}
   \sum_{\mathbf{h}\in\{0,1\}^N}\prod_{n=1}^N ( - \mathsf{v}_{n,\boldsymbol{\varepsilon}} )^{ h_n} 
   \nonumber\\
   &\hspace{3cm}\times
   e^{-\sum_jh_j\xi_j}\,\widehat{V}(\xi_1^{(1-h_1)},\ldots,\xi_N^{(1-h_N)})\
   {}_{\boldsymbol{\varepsilon}}\bra{\beta-1,\mathbf{h}}.
   \label{ref-state-L}
\end{align}

Note that we can define the separate states \eqref{separate-R} and \eqref{separate-L}, as well as the reference states \eqref{ref-state-R} and \eqref{ref-state-L} for different choices of $\boldsymbol{\varepsilon}\in\{-1,1\}^4$.
In particular, we can use the relation between the SoV basis for $\boldsymbol{\varepsilon}$ and $-\boldsymbol{\varepsilon}$,
\begin{align}
   &\ket{\mathbf{h},\beta+1}_{-\boldsymbol{\varepsilon}}
   =\prod_{n=1}^N \mathsf{v}_{n,\boldsymbol{\varepsilon}}^{-h_n}\, 
   \ket{\mathbf{h},\beta+1}_{\boldsymbol{\varepsilon}}
   \\
   &{}_{-\boldsymbol{\varepsilon}}\bra{\beta-1,\mathbf{h}}
   =\prod_{n=1}^N \mathsf{v}_{n,\boldsymbol{\varepsilon}}^{h_n-1}\
   {}_{\boldsymbol{\varepsilon}}\bra{\beta-1,\mathbf{h}},
\end{align}
to re-express the reference states  \eqref{ref-state-R} and \eqref{ref-state-L} associated with $-\boldsymbol{\varepsilon}$ in terms of the SoV basis for $\boldsymbol{\varepsilon}$ as
\begin{equation}\label{ref-state-R-}
  \ket{\Omega_{\beta+1}}_{-\boldsymbol{\varepsilon}}
  = \sum_{\mathbf{h}\in\{0,1\}^N} \prod_{n=1}^N \mathsf{v}_{n,\boldsymbol{\varepsilon}}^{-h_n}\,
  e^{-\sum_jh_j\xi_j}\,\widehat{V}(\xi_1^{(h_1)},\ldots,\xi_N^{(h_N)}) \,\ket{\mathbf{h},\beta+1}_{\boldsymbol{\varepsilon}},
\end{equation}
and
\begin{align}
  {}_{-\boldsymbol{\varepsilon}}\bra{\Omega_{\beta-1}}
  &=\prod_{n=1}^N \mathsf{v}_{n,\boldsymbol{\varepsilon}}^{-1}
  \sum_{\mathbf{h}\in\{0,1\}^N}\prod_{n=1}^N \mathsf{u}_n^{ h_n} \,
   e^{-\sum_jh_j\xi_j}\,\widehat{V}(\xi_1^{(h_1)},\ldots,\xi_N^{(h_N)})\
   {}_{\boldsymbol{\varepsilon}}\bra{\beta-1,\mathbf{h}}
   \\
   &=\frac{\widehat{V}(\xi_1^{(0)},\ldots,\xi_N^{(0)})}{\widehat{V}(\xi_1^{(1)},\ldots,\xi_N^{(1)})}\,
   \prod_{n=1}^N \mathsf{v}_{n,\boldsymbol{\varepsilon}}^{-1}
   \sum_{\mathbf{h}\in\{0,1\}^N}\prod_{n=1}^N ( - 1 )^{ h_n} 
   \nonumber\\
   &\hspace{3cm}\times
   e^{-\sum_jh_j\xi_j}\,\widehat{V}(\xi_1^{(1-h_1)},\ldots,\xi_N^{(1-h_N)})\
   {}_{\boldsymbol{\varepsilon}}\bra{\beta-1,\mathbf{h}}.
   \label{ref-state-L-}
\end{align}

\begin{proposition}
Let us suppose that the hypothesis of Theorem~\ref{th-SoV-spectrum} are satisfied.
For $\tau(\lambda)\in\Sigma_\mathcal{T}$ and for a given choice of $\boldsymbol{\varepsilon}$,
we denote by
\begin{equation}\label{Q-form}
    Q_{\tau,\boldsymbol{\varepsilon}}(\lambda)
    =\prod_{j=1}^{q_{\boldsymbol{\varepsilon}}}\frac{\cosh(2\lambda)-\cosh(2\lambda_{\boldsymbol{\varepsilon},j})}{2} 
    \qquad
    (q_{\boldsymbol{\varepsilon}}\le N)
\end{equation}
%
the unique solution of the T-Q equation \eqref{hom-TQ} if the condition
\begin{equation}
\mathfrak{f}^{(q_{\boldsymbol{\varepsilon}})}_{\boldsymbol{\varepsilon}}(\tau_+,\tau_-,\alpha_+,\alpha_-,\beta_+,\beta_-)=0
\end{equation}
%
is satisfied with $\tau(\lambda)\in\Sigma_{\boldsymbol{\varepsilon},\mathcal{T}}^{q_{\boldsymbol{\varepsilon}}}$, or the unique solution of \eqref{inhom-TQ} with $q_{\boldsymbol{\varepsilon}}=N$ otherwise.

Then,  the one-dimensional right eigenspace of the transfer matrix $\mathcal{T}(\lambda)$ associated with the eigenvalue $\tau(\lambda)$ is generated by any of the separate states $\ket{Q_{\tau,\boldsymbol{\varepsilon}}}_{\boldsymbol{\varepsilon}}$ for any choice of $\boldsymbol{\varepsilon}\equiv( \epsilon_{\alpha_+}, \epsilon_{\alpha_-}, \epsilon_{\beta_+}, \epsilon_{\beta_-})\in\{-1,1\}^4$ such that $\epsilon_{\alpha_+} \epsilon_{\alpha_-}\epsilon_{\beta_+} \epsilon_{\beta_-}=1$.
For two such choices of $\boldsymbol{\varepsilon}$, the corresponding states are proportional: 
\begin{equation}
   \ket{Q_{\tau,\boldsymbol{\varepsilon'}}}_{\boldsymbol{\varepsilon'}}
   =\prod_{n=1}^N\frac{Q_{\tau,\boldsymbol{\varepsilon'}}(\xi_n+\frac{\eta}{2})}{Q_{\tau,\boldsymbol{\varepsilon}}(\xi_n+\frac{\eta}{2})}\
   \ket{Q_{\tau,\boldsymbol{\varepsilon}}}_{\boldsymbol{\varepsilon}}
   =\frac{\prod_{j=1}^{q_{\boldsymbol{\varepsilon'}}} d(\lambda_{\boldsymbol{\varepsilon'},j})\, d(-\lambda_{\boldsymbol{\varepsilon'},j})}
            {\prod_{j=1}^{q_{\boldsymbol{\varepsilon}}} d(\lambda_{\boldsymbol{\varepsilon},j})\, d(-\lambda_{\boldsymbol{\varepsilon},j})}\
   \ket{Q_{\tau,\boldsymbol{\varepsilon}}}_{\boldsymbol{\varepsilon}}.
\end{equation}

Similarly,  the one-dimensional left eigenspace of the transfer matrix $\mathcal{T}(\lambda)$ associated with the eigenvalue $\tau(\lambda)$ is generated by any of the separate states ${}_{\boldsymbol{\varepsilon}}\bra{Q_{\tau,\boldsymbol{\varepsilon}}}$ for any choice of $\boldsymbol{\varepsilon}\equiv( \epsilon_{\alpha_+}, \epsilon_{\alpha_-}, \epsilon_{\beta_+}, \epsilon_{\beta_-})\in\{-1,1\}^4$ such that $\epsilon_{\alpha_+} \epsilon_{\alpha_-}\epsilon_{\beta_+} \epsilon_{\beta_-}=1$.
For two such choices of $\boldsymbol{\varepsilon}$, the corresponding states are proportional: 
\begin{align}
  {}_{\boldsymbol{\varepsilon'}}\bra{Q_{\tau,\boldsymbol{\varepsilon'}}}
  &= \prod_{n=1}^N\frac{Q_{\tau,\boldsymbol{\varepsilon'}}(\xi_n+\frac{\eta}{2})}{Q_{\tau,\boldsymbol{\varepsilon}}(\xi_n+\frac{\eta}{2})}
  \frac{\mathbf{a}_{\boldsymbol{\epsilon'}}(\xi_n+\frac{\eta}{2})}{\mathbf{a}_{\boldsymbol{\epsilon}}(\xi_n+\frac{\eta}{2})}\
  {}_{\boldsymbol{\varepsilon}}\bra{Q_{\tau,\boldsymbol{\varepsilon}}}
     \nonumber\\
  &=\frac{\prod_{j=1}^{q_{\boldsymbol{\varepsilon'}}} d(\lambda_{\boldsymbol{\varepsilon'},j})\, d(-\lambda_{\boldsymbol{\varepsilon'},j})}
            {\prod_{j=1}^{q_{\boldsymbol{\varepsilon}}} d(\lambda_{\boldsymbol{\varepsilon},j})\, d(-\lambda_{\boldsymbol{\varepsilon},j})}\,
            \prod_{n=1}^N \frac{\mathbf{a}_{\boldsymbol{\epsilon'}}(\xi_n+\frac{\eta}{2})}{\mathbf{a}_{\boldsymbol{\epsilon}}(\xi_n+\frac{\eta}{2})} \
             {}_{\boldsymbol{\varepsilon}}\bra{Q_{\tau,\boldsymbol{\varepsilon}}}.
\end{align}
\end{proposition}

\begin{proof}
This is a direct consequence of the previous study, of the following identities,
\begin{equation}
    \frac{Q_{\tau,\boldsymbol{\varepsilon'}}(\xi_n^{(1)})}{Q_{\tau,\boldsymbol{\varepsilon'}}(\xi_n^{(0)})}
    = \frac{\mathbf{a}_{\boldsymbol{\epsilon}}(\xi_n^{(0)})}{\mathbf{a}_{\boldsymbol{\epsilon'}}(\xi_n^{(0)})}\,
    \frac{Q_{\tau,\boldsymbol{\varepsilon}}(\xi_n^{(1)})}{Q_{\tau,\boldsymbol{\varepsilon}}(\xi_n^{(0)})} ,   
\end{equation}
and
\begin{align}
 &\ket{\mathbf{h},\beta+1}_{\boldsymbol{\varepsilon'}}
 =\prod_{n=1}^N \left(\frac{\mathbf{a}_{-\boldsymbol{\epsilon}}(\xi_n+\frac{\eta}{2})}{\mathbf{a}_{-\boldsymbol{\epsilon'}}(\xi_n+\frac{\eta}{2})}\right)^{\! h_n}\
 \ket{\mathbf{h},\beta+1}_{\boldsymbol{\varepsilon}},
 \label{prop-states-1}\\
 & {}_{\boldsymbol{\varepsilon'}}\bra{\beta-1,\mathbf{h}}
 =\prod_{n=1}^N \left(\frac{\mathbf{a}_{-\boldsymbol{\epsilon}}(\xi_n+\frac{\eta}{2})}{\mathbf{a}_{-\boldsymbol{\epsilon'}}(\xi_n+\frac{\eta}{2})}\right)^{\! 1-h_n}
 {}_{\boldsymbol{\varepsilon}}\bra{\beta-1,\mathbf{h}},
 \label{prop-states-2}
\end{align}
and of the fact that the product $\mathbf{a}_{\boldsymbol{\epsilon}}(\lambda)\, \mathbf{a}_{-\boldsymbol{\epsilon}}(\lambda)$ does not depend on $\boldsymbol{\epsilon}$.
\end{proof}

\section{Scalar product of separate states}
\label{sec-sp}

Let $P(\lambda)$ and $Q(\lambda)$ be two polynomials in $\cosh(2\lambda)$, of respective degree $n_p$ and $n_q$, and which can be expressed as
\begin{equation}\label{P-Q-form}
   P(\lambda)=\prod_{j=1}^{n_p}\frac{\cosh(2\lambda)-\cosh(2p_j)}{2},
   \qquad
   Q(\lambda)=\prod_{j=1}^{n_q}\frac{\cosh(2\lambda)-\cosh(2q_j)}{2}.
\end{equation}
It is easy to see that the scalar product of any two separate states of the form \eqref{separate-R} and \eqref{separate-L} (or \eqref{separate-L-bis}) constructed from $P$ and $Q$ can be represented as a determinant.

\begin{proposition}\label{prop-sc-det-1}
Let us suppose that the inhomogeneity parameters are generic \eqref{cond-inh} and that the condition \eqref{cond-3bis} is satisfied. 
Let $\boldsymbol{\varepsilon},\, \boldsymbol{\varepsilon'}\in\{-1,1\}^4$ such that $\epsilon_{\alpha_+}\! \epsilon_{\alpha_-}\!\epsilon_{\beta_+} \!\epsilon_{\beta_-}\!\!=\epsilon_{\alpha_+}' \!\epsilon_{\alpha_-}'\!\epsilon_{\beta_+}' \!\epsilon_{\beta_-}'=1$.
The scalar products of the separate states ${}_{\boldsymbol{\varepsilon}}\bra{Q}$ built as in \eqref{separate-L} or \eqref{separate-L-bis} from $Q(\lambda)$, and $\ket{P}_{\boldsymbol{\varepsilon'}}$ built as in \eqref{separate-R} from $P(\lambda)$, admit the following determinant representation:
\begin{multline}\label{sc-det-1}
    {}_{\boldsymbol{\varepsilon}}\moy{Q\, |\, P}_{\boldsymbol{\varepsilon'}}
    =\frac 1{N(\{\xi\},\beta)}\, 
    \frac{\widehat{V}(\xi_1^{(0)},\ldots,\xi_N^{(0)})}{\widehat{V}(\xi_1^{(1)},\ldots,\xi_N^{(1)})}\,
    \\
    \times
    \det_{1\le i,j\le N}\left[\sum_{h=0}^1\left(-\frac{\mathbf{a}_{\boldsymbol{\epsilon'}}(\xi_i+\frac{\eta}{2})}{\mathbf{a}_{\boldsymbol{-\epsilon}}(\xi_i+\frac{\eta}{2})}\right)^{\! h}\, P(\xi_i^{(h)})\, Q(\xi_i^{(h)})\, \left(\frac{\cosh(2\xi_i^{(1-h)})}{2}\right)^{\! j-1}\right].
\end{multline}
\end{proposition}

The representation \eqref{sc-det-1} of Proposition~\ref{prop-sc-det-1} is a direct consequence of the representations \eqref{separate-L-bis} and \eqref{separate-R} of the separate states, of the proportionality relation \eqref{prop-states-1}, and of the orthogonality relation \eqref{orthog-states}.

It is an important advantage of the separation of variables that the scalar products of separate states are always expressed as determinants. However the present formula \eqref{sc-det-1} becomes very difficult to use in the homogeneous and thermodynamic limits. For this reason it is important to recast this expression in a more convenient form depending more directly on the roots of polynomials $P$ and $Q$, and in which the dependence on the inhomogeneous parameters is such that taking the homogeneous limit becomes straightforward. As in the XXX case \cite{KitMNT17}, we will in fact show that it is possible to rewrite the expression \eqref{sc-det-1} for the scalar product of two arbitrary separate states (without requiring any of them to be an eigenstate) in terms of a generalized Slavnov determinant \cite{Sla89}. We will proceed with this computations following an approach similar to what was done in \cite{KitMNT17}. However, due to some additional difficulties in the XXZ case, the most general result  becomes much more involved than in the rational case. Hence, we will gather many technical details in Appendices~\ref{app-det-id} and \ref{app-Slavnov}, and insist mainly on the cases which seem to be the most interesting for the physical applications. As we shall see, in these cases, the obtained formulas simplify due to the use of the Bethe equations.

\subsection{Scalar product of two arbitrary separate states} 

Let us first notice that the whole dependence on $\boldsymbol{\varepsilon},\, \boldsymbol{\varepsilon'}$ in the expression \eqref{sc-det-1} is contained in the ratio of $\mathbf{a}_{\boldsymbol{\epsilon'}}(\xi_i+\frac{\eta}{2})$ and $\mathbf{a}_{\boldsymbol{-\epsilon}}(\xi_i+\frac{\eta}{2})$, which is of the form
\begin{equation}\label{form-ratio}
   \frac{\mathbf{a}_{\boldsymbol{\epsilon'}}(\xi_i+\frac{\eta}{2})}{\mathbf{a}_{\boldsymbol{-\epsilon}}(\xi_i+\frac{\eta}{2})}
   =\prod_{\ell=1}^{n_{\boldsymbol{\varepsilon},\boldsymbol{\varepsilon'}}}\frac{\sinh(\xi_i+a_\ell)}{\sinh(\xi_i-a_\ell)}.
\end{equation}
In this expression,
\begin{align}
   &\bullet\ n_{\boldsymbol{\varepsilon},\boldsymbol{\varepsilon'}}= 0 \qquad \text{if}  \quad \boldsymbol{\varepsilon}=-\boldsymbol{\varepsilon'}
   \label{case0} \\
   &\bullet\ n_{\boldsymbol{\varepsilon},\boldsymbol{\varepsilon'}}=4, \quad\
   \{a_\ell\}=\{\epsilon'_{\alpha_+}\alpha_+,\epsilon'_{\alpha_-}\alpha_-,\epsilon'_{\beta_+}(-\beta_++i\pi/2),\epsilon'_{\beta_-}(\beta_-+i\pi/2)\}   \nonumber\\
   &\hspace{10cm}
   \text{if} \quad \boldsymbol{\varepsilon}=\boldsymbol{\varepsilon'},
   \label{case4}\\
   &\bullet\ n_{\boldsymbol{\varepsilon},\boldsymbol{\varepsilon'}}=2, \quad\
   \{a_\ell\}\subset\{\epsilon'_{\alpha_+}\alpha_+,\epsilon'_{\alpha_-}\alpha_-,\epsilon'_{\beta_+}(-\beta_++i\pi/2),\epsilon'_{\beta_-}(\beta_-+i\pi/2)\}   \nonumber\\
   &\hspace{10cm}
   \text{otherwise.} 
   \label{case2}
\end{align}
For purely technical reasons, it is in fact convenient to treat the case $n_{\boldsymbol{\varepsilon},\boldsymbol{\varepsilon'}}= 0$ as the $n_{\boldsymbol{\varepsilon},\boldsymbol{\varepsilon'}}= 2$ case. Therefore we introduce in this case an arbitrary parameter $\tilde{a}$ so as to rewrite the ratio \eqref{form-ratio} as
\begin{equation}\label{form-ratio2}
   \frac{\mathbf{a}_{\boldsymbol{\epsilon'}}(\xi_i+\frac{\eta}{2})}{\mathbf{a}_{\boldsymbol{-\epsilon}}(\xi_i+\frac{\eta}{2})}
   =\prod_{\ell=1}^{n_a}\frac{\sinh(\xi_i+a_\ell)}{\sinh(\xi_i-a_\ell)},
\end{equation}
where $n_a=2$ and $\{a_1,a_2\}=\{\tilde a, -\tilde a\}$ if $\boldsymbol{\varepsilon}=-\boldsymbol{\varepsilon'}$, whereas $n_a=n_{\boldsymbol{\varepsilon},\boldsymbol{\varepsilon'}}$ and $\{a_\ell\}$ is given by \eqref{case4} or \eqref{case2} otherwise.
In the following, we shall use the notation of \eqref{form-ratio2}. 

Let us also introduce some additional notations. For two arbitrary functions $f$ and $g$, and any set of variables $\{z\}\equiv\{z_1,\ldots,z_L\}$, we define 
\begin{equation}\label{def-Afg}
   \mathcal{A}_{\{z\}}[f, g]
   =
   \frac{\det_{1\le i,j \le L}\left[ \sum_{\bar\epsilon\in\{+,-\}} 
     f(\bar\epsilon z_i)
     \left(\frac{\cosh(2 z_i+\bar\epsilon\eta)}{2}\right)^{j-1}+\delta_{j,L}\, g(z_i)\right]}
          {\widehat{V}(z_1,\ldots,z_L)}.
\end{equation}
When the function $g$ vanishes identically, we may simply denote \eqref{def-Afg} by $\mathcal{A}_{\{z\}}[f]$.
We also consider a particular function $f_{\boldsymbol{\varepsilon},\boldsymbol{\varepsilon}'}$, defined in terms of the corresponding set $\{a_\ell\}_{1\le \ell \le n_a}$ as 
\begin{equation}
  f_{\boldsymbol{\varepsilon},\boldsymbol{\varepsilon}'}(\lambda)
    = (-1)^N\,
     \prod_{\ell=1}^{n_a}\frac{\sinh(\lambda-a_\ell+\frac{\eta}{2})}{\sinh(a_\ell)}\
     \frac{a(-\lambda)\, d(\lambda)}{\sinh(2\lambda)},
    \label{feps}
\end{equation}
and, in the case $\boldsymbol{\varepsilon}=\boldsymbol{\varepsilon}'$,  a function $g_{\boldsymbol{\varepsilon}}^{(L)}$ such that
\begin{align}
  &g_{\boldsymbol{\varepsilon}}^{(N)}(\lambda)
  =  \frac{\sinh(\sum_\ell a_\ell-\eta)}{\prod_{\ell}\sinh(a_\ell)}\,
     a(\lambda)\, d(\lambda)\, a(-\lambda)\, d(-\lambda),
     \label{geps-N}\\
   &g_{\boldsymbol{\varepsilon}}^{(L)}(\lambda)
   = (-1)^{L-N}\, g_{\boldsymbol{\varepsilon}}^{(N)}(\lambda)
   - \bar{f}_{\boldsymbol{\varepsilon},\boldsymbol{\varepsilon}}^{(L)}(\lambda)
   \qquad \text{if } L>N,
   \label{geps-L>N}
\end{align}
whereas, if $L<N$, the function $g_{\boldsymbol{\varepsilon}}^{(L)}(\lambda)$ is defined by induction as
\begin{multline}
    g^{(L)}_{\boldsymbol{\varepsilon}}(z)
    =
    \frac{\prod_{\ell}\sinh(a_\ell)\ \bar{f}^{(L)}_{\boldsymbol{\varepsilon},\boldsymbol{\varepsilon}}(z)}{\sinh((L+1-N)\eta-\sum_\ell a_\ell)}
    \cdot 
  \lim_{z'\to\infty}
  \frac{\bar{f}^{(L+1)}_{\boldsymbol{\varepsilon},\boldsymbol{\varepsilon}}(z')+g^{(L+1)}_{\boldsymbol{\varepsilon}}(z')}{\varsigma(z')^{N+L}}
       \\
    - \bar{f}^{(L)}_{\boldsymbol{\varepsilon},\boldsymbol{\varepsilon}}(z)  
   - \bar{f}^{(L+1)}_{\boldsymbol{\varepsilon},\boldsymbol{\varepsilon}}(z)
   -g^{(L+1)}_{\boldsymbol{\varepsilon}}(z).
   \label{rec-geps}
\end{multline}
Here we have used the shortcut notation:
\begin{equation}\label{feps-L}
\bar{f}^{(l)}_{\boldsymbol{\varepsilon},\boldsymbol{\varepsilon}'}(\lambda) 
     =
\sum_{\bar\epsilon\in\{+,-\}} 
     f_{\boldsymbol{\varepsilon},\boldsymbol{\varepsilon}'}(\bar\epsilon \lambda)
     \left(\frac{\cosh(2\lambda+\bar\epsilon\eta)}{2}\right)^{l-1},
\end{equation}
for a generic integer $l$. Then we have the following result:

\begin{theorem}\label{prop-hom}
Let us suppose that the inhomogeneity parameters are generic \eqref{cond-inh} and that the condition \eqref{cond-3bis} is satisfied. 
Let $\boldsymbol{\varepsilon},\, \boldsymbol{\varepsilon'}\in\{-1,1\}^4$ such that $\epsilon_{\alpha_+}\! \epsilon_{\alpha_-}\!\epsilon_{\beta_+} \!\epsilon_{\beta_-}\!\!=\epsilon_{\alpha_+}' \!\epsilon_{\alpha_-}'\!\epsilon_{\beta_+}' \!\epsilon_{\beta_-}'=1$.

The scalar products of the separate states ${}_{\boldsymbol{\varepsilon}}\bra{Q}$ built as in \eqref{separate-L} or \eqref{separate-L-bis} from $Q(\lambda)$ with roots $\{q\}\equiv\{q_1,\ldots,q_{n_q}\}$, and $\ket{P}_{\boldsymbol{\varepsilon'}}$ built as in \eqref{separate-R} from $P(\lambda)$ with roots $\{p\}\equiv\{p_1,\ldots,p_{n_p}\}$, admit the following determinant representation:
\begin{equation}\label{sc-hom-A}
   {}_{\boldsymbol{\varepsilon}}\moy{Q\, |\, P}_{\boldsymbol{\varepsilon'}}
    = (-1)^{N(n_p+n_q)}\ Z_\beta \ \bar{Z}_{(\{a\},\{\xi\})}\
    \Gamma_{\{a\}}^{(n_p+n_q)}\
     \mathcal{A}_{\{q\}\cup\{p\}}[f_{\boldsymbol{\varepsilon},\boldsymbol{\varepsilon}'}, g_{\boldsymbol{\varepsilon},\boldsymbol{\varepsilon}'}],
\end{equation}
where the function $f_{\boldsymbol{\varepsilon},\boldsymbol{\varepsilon}'}$ is given as in \eqref{feps}.

Here the function $g_{\boldsymbol{\varepsilon},\boldsymbol{\varepsilon}'}$ is defined to be identically zero if $\boldsymbol{\varepsilon}\not= \boldsymbol{\varepsilon'}$, while for $\boldsymbol{\varepsilon}= \boldsymbol{\varepsilon'}$, it is equal to the function $g_{\boldsymbol{\varepsilon}}^{(n_p+n_q)}$ defined as in \eqref{geps-N}, \eqref{geps-L>N} or \eqref{rec-geps} according to whether $n_p+n_q$ is equal, larger or smaller than $N$.

Finally, the normalization coefficient $Z_\beta$, $\bar{Z}_{(\{a\},\{\xi\})}$ and $\Gamma_{\{a\}}^{(n_p+n_q)}$ are respectively given by
\begin{align}\label{Zbeta}
    &Z_\beta=\prod_{j=1}^N\left[ \frac{1}{\mathsf{b}_-(\beta+1+N-2j)}\frac{\sinh(\eta(\beta+N-j))}{\sinh(\eta(\beta+1+N-2j))}\right] ,\\
    &\bar{Z}_{(\{a\},\{\xi\})}=\prod_{i=1}^N\left\{ e^{\xi_i}\, \mathsf{g}_-(\eta/2-\xi_i)\, \prod_{\ell=1}^{n_a}\frac{\sinh(a_\ell)}{\sinh(\xi_i-a_\ell)}\right\}
    \label{Zxi}
\end{align}
and
\begin{equation}\label{norm-gamma}
   \Gamma_{\{a\}}^{(n_p+n_q)}=
   \begin{cases}
       {\displaystyle \prod\limits_{j=1}^{n_p+n_q-N}\frac{\prod_\ell\sinh(a_\ell)}{\sinh\big(j\eta-\sum_\ell a_\ell\big)} }
               &\text{if }\ n_p+n_q\ge N, \vspace{2mm}\\
      {\displaystyle \prod\limits_{j=0}^{N-n_p-n_q-1} 
      \frac{\sinh\big(-j\eta-\sum_\ell a_\ell\big)}{\prod_\ell\sinh(a_\ell)} }
                &\text{if }\  n_p+n_q<N.
   \end{cases}
\end{equation}
\end{theorem}

\begin{proof}
Using the notation \eqref{def-Afg}, we can rewrite \eqref{sc-det-1} as
\begin{multline}\label{sc-det-2}
    {}_{\boldsymbol{\varepsilon}}\moy{Q\, |\, P}_{\boldsymbol{\varepsilon'}}
    =(-1)^N\, Z_\beta\,
    \prod_{j=1}^N\frac{e^{\xi_j}\, \mathsf{g}_-(\eta/2-\xi_j)}{\prod_{\ell=1}^{n_a}\sinh(\xi_j-a_\ell)}\\
   \times
    \prod_{i=1}^N\frac{P(\xi_i^{(0)})\, Q(\xi_i^{(0)})\,P(\xi_i^{(1)})\, Q(\xi_i^{(1)})}{P(\xi_i)\, Q(\xi_i)}\
    \mathcal{A}_{\{\xi_1,\ldots,\xi_N\}}\left[ f_{\{a\}, \{p\}\cup\{q\}} \right]
\end{multline}
in terms of the function $f_{\{a\}, \{p\}\cup\{q\}}\equiv f_{\{a_1,\ldots, a_{n_a}\},\{p_1,\ldots,p_{n_p}\}\cup\{q_1,\ldots,q_{n_q}\}}$ defined as
\begin{equation}
   f_{\{a\}, \{p\}\cup\{q\}}(z)=\frac{\prod_{\ell=1}^{n_a}\sinh(z+a_\ell)}{\sinh(2z)}\prod_{\mu\in\{p\}\cup\{q\}}\frac{\cosh(2z)-\cosh(2\mu)}{\cosh(2z+\eta)-\cosh(2\mu)}.
\end{equation}
In \eqref{sc-det-2},  we have also used  the explicit expression \eqref{norm-states} of $N(\{\xi\},\beta)$.

We now use the identities of Appendix~\ref{app-det-id} to transform \eqref{sc-det-2} in terms of a new ratio of determinants in which the role of the sets of variables $\{\xi\}$ and $\{p\}\cup\{q\}$ are exchanged. Reinserting part of the normalization coefficient into the determinant in the numerator, we finally obtain \eqref{sc-hom-A}.
\end{proof}

\begin{rem}
  Note that the normalization coefficient \eqref{norm-gamma}, and hence the scalar product \eqref{sc-hom-A}, vanishes in the case $\boldsymbol{\varepsilon}= -\boldsymbol{\varepsilon'}$ and $n_p+n_q<N$.
\end{rem}

Theorem~\ref{prop-hom} is the XXZ analog of the first part of Theorem~4.1 of \cite{KitMNT17}. The expression~\eqref{sc-hom-A} is now in a completely regular form with respect to the homogeneous limit. However, the formulas that we obtain for the XXZ scalar products appear to be significantly more complicated than their XXX analogs, at least in the $\boldsymbol{\varepsilon}= \boldsymbol{\varepsilon'}$ case, due to the appearance of the non-zero function $g_{\boldsymbol{\varepsilon},\boldsymbol{\varepsilon}}$. 

As shown in Appendix~\ref{app-Slavnov}, it is also possible to transform the functional~\eqref{def-Afg} of $f$ and $g$ into a new functional which takes the form of a generalized version of the famous Slavnov formula \cite{Sla89} (see Identities~\ref{id-Slav_gen=} and \ref{id-Slav_gen>}). This can be done whatever the form of $f$ and $g$ and for any arbitrary set of integers $\{z_1,\ldots,z_L\}$ defined as the union of two subsets $\{x_1,\ldots,x_{L_1}\}\cup\{y_1,\ldots,y_{L_2}\}$. This means in particular that we can re-write the scalar product \eqref{sc-hom-A} in terms of a generalized Slavnov determinant by using Identities~\ref{id-Slav_gen=} and \ref{id-Slav_gen>}, and this for any two arbitrary separate states (i.e. without supposing one of the two states to be an eigenstate). However, the general formula is quite cumbersome\footnote{It nevertheless simplifies in the cases  $\boldsymbol{\varepsilon}\not= \boldsymbol{\varepsilon'}$, for which the function $g_{\boldsymbol{\varepsilon},\boldsymbol{\varepsilon}'}$ vanishes: the resulting formulas are then quite similar to the ones obtained for the XXX chain in \cite{KitMNT17}.}. 
Therefore, we chose not to present it in the main text (the interested reader can refer to Appendix~\ref{app-Slavnov}), and instead to emphasize on a particular case which seems to be the most relevant for the computation of correlation functions: when one of the two sets of variables satisfies the Bethe equations following from the homogeneous T-Q equation \eqref{hom-TQ}.
This is the purpose of the next subsection.

\subsection{Scalar product of an eigenstate with an arbitrary separate state}

As mentioned above, the generalization of the Slavnov formula for the scalar products of two arbitrary separate states is more cumbersome than in the XXX rational case~\cite{KitMNT17}, at least when $\boldsymbol{\varepsilon}= \boldsymbol{\varepsilon'}$ due to the appearance of the non-zero function $g_{\boldsymbol{\varepsilon},\boldsymbol{\varepsilon}}$.
However, the most general case is not the most important for the computation of correlation functions in the thermodynamic limit (i.e. on the half line), for which we need to consider normalized mean values of the form
\begin{equation}
  E(\mathcal{O})
  =\lim_{N\rightarrow\infty}
    \frac{ {}_{\boldsymbol{\varepsilon}}\moy{Q\, |\mathcal{O}|\, Q}_{\boldsymbol{\varepsilon}}}
           { {}_{\boldsymbol{\varepsilon}}\moy{Q\, |\, Q}_{\boldsymbol{\varepsilon}}},
\label{elementary_block}
\end{equation}
where $\mathcal{O}$ is a product of local spin operators \cite{KitMT00,KitKMNST07,KitKMNST08}.
Note that, strictly speaking, the denominator in \eqref{elementary_block} is not the square of the norm of the separate state $\ket{Q}_{\boldsymbol{\varepsilon}}$, but it plays exactly the same role, so, by a slight abuse of language, we will still use the terminology 'norm' in this section.

For the consideration of quantities of the form \eqref{elementary_block} we can make the following remarks:
\begin{itemize}
\item  To compute the norms  and mean values of local operators the case $\boldsymbol{\varepsilon}= \boldsymbol{\varepsilon'}$ (i.e. $n_a=4$) seems to be the most relevant.
 \item In the half-line limit one of the boundaries should become irrelevant, which means that we can  impose the most convenient boundary constraint \eqref{cond-hom-N}. Then all the eigenstates are characterized by polynomials $Q$ of degree $N$ satisfying the homogeneous Baxter equation \eqref{hom-TQ}.
 \item We expect that the resulting action of the product of local operators $\mathcal{O}$ on an eigenstate can always be simply expressed as a linear combination of off-shell separate states associated to polynomials of degree $L\ge N$. 
\end{itemize}
Hence, in this subsection, we will restrict ourselves to the case $\boldsymbol{\varepsilon}= \boldsymbol{\varepsilon'}$ (i.e. $n_a=4$). We shall moreover suppose that $n_p\ge n_q$ and that the polynomial $Q$ satisfies the homogeneous T-Q equation~\eqref{hom-TQ}. We recall that all other cases can be deduced from the general formulas presented in Appendix~\ref{app-Slavnov}.

\begin{theorem}
\label{th-Slavnov1}
Let $P$ and $Q$ be two trigonometric polynomials of the form \eqref{P-Q-form} and of the same degree $n_p=n_q=n$.
We suppose moreover that $Q(\la)$ satisfies the homogeneous T-Q equation \eqref{hom-TQ} with  $\tau(\la)\in\Sigma_\mathcal{T}$, whereas the roots $p_j$ of the trigonometric polynomial $P$ are arbitrary complex numbers.

Then the scalar product of the two corresponding separate states ${}_{\boldsymbol{\varepsilon}}\bra{Q}$ and $\ket{P}_{\boldsymbol{\varepsilon}}$ can be written as
\begin{multline}
\label{eq_Slav_formula}
 {}_{\boldsymbol{\varepsilon}}\moy{Q\, |\, P}_{\boldsymbol{\varepsilon}}
    =  Z_\beta \ \bar{Z}_{(\{a\},\{\xi\})}\
    \Gamma_{\{a\}}^{(2n)}\ H_Q\big[f_{\boldsymbol{\varepsilon},\boldsymbol{\varepsilon}}, g_{\boldsymbol{\varepsilon}}^{(2n)}\big]\,    
    \pl_{j=1}^n\frac{Q(p_j)}{\sinh(2p_j+\eta)\, \sinh(2p_j-\eta)}
    \\
    \times 
    \pl_{j=1}^{n}\left(-\frac{\mathbf{A}_{\boldsymbol{\varepsilon}}(q_j)}{\sinh(2q_j+\eta)}\right) \
    \frac{\widehat{V}(q_1-\frac\eta2,\dots, q_{n}-\frac\eta2)}{\widehat{V}(q_1+\frac\eta2,\dots, q_{n}+\frac\eta2)}\,
    \frac{\det_{1\le j,k\le n}\left[\frac{\partial\tau(p_j)}{\partial q_k}\right]}{\widehat{V}(q_n,\dots, q_1)\widehat{V}(p_1,\dots,p_n)} ,
\end{multline}
where the normalization coefficient $Z_\beta$, $\bar{Z}_{(\{a\},\{\xi\})}$ and $\Gamma_{\{a\}}^{(2n)}$ are respectively given by \eqref{Zbeta}, \eqref{Zxi} and \eqref{norm-gamma}.
The normalization coefficient $H_Q\big[f_{\boldsymbol{\varepsilon},\boldsymbol{\varepsilon}}, g_{\boldsymbol{\varepsilon}}^{(2n)}\big]$ is defined
in terms of the roots $q_1,\ldots, q_{n}$ of $Q$ as
\begin{equation}
H_Q\big[f_{\boldsymbol{\varepsilon},\boldsymbol{\varepsilon}}, g_{\boldsymbol{\varepsilon}}^{(2n)}\big]
= 1+\sum_{j=1}^{n} \frac{g_{\boldsymbol{\varepsilon}}^{(2n)}(q_j)\, \sinh(2 q_j-\eta)}{f_{\boldsymbol{\varepsilon},\boldsymbol{\varepsilon}}(-q_j)\, Q'(q_j)Q(q_j-\eta)}.
\end{equation}
It involves the functions $f_{\boldsymbol{\varepsilon},\boldsymbol{\varepsilon}}$ \eqref{feps}, which can be more simply written as
%
\begin{equation}
 \label{f-function}
  f_{\boldsymbol{\varepsilon},\boldsymbol{\varepsilon}}(\la)=\frac{\mathbf{A}_{\boldsymbol{\varepsilon}}(-\lambda)}{\sinh(2\la-\eta)},
\end{equation}
and $g_{\boldsymbol{\varepsilon}}^{(2n)}$ defined as in \eqref{geps-N}, \eqref{geps-L>N} or \eqref{rec-geps} according to whether $2n$ is equal, larger or smaller than $N$.
\end{theorem}

A comment is due here. This result is essentially the {\it Slavnov formula} \cite{Sla89} (presented as in \cite{KitMT99} in terms of a jacobian), the only difference being the normalization coefficients.
Note however that most of these normalization coefficients
do not depend on the off-shell separate state $\ket{P}$. Hence, they are irrelevant for the computation of the  correlation functions since they will always be cancelled by the same coefficient from the norm of the on-shell state in the denominator of \eqref{elementary_block}. 

\begin{proof}
This is a direct consequence of Theorem~\ref{prop-hom} and of Identity~\ref{slavnov_square}.\end{proof}

As a corollary  of this theorem is the analog of the Gaudin formula \cite{GauMcCW81} for the square of the norm of on-shell separate states:
\begin{corollary}
Let $Q$ be a polynomial of the form \eqref{P-Q-form} of degree degree $n_q=n$ satisfiyng the homogeneous T-Q equation \eqref{hom-TQ} with  $\tau(\la)\in\Sigma_\mathcal{T}$.
Then
\begin{multline}
 {}_{\boldsymbol{\varepsilon}}\moy{Q\, |\, Q}_{\boldsymbol{\varepsilon}}
    = 
     Z_\beta \ \bar{Z}_{(\{a\},\{\xi\})}\
    \Gamma_{\{a\}}^{(2n)}\ H_Q\big[f_{\boldsymbol{\varepsilon},\boldsymbol{\varepsilon}}, g_{\boldsymbol{\varepsilon}}^{(2n)}\big]\,
      \pl_{j=1}^n\frac{\mathbf{A}_{\boldsymbol{\varepsilon}}(q_j)^2\, Q(q_j-\eta)}{\sinh(2q_j+\eta)^2\, \sinh(2q_j-\eta)}
    \\
    \times 
   \frac{\widehat{V}(q_1-\frac\eta2,\dots, q_{n}-\frac\eta2)}{\widehat{V}(q_1+\frac\eta2,\dots, q_{n}+\frac\eta2)}\,
    \frac{\det_{1\le j,k\le n}\left[\frac{\partial}{\partial q_k}\log
    \left(\frac{\mathbf{A}_{\boldsymbol{\varepsilon}}(-q_j)\, Q(q_j+\eta)}{\mathbf{A}_{\boldsymbol{\varepsilon}}(q_j)\, Q(q_j-\eta)}\right)\right]}{\widehat{V}(q_n,\dots, q_1)\widehat{V}(q_1,\dots,q_n)} .
\end{multline}
\end{corollary}

Finally, we present a generalization of Theorem~\ref{th-Slavnov1} in the case $n_p>n_q$:

\begin{theorem}
\label{th-Slavnov2}
Let $P$ and $Q$ be two trigonometric polynomials of the form \eqref{P-Q-form} with $n_p>n_q$.
We suppose moreover that $Q(\la)$ satisfies the homogeneous T-Q equation \eqref{hom-TQ} with  $\tau(\la)\in\Sigma_\mathcal{T}$, whereas the roots $p_j$ of the trigonometric polynomial $P$ are arbitrary complex numbers.

Then the scalar product of the two corresponding separate states ${}_{\boldsymbol{\varepsilon}}\bra{Q}$ and $\ket{P}_{\boldsymbol{\varepsilon}}$ can be written as
 a generalised Slavnov determinant with a rank one correction:
\begin{multline}
\label{eq_gen_Slavnov_formula}
 {}_{\boldsymbol{\varepsilon}}\moy{Q\, |\, P}_{\boldsymbol{\varepsilon}}
    = 
    (-1)^{N(n_p+n_q)}\ Z_\beta \ \bar{Z}_{(\{a\},\{\xi\})}\
    \Gamma_{\{a\}}^{(n_p+n_q)}\
    \prod_{i=1}^{n_p}\frac{Q(p_i)}{\sinh(2p_i+\eta)\,\sinh(2p_i-\eta)}
    \\
    \times
    \prod_{k=1}^{n_q}\frac{f_{\boldsymbol{\varepsilon},\boldsymbol{\varepsilon}}(-q_k)}{\sinh\eta\,\sinh(2q_k)}\,
    \frac{\widehat{V}(q_1-\frac\eta2,\dots, q_{n_q}-\frac\eta2)}{\widehat{V}(q_1+\frac\eta2,\dots, q_{n_q}+\frac\eta2)}\,
 \frac{\det_{n_p}\left(\mathcal{S}_\tau +\mathcal{P}\right)}{\widehat{V}(q_1,\dots q_{n_q})\widehat{V}(p_{n_p},\dots p_1)} .
\end{multline}
The $n_p\times n_p$ matrix $\mathcal{S}_\tau$ is a generalized Slavnov matrix with elements:
\begin{align}
   \big[\mathcal{S}_\tau\big]_{ j k}
       =&\frac{\partial\tau(p_j)}{\partial q_k},\qquad \text{if}\quad k\le n_q,
       \nonumber\\
   \big[ \mathcal{S}_\tau \big]_{ j k}
       =&\sum_{\bar\epsilon=\pm 1}\bar\epsilon\,
       \mathbf{A}_{\boldsymbol{\varepsilon}}(-\bar\epsilon p_j)\,
       \sinh(2p_j+\bar\epsilon\eta)\, \frac{Q(p_j+\bar\epsilon\eta)}{Q(p_j)}\,
       \left(\frac{\cosh(2 p_j+\bar\epsilon\eta)}2\right)^{k-n_q-1} 
       \nonumber\\
   &\hphantom{\frac{\partial\tau(p_j)}{\partial q_k},}\qquad \text{if}\quad k> n_q.
\end{align}
The additional rank one matrix $\mathcal{P}$ has only one non-zero column:  $\mathcal{P}_{ j k}=0$ if $j\neq n_p$ and $k\le n_q$, and
\begin{align}
\mathcal{P}_{ j n_p}
   &= g_{\boldsymbol{\varepsilon}}^{(n_p+n_q)}(p_j)\,
     \frac{\sinh(2p_j+\eta)\sinh(2p_j-\eta)}{Q^2(p_j)}
        \nonumber\\
  &\qquad
  -\sum_{\bar\epsilon=\pm 1}\bar\epsilon\,\mathbf{A}_{\boldsymbol{\varepsilon}}(-\bar\epsilon p_j)\,
    \sinh(2p_j+\bar\epsilon\eta)\, \frac{Q(p_j+\bar\epsilon\eta)}{Q(p_j)}
    \nonumber\\
   &\qquad
   \times
   \sum_{l=1}^{n_q} \frac{2\,g_{\boldsymbol{\varepsilon}}^{(n_p+n_q)}(q_l)\, \sinh(2q_l-\eta)}{  f_{\boldsymbol{\varepsilon},\boldsymbol{\varepsilon}}(-q_l)\, Q'(q_l)\, Q(q_l-\eta)\,\big[\cosh(2 p_j+\bar\epsilon\eta)-\cosh(2q_l-\eta)\big]}
\end{align}
if $k>n_q$.
\end{theorem}

\begin{proof}
This is a direct consequence of Theorem~\ref{prop-hom} and of Identity~\ref{id-Slav_gen>}.\end{proof}

\section{Conclusion}
We have shown that the program of rewriting SoV type determinant representations for the scalar products of separate states (that include all eigenstates of the transfer matrix) in terms of generalized Slavnov's type determinants can be achieved for the most general  XXZ spin-1/2 integrable open chain. It generalizes to the trigonometric case the results obtained in the rational model \cite{KitMNT17}. It paves the way to the computation of form factors and correlation functions and to the study of the dynamics of these models that we plan to address in future publications. In particular, these formulae allow to start the computation of elementary blocks of correlation functions in the SoV framework on a similar ground to that previously developed for some special boundary conditions in \cite{KitKMNST07,KitKMNST08}. Along these lines, it would be very important also to pursue this program for the case of cyclic representations of the 6-vertex Yang-Baxter algebra. Indeed this case is relevant for the lattice version of the Sine-Gordon field theory and to the Chiral Potts model \cite{BazS90,GroMN12}. A rewriting of the scalar products in a way similar to what we achieved here would give, in particular, the possibility to obtain a direct re-derivation of the order parameter of the Chiral Potts model. More interestingly, it would also lead to the computation of the relevant form factors of local operators and correlation functions in the thermodynamic limit for these important models. We would like to stress here that the rewriting of the various determinants arising all along such a program is essential to be able to take the homogeneous limit 
and then eventually the thermodynamic limit. Hence those rewritings constitute a cornerstone for the applicability of the SoV method itself that  requires 
to consider the inhomogeneity parameters 
to be in generic position. This is perfectly acceptable only if the 
 homogeneous limit can be effectively taken in this end. 
We have no doubt that such a program can be achieved for general integrable models solvable by the quantum Separation of Variables method along the lines presented in this article.

\section*{Acknowledgements}
J. M. M., G. N. and V. T. are supported by CNRS. N. K. 
would like to thank LPTHE, Sorbonne University, and LPTMS, Univ. Paris-Sud, for hospitality. N. K. and V. T. would like to thank Laboratoire de Physique, ENS-Lyon, for hospitality.


\appendix

\section{A few useful properties of the gauged transformed operators}
\label{app-prop-gauged}

We gather here for completeness some relations involving the elements of the gauged transformed boundary monodromy matrix \eqref{gauged-U} or \eqref{U-SOS}.

Both matrices satisfy the dynamical reflection equation \eqref{dyn_refl}. A few useful commutation relations (which are therefore equally valid for the elements of \eqref{gauged-U} or \eqref{U-SOS})  issued from this equation are gathered here:
\begin{equation}
   \mathcal{B}(\lambda|\beta+1)\, \mathcal{B}(\mu |\beta-1)
   = \mathcal{B}(\mu |\beta+1)\, \mathcal{B}(\lambda|\beta-1),
   \label{comm-BB}
\end{equation}
\begin{multline}
   \big[ \mathcal{A}(\lambda|\beta) , \mathcal{A}(\mu|\beta)\big]
   =\frac{\sinh\eta\ \sinh(\lambda+\mu-\eta\beta)}{\sinh(\lambda+\mu)\ \sinh(\eta(\beta-1))}
   \\
   \times
   \big\{ \mathcal{B}(\lambda|\beta)\, \mathcal{C}(\mu |\beta)
   -\mathcal{B}(\mu|\beta)\, \mathcal{C}(\lambda|\beta)\big\},
   \label{comm-AA}
\end{multline}
\begin{multline}
   \big[ \mathcal{D}(\lambda|\beta) , \mathcal{D}(\mu|\beta)\big]
   =\frac{\sinh\eta\ \sinh(\lambda+\mu+\eta\beta)}{\sinh(\lambda+\mu)\ \sinh(\eta(\beta+1))}
   \\
   \times
   \big\{ \mathcal{C}(\mu|\beta)\, \mathcal{B}(\lambda |\beta)
   -\mathcal{C}(\lambda |\beta)\, \mathcal{B}(\mu|\beta)\big\},
   \label{comm-DD}
\end{multline}
\begin{multline}
  \mathcal{A}(\mu|\beta+1)\, \mathcal{B}(\lambda|\beta+1)
  =\frac{\sinh(\lambda+\mu-\eta)\ \sinh(\lambda-\mu+\eta)}{\sinh(\lambda+\mu)\ \sinh(\lambda-\mu)}\,
  \mathcal{B}(\lambda|\beta+1)\, \mathcal{A}(\mu|\beta-1)
  \\
  -\frac{\sinh(\lambda+\mu-\eta)\ \sinh\eta\ \sinh(\lambda-\mu+\eta\beta)}{\sinh(\lambda+\mu)\ \sinh(\lambda-\mu)\ \sinh(\eta\beta)} \,\mathcal{B}(\mu|\beta+1)\, \mathcal{A}(\lambda|\beta-1)
  \\
  +\frac{\sinh\eta\ \sinh(\lambda+\mu-\eta(\beta+1))}{\sinh(\lambda+\mu)\ \sinh(\eta\beta)}\,
  \mathcal{B}(\mu|\beta+1)\, \mathcal{D}(\lambda|\beta+1),
  \label{comm-AB}
\end{multline}
\begin{multline}
  \mathcal{B}(\lambda|\beta-1)\, \mathcal{D}(\mu|\beta-1)
  =\frac{\sinh(\lambda+\mu-\eta)\ \sinh(\lambda-\mu+\eta)}{\sinh(\lambda+\mu)\ \sinh(\lambda-\mu)}\,
  \mathcal{D}(\mu|\beta+1)\, \mathcal{B}(\lambda|\beta-1)
  \\
  +\frac{\sinh(\lambda+\mu-\eta)\ \sinh\eta\ \sinh(\lambda-\mu-\eta\beta)}{\sinh(\lambda+\mu)\ \sinh(\lambda-\mu)\ \sinh(\eta\beta)} \,\mathcal{D}(\lambda|\beta+1)\, \mathcal{B}(\mu|\beta-1)
  \\
  -\frac{\sinh\eta\ \sinh(\lambda+\mu+\eta(\beta-1))}{\sinh(\lambda+\mu)\ \sinh(\eta\beta)}\,
  \mathcal{A}(\lambda|\beta-1)\, \mathcal{B}(\mu|\beta-1).
  \label{comm-BD}
\end{multline}

By means of the transformation \eqref{gauged-U}, the elements of the matrix $\widetilde{\mathcal{U}}(\lambda|\beta)$ can explicitly be expressed in terms of the elements of the matrix  $\mathcal{U}_-(\lambda)$ \eqref{def-U-} as
\begin{align}
   &\widetilde{\mathcal{A}}_-(\lambda|\beta)
   =\widetilde{\mathcal{D}}_-(\lambda|-\beta)
   =\frac{1}{2\sinh(\eta\beta)}
   \Big\{ - e^{2\lambda-\eta-\eta\beta}\mathcal{A}_-(\lambda)-e^{\lambda-\frac{\eta}{2}+\eta\alpha} \mathcal{B}_-(\lambda) \nonumber\\
   &\hspace{6cm}
   +e^{\lambda-\frac{\eta}{2}-\eta\alpha}\mathcal{C}_-(\lambda)+e^{\eta\beta}\mathcal{D}_-(\lambda)\Big\},
    \label{Atilde}\displaybreak[0]\\
  &\widetilde{\mathcal{B}}_-(\lambda|\beta)
  =\widetilde{\mathcal{C}}_-(\lambda|-\beta)
  =\frac{1}{2\sinh(\eta\beta)}
   \Big\{ - e^{2\lambda-\eta+\eta\beta}\mathcal{A}_-(\lambda)-e^{\lambda-\frac{\eta}{2}+\eta\alpha} \mathcal{B}_-(\lambda) \nonumber\\
   &\hspace{6cm}
   +e^{\lambda-\frac{\eta}{2}+\eta(2\beta-\alpha)}\mathcal{C}_-(\lambda)+e^{\eta\beta}\mathcal{D}_-(\lambda)\Big\}  .
   \label{Btilde}
\end{align}
The gauged transformed matrix elements of \eqref{gauged-U}, as well as the SOS boundary matrix elements of \eqref{U-SOS}, also satisfy parity properties of the form
\begin{align}
    &e^{2\lambda}\,\sinh(2\lambda-\eta)\,{\mathcal{A}}_-(-\lambda|\beta-1)
    =\frac{\sinh(\eta(\beta+1))}{\sinh(\eta\beta)}\, \sinh(2\lambda)\,
    {\mathcal{D}}_-(\lambda|\beta+1)
      \nonumber\\
     &\hspace{6.5cm}
    -\frac{\sinh(2\lambda+\eta\beta)}{\sinh(\eta\beta)}\,\sinh\eta \,
    {\mathcal{A}}_-(\lambda|\beta-1),
    \label{parity-A}\\
    &e^{2\lambda}\,\sinh(2\lambda-\eta)\,{\mathcal{D}}_-(-\lambda|\beta+1)
    =\frac{\sinh(\eta(\beta-1))}{\sinh(\eta\beta)}\,\sinh(2\lambda)\,
    {\mathcal{A}}_-(\lambda|\beta-1)
         \nonumber\\
     &\hspace{6.5cm}
    +\frac{\sinh(2\lambda-\eta\beta)}{\sinh(\eta\beta)}\,\sinh\eta\,
    {\mathcal{D}}_-(\lambda|\beta+1),
    \label{parity-D}\\
    &e^{2\lambda}\,\sinh(2\lambda-\eta)\,{\mathcal{B}}_-(-\lambda|\beta)
    =-\sinh(2\lambda+\eta)\, {\mathcal{B}}_-(\lambda|\beta),
    \label{parity-B} \\
    &e^{2\lambda}\,\sinh(2\lambda-\eta)\,{\mathcal{C}}_-(-\lambda|\beta)
    =-\sinh(2\lambda+\eta)\, {\mathcal{C}}_-(\lambda|\beta).
    \label{parity-C}
\end{align}

\section{Action of the SOS boundary operators on the SoV states}
\label{app-act-SOS}

Computing the action of $\mathcal{A}_-^\mathrm{SOS}(\la|\beta-1)$ on $\bra{\beta-1,\mathbf{h}}$ at the $2N$ points $\pm\xi_n^{(h_n)}$, $n\in\{1,\ldots,N\}$, using the fact that
\begin{align}
   &\mathcal{U}_-^\text{SOS}(\eta/2|\beta)=(-1)^N{\det}_qM(0),\\
   &\mathcal{U}_-^\text{SOS}(\eta/2+i\pi/2|\beta)=i\coth\varsigma_-\,{\det}_qM(i \pi/2),
\end{align}
and that $e^{(2N+1)\lambda}\mathcal{A}_-^\mathrm{SOS}(\lambda|\beta-1)$ is a polynomial in $e^{2\lambda}$ of degree $2N+2$, we obtain,
\begin{multline}\label{act-AL}
   \bra{\beta-1,\mathbf{h}}\, \mathcal{A}_-^\mathrm{SOS}(\la|\beta-1)
   =\sum_{n=1}^N\sum_{\epsilon=\pm}
   \frac{\sinh(2\lambda-\eta)\, \sinh(\lambda+\epsilon\xi_n^{(h_n)})}{\sinh(2\xi_n^{(h_n)}-\epsilon\eta)\, \sinh(2\xi_n^{(h_n)})}
   \\
   \times
   \prod_{\substack{j=1\\j\not=n}}^N
   \frac{\sinh^2\lambda-\sinh^2\xi_j^{(h_j)}}{\sinh^2\xi_n^{(h_n)}-\sinh^2\xi_j^{(h_j)}}\,
   \mathsf{A}_-(\epsilon\xi_n^{(h_n)})\, \bra{\beta-1,\mathsf{T}^{\epsilon}_n\mathbf{h}}
   \\
   +(-1)^N\left[ \mathrm{det}_q M(0)\,\cosh(\lambda-\eta/2)\prod_{j=1}^N\frac{\sinh^2\lambda-\sinh^2\xi_j^{(h_j)}}{\sinh^2\frac{\eta}{2}-\sinh^2\xi_j^{(h_j)}}\, \right. 
   \\
    +\left. \coth \varsigma_- \,\mathrm{det}_q M(i\pi/2)\,\sinh(\lambda-\eta/2)\prod_{j=1}^N\frac{\sinh^2\lambda-\sinh^2\xi_j^{(h_j)}}{\cosh^2\frac{\eta}{2}+\sinh^2\xi_j^{(h_j)}}\right] 
    \bra{\beta-1,\mathbf{h}}
    \\
    +2^{2N+1} e^{\lambda+\eta}\sinh(2\lambda-\eta)\,\prod_{j=1}^N\left[ \sinh^2\lambda-\sinh^2\xi_j^{(h_j)}\right]
     \bra{\beta-1,\mathbf{h}}\, \mathcal{A}_-^\infty(\beta-1),
\end{multline}
in which
\begin{multline}
    \bra{\beta-1,\mathbf{h}}\, \mathcal{A}_-^\infty(\beta-1)
    =-\frac{e^{-3\eta/2-\eta(\beta-1)}}{2^{2N+1}\sinh(\eta(\beta-1))}
    \Bigg\{\Bigg[ \frac{\kappa_- \, e^{\eta(\beta-1)} \sinh(\eta\alpha+\tau_-)}{\sinh\varsigma_-}
    \\
    +\frac{(-1)^N\det_qM(0)}{2\prod_{j=1}^N\big(\sinh^2\frac{\eta}{2}-\sinh^2\xi_j^{(h_j)}\big)}
    -\frac{(-1)^N\coth\varsigma_-\, \det_qM(i\frac{\pi}{2})}{2\prod_{j=1}^N\big(\cosh^2\frac{\eta}{2}+\sinh^2\xi_j^{(h_j)}\big)}\Bigg]\bra{\beta-1,\mathbf{h}}
    \\
    +\sum_{n=1}^N\sum_{\epsilon=\pm}
    \frac{e^{\eta/2-\epsilon\xi_n^{(h_n)}}}{\sinh(2\xi_n^{(h_n)}-\epsilon\eta)\, \sinh(2\xi_n^{(h_n)})}
    \frac{\mathsf{A}_-(\epsilon\xi_n^{(h_n)})}{\prod\limits_{j\neq n}\big[ \sinh^2\xi_n^{(h_n)}-\sinh^2\xi_j^{(h_j)}\big]}\, \bra{\beta-1,\mathsf{T}^{\epsilon}_n\mathbf{h}}
\Bigg\},
\end{multline}
and
\begin{equation}
   \mathsf{T}_n^\pm\mathbf{h}=(h_1,\ldots,h_n\pm 1,\ldots, h_N)
   \qquad \text{for}\quad n\in\{1,\ldots,N\}.
\end{equation}
The action of $\mathcal{D}_-^\mathrm{SOS}(\lambda|\beta+1)$ on $\bra{\beta-1,\mathbf{h}}$ can then be obtained by using the parity identity \eqref{parity-A}.

Similarly,
\begin{multline}\label{act-DR}
   \mathcal{D}_-^\mathrm{SOS}(\la|\beta+1)\, \ket{\mathbf{h},\beta+1}
   =\sum_{n=1}^N\sum_{\epsilon=\pm}
   \frac{\sinh(2\lambda-\eta)\, \sinh(\lambda+\epsilon\xi_n^{(h_n)})}{\sinh(2\xi_n^{(h_n)}-\epsilon\eta)\, \sinh(2\xi_n^{(h_n)})}
   \\
   \times
   \prod_{\substack{j=1\\j\not=n}}^N
   \frac{\sinh^2\lambda-\sinh^2\xi_j^{(h_j)}}{\sinh^2\xi_n^{(h_n)}-\sinh^2\xi_j^{(h_j)}}\ 
   k_n^\epsilon\, \mathsf{A}_-(-\epsilon\xi_n^{(1-h_n)})\, \ket{\mathsf{T}^{\epsilon}_n\mathbf{h},\beta+1}
   \\
   +(-1)^N\left[ \mathrm{det}_q M(0)\,\cosh(\lambda-\eta/2)\prod_{j=1}^N\frac{\sinh^2\lambda-\sinh^2\xi_j^{(h_j)}}{\sinh^2\frac{\eta}{2}-\sinh^2\xi_j^{(h_j)}}\, \right. 
   \\
    +\left. \coth \varsigma_- \,\mathrm{det}_q M(i\pi/2)\,\sinh(\lambda-\eta/2)\prod_{j=1}^N\frac{\sinh^2\lambda-\sinh^2\xi_j^{(h_j)}}{\cosh^2\frac{\eta}{2}+\sinh^2\xi_j^{(h_j)}}\right]
     \ket{\mathbf{h},\beta+1}
    \\
    +2^{2N+1} e^{\lambda+\eta}\sinh(2\lambda-\eta)\,\prod_{j=1}^N\left[ \sinh^2\lambda-\sinh^2\xi_j^{(h_j)}\right]
     \mathcal{D}_-^\infty(\beta+1) \, \ket{\mathbf{h},\beta+1},
\end{multline}
in which
\begin{multline}
    \mathcal{D}_-^\infty(\beta+1) \, \ket{\mathbf{h},\beta+1}
    =\frac{e^{-3\eta/2+\eta(\beta+1)}}{2^{2N+1}\sinh(\eta(\beta+1))}
    \Bigg\{\Bigg[ \frac{\kappa_- \, e^{-\eta(\beta+1)} \sinh(\eta\alpha+\tau_-)}{\sinh\varsigma_-}
    \\
    +\frac{(-1)^N\det_qM(0)}{2\prod_{j=1}^N\big(\sinh^2\frac{\eta}{2}-\sinh^2\xi_j^{(h_j)}\big)}
    -\frac{(-1)^N\coth\varsigma_-\, \det_qM(i\frac{\pi}{2})}{2\prod_{j=1}^N\big(\cosh^2\frac{\eta}{2}+\sinh^2\xi_j^{(h_j)}\big)}\Bigg]\ket{\mathbf{h},\beta+1}
    \\
    +\sum_{n=1}^N\sum_{\epsilon=\pm}
    \frac{e^{\eta/2-\epsilon\xi_n^{(h_n)}}}{\sinh(2\xi_n^{(h_n)}-\epsilon\eta)\, \sinh(2\xi_n^{(h_n)})}
    \frac{k_n^\epsilon\, \mathsf{A}_-(-\epsilon\xi_n^{(1-h_n)})}{\prod\limits_{j\neq n}\big[ \sinh^2\xi_n^{(h_n)}-\sinh^2\xi_j^{(h_j)}\big]}\, \ket{\mathsf{T}^{\epsilon}_n\mathbf{h},\beta+1}
\Bigg\},
\end{multline}
and the action of $\mathcal{A}_-^\mathrm{SOS}(\lambda|\beta-1)$ on $\ket{\mathbf{h},\beta+1}$ can be obtained by means of \eqref{parity-D}.

\section{Computation of the normalization coefficient $N(\{\xi\},\beta)$}\label{Normalization-Factor}
We want to compute the following matrix element:
\begin{equation}
F(\beta)=\bra{0}\pl_{k=1}^N\mathcal{A}^{\mathrm{SOS}}_-\left(\frac \eta 2-\xi_k|\beta\right)\ket{\underline{0}},
\end{equation}
from which the value of the normalization $N(\{\xi\},\beta)$ follows. From the boundary bulk decomposition (\ref{boundary-bulk}) it is easy to see that 
\begin{multline}
      \mathcal{A}_-^{\mathrm{SOS}}(\la|\beta)
      =A^{\mathrm{SOS}}(\la|\beta)\, \mathsf{a}_-(\la|\beta+\mathbf{S}_z)\,\widehat{A}^{\mathrm{SOS}}(\la|\beta)
      +B^{\mathrm{SOS}}(\la|\beta)\, \mathsf{c}_-(\la|\beta+\mathbf{S}_z)\, \widehat{A}^{\mathrm{SOS}}(\la|\beta)
      \\
     +A^{\mathrm{SOS}}(\la|\beta)\, \mathsf{b}_-(\la|\beta+\mathbf{S}_z)\, \widehat{C}^{\mathrm{SOS}}(\la|\beta)
     +B^{\mathrm{SOS}}(\la|\beta)\, \mathsf{d}_-(\la|\beta+\mathbf{S}_z)\, \widehat{C}^{\mathrm{SOS}}(\la|\beta).
\end{multline}
Evidently only the term with $\mathsf{b}_-$ will contribute leading to the following expression
\begin{multline}
     F(\beta)
     =\pl_{k=0}^{N-1}\mathsf{b}_-(\frac \eta 2-\xi_{k+1}|\beta+N-2k)
     \\
     \times
     \bra{0}\pl_{j=1}^N\left[A^{\mathrm{SOS}}\left(\frac \eta 2-\xi_j|\beta\right)\, \widehat{C}^{\mathrm{SOS}}\left(\frac \eta 2-\xi_j|\beta\right)\right]
     \ket{\underline{0}}.
\end{multline}
By using the explicit form of the $M^{\mathrm{SOS}}$ and $\widehat{M}^{\mathrm{SOS}}$ monodromy matrices, we can now compute the above matrix element. We can show that
\begin{equation}
\bra{0}\widehat{C}^{\mathrm{SOS}}\left(\frac \eta 2-\xi_1|\beta\right)\dots \widehat{C}^{\mathrm{SOS}}\left(\frac \eta 2-\xi_k|\beta\right)=(-1)^{Nk}\left(\pl_{j=1}^k\,d\left(\xi_j-\frac \eta 2\right)\right)\,\bra{\underline{0}_k}\otimes\bra{0_{N-k}},
\end{equation}
where in the right hand side we obtain a state with first $k$ spins down and all the remaining spins up, and
\begin{multline}
           \bra{\underline{0}_k}\otimes\bra{0_{N-k}}A^{\mathrm{SOS}}(\la|\beta)
           =a(\la)\frac{\sinh(\beta+N-2k)\eta}{\sinh(\beta+N-k)\eta}\\
           \times\pl_{j=1}^{k}\frac{\sinh(\la-\xi_k-\frac\eta 2)}{\sinh(\la-\xi_k+\frac\eta 2)}
           \bra{\underline{0}_k}\otimes\bra{0_{N-k}}.
\end{multline}
These two equations lead to the following result
\begin{multline}
       F(\beta)
       =(-1)^N \left(\pl_{r=0}^{N-1}\mathsf{b}_-\left(\frac \eta 2-\xi_{r+1}|\beta+N-2r\right)\,\frac{\sinh(\beta+N-2r)\eta}{\sinh(\beta+N-r)\eta}\right)
       \\
       \times
      \left(\pl_{j=1}^N\,a\left(\frac\eta 2-\xi_j\right)\,d\left(\xi_j-\frac \eta 2\right)\right) \pl_{j<k} \frac{\sinh(\xi_j+\xi_k)}{\sinh(\xi_j+\xi_k-\eta)}.
\end{multline}


\section{Determinant identities: exchanging the role of the two sets of variables}
\label{app-det-id}

In this appendix we give a detailed proof of several  identities that we use to establish Theorem~\ref{prop-hom}. 

For three sets of arbitrary variables $\{a\}\equiv\{a_1,\ldots,a_{n_a}\}$ with $n_a\in\{2,4\}$,  $\{x\}\equiv\{x_1,\ldots,x_N\}$ and $\{z\}\equiv\{z_1,\ldots,z_M\}$, we consider the quantity
\begin{equation}\label{Abis}
     \mathcal{A}_{\{x\}}[f_{\{a\},\{z\}}]
     \equiv \mathcal{A}_{\{x_1,\ldots,x_N\}}[f_{\{a_1,\ldots,a_{n_a}\},\{z_{1},\ldots, z_M\}}],
\end{equation}
defined as in \eqref{def-Afg} in terms of a function $f_{\{a\},\{z\}}$ of the form
\begin{equation}\label{special-f}
 f_{\{a\},\{z_{1},\ldots, z_M\}}(\lambda)
   = \frac{\prod_{\ell=1}^{n_a}\sinh (\lambda+ a_\ell )}{\sinh 2\lambda}
   \prod_{\ell =1}^M 
   \frac{\varsigma(\lambda)-\varsigma (z_\ell)}{\varsigma (\lambda+ \eta /2)-\varsigma (z_\ell )}.
\end{equation}
Here and in the following, we use for simplicity the shorthand notation:
\begin{equation}\label{short-notation}
   \varsigma(\lambda)=\frac{\cosh(2\lambda)}{2}.
\end{equation}
It is also convenient to introduce the function
\begin{equation}\label{f-j}
 \bar{f}^{(j)}_{\{a\},\{z\}}(\lambda) 
   =\sum_{\epsilon =\pm 1} f_{\{a\},\{z\}}(\epsilon \lambda)\, [\varsigma(\lambda+\epsilon \eta /2)]^{j-1},
\end{equation}
where $f_{\{a\},\{z\}}$ is given by \eqref{special-f}, so that \eqref{Abis} can be simply written as
\begin{equation}\label{Ater}
    \mathcal{A}_{\{x\}}[f_{\{a\},\{z\}}]= \frac{\det_{1\le i,j\le N}\big[ \bar{f}^{(j)}_{\{a\},\{z\}}(x_i) \big]}{\widehat{V}(x_1,\ldots,x_N)}.
\end{equation}

The aim of this appendix is to express the quantity \eqref{Abis} (or equivalently \eqref{Ater}), which is a ratio of two determinants, as a new ratio of determinants in which the role of the two sets of variables $\{x\}$ and $\{z\}$ has been exchanged. When applied to scalar products, and in particular to formula~\eqref{sc-det-2}, these identities lead to Theorem~\ref{prop-hom}, which makes possible the computation of the homogeneous limit of these scalar products.

\subsection{The case $N=M$}

\begin{identity}\label{id-4-N=L}
Let $N=M$. Then
\begin{equation}\label{id-A-B-1}
    \mathcal{A}_{\{x\}}\big[ f_{\{a\},\{z\}}\big]
    =(-1)^N\,
    \mathcal{A}_{\{z\}}\big[ f_{\{\eta/2-a\},\{x\}} ,   g_{\{a\},\{x\}}\big].
\end{equation}
Here $f_{\{\eta/2-a\},\{x\}}$ is the function defined as in \eqref{special-f} in terms of the sets $\{\eta/2-a\}\equiv\{\eta/2-a_1,\ldots,\eta/2-a_{n_a}\}$ and $\{x\}\equiv\{x_1,\ldots,x_N\}$, and
\begin{equation}\label{special-g}
    g_{\{a\},\{x\}}(\lambda)=\delta_{n_a,4}\, \sinh( a_1+a_2+a_3+a_4 -\eta)\, 
    \prod_{\ell=1}^N\left[\varsigma(\lambda)-\varsigma(x_\ell)\right].
\end{equation}
\end{identity}

\begin{proof}
The proof goes along the same lines as for Identity~1 of \cite{KitMNT17}. We consider the matrices $\mathcal{C}^X$ and $\mathcal{C}^Z$ whose elements are defined respectively from the sets of variables $\{x\}$ and $\{z\}$ as:
\begin{align*}
    \prod_{\substack{\ell=1 \\ \ell\neq k}}^N\!\big(\varsigma(\lambda)-\varsigma(x_\ell)\big)
    =\sum_{j=1}^N\mathcal{C}_{j,k}^X\, [\varsigma(\lambda)]^{j-1},
    \qquad
    \prod_{\substack{\ell=1 \\ \ell\neq k}}^M\!\big(\varsigma(\lambda)-\varsigma(z_\ell)\big)
    =\sum_{j=1}^M\mathcal{C}_{j,k}^Z\, [\varsigma(\lambda)]^{j-1},
\end{align*}
and with respective determinants
\begin{equation}
     \det_N\left[\mathcal{C}^X\right]=\widehat{V}(x_N,\ldots,x_1),
     \qquad
     \det_M\left[\mathcal{C}^Z\right]=\widehat{V}(z_M,\ldots,z_1).
\end{equation}
For $M=N$, we can compute the product of the determinant of  $\mathcal{C}^Z$ with the determinant in the numerator of \eqref{Ater}
by using that
\begin{align}
   \sum_{j=1}^N \bar{f}^{(j)}_{\{a\},\{z\}}(x_i)
\, \mathcal{C}_{j,k}^Z
   &=\!\sum_{\epsilon\in\{+,-\}} \!\!  f_{\{a\},\{z\}}(\epsilon x_i)   \  \sum_{j=1}^N  \mathcal{C}_{j,k}^Z\; \varsigma\! \left(x_i+\epsilon\frac{\eta}{2}\right)^{j-1}
   \nonumber\\
   &
   = \prod_{\ell=1}^N\left[\varsigma(x_i)-\varsigma(z_\ell)\right]
   \!\sum_{\epsilon\in\{+,-\}} \frac{\epsilon\,\prod\limits_{\ell=1}^{n_a}\sinh(x_i+\epsilon a_\ell)}{\sinh(2x_i)\,[\varsigma(x_i+\epsilon\eta/2)-\varsigma(z_k)]}.
   \nonumber
\end{align}
Noticing that
\begin{multline}
  \!\sum_{\epsilon\in\{+,-\}} \frac{\epsilon\,\prod\limits_{\ell=1}^{n_a}\sinh(x_i+\epsilon a_\ell)}{\sinh(2x_i)\,[\varsigma(x_i+\epsilon\eta/2)-\varsigma(z_k)]}
  \\
  = \!\sum_{\epsilon\in\{+,-\}} \frac{\epsilon\,\prod\limits_{\ell=1}^{n_a}\sinh(z_k+\epsilon \eta/2-\epsilon a_\ell)}{\sinh(2 z_k)\,[\varsigma(z_k+\epsilon\eta/2)-\varsigma(x_i)]}
 +\delta_{n_a,4}\,c_{\{a\}},
\end{multline}
with $c_{\{a\}}= \sinh(a_1+a_2+a_3+a_4-\eta)$, we can rewrite the initial determinant as
\begin{equation}
  \det_{1\le i,j\le N}\left[ \bar{f}^{(j)}_{\{a\},\{z\}}(x_i) \right]
   =\frac{(-1)^N}{\widehat{V}(z_N,\ldots,z_1)}\,
      \lim_{\Lambda\to +\infty}\det_N \left[ \mathcal{B}_\Lambda\right],
\end{equation}
in terms of a matrix 
$ \mathcal{B}_\Lambda$ defined as
\begin{multline}
  \left[  \mathcal{B}_\Lambda \right]_{i,k}
   = \!\sum_{\epsilon\in\{+,-\}} \!\!\epsilon\,\frac{\prod\limits_{\ell=1}^{n_a}\sinh(z_i+\epsilon \frac{ \eta}{2}-\epsilon a_\ell)}{\sinh(2z_i)}
   \\
   \times
   \prod_{\ell=1}^N\frac{\varsigma(z_i)-\varsigma(x_\ell)}{\varsigma(z_i+\epsilon\frac{\eta}{2})-\varsigma(x_\ell)}\  \prod_{\substack{\ell=1 \\ \ell\neq k}}^N\left[\varsigma(z_i+\epsilon\frac{\eta}{2})-\varsigma(x_\ell)\right]
   \\
   +\delta_{n_a,4}\,c_{\{a\}}\, \varsigma(z_i+\Lambda)
   \prod_{\ell=1}^N\frac{\varsigma(z_i)-\varsigma(x_\ell)}{\varsigma(z_i+\Lambda)-\varsigma(x_\ell)}\  \prod_{\substack{\ell=1 \\ \ell\neq k}}^N\left[\varsigma(z_i+\Lambda)-\varsigma(x_\ell)\right].
\end{multline}
We can now easily factor the matrix $\mathcal{C}^X$ out of 
$ \mathcal{B}_\Lambda$, and taking the limit $\Lambda\to +\infty$, we obtain \eqref{id-A-B-1}.
\end{proof}

\subsection{The case $N<M$}

\begin{identity}\label{id-4-L>N}
Let  $N<M$. Then
\begin{equation}\label{id-A-B-2}
    \mathcal{A}_{\{x\}}\big[ f_{\{a\},\{z\}}\big]
    =(-1)^M\,
    \frac{\mathcal{A}_{\{z\}}\big[ f_{\{\eta/2-a\},\{x\}} ,  \tilde{g}^{(M)}_{\{a\},\{x\}}\big]}
           {\prod_{j=1}^{M-N}\sinh\big(  \sum_{\ell=1}^{n_a} a_\ell-j\eta\big)},
\end{equation}
with
\begin{multline}\label{gtilde-M}
   \tilde{g}^{(M)}_{\{a\},\{x\}}(\lambda)= \delta_{n_a,4}\, \Bigg\{ (-1)^{M-N}\sinh( a_1+a_2+a_3+a_4-\eta)
      \prod_{\ell=1}^N\left[\varsigma(\lambda)-\varsigma(x_\ell)\right]  \\
      -\bar{f}^{(M)}_{\{\eta/2-a\},\{x\}}(\lambda)
      \Bigg\},
\end{multline}
where $f_{\{\eta/2-a\},\{x\}}$ and $\bar{f}^{(M)}_{\{\eta/2-a\},\{x\}}$ are respectively defined as in \eqref{special-f} and \eqref{f-j} in terms of the sets $\{\eta/2-a\}\equiv\{\eta/2-a_1,\ldots,\eta/2-a_{n_a}\}$ and $\{x\}\equiv\{x_1,\ldots,x_N\}$.
\end{identity}

\begin{proof}
 Let us rewrite $\mathcal{A}_{\{x\}}\big[ f_{\{a\},\{z\}}\big]$
 as the following limit:
\begin{multline}
   \mathcal{A}_{\{x_1,\ldots,x_N\}}\big[ f_{\{a\},\{z_1,\ldots,z_M\}}\big]
   \\
   =\lim_{x_{N+1}\to +\infty}\ldots \lim_{x_M\to +\infty}
   \frac{\mathcal{A}_{\{x_1,\ldots,x_M\}}\big[ f_{\{a\},\{z_1,\ldots,z_M\}}\big]}
   {\prod_{i=N+1}^M\big[ \varsigma(x_i)^{\frac{n_a}{2}-1}\,\sinh(\sum_\ell a_\ell +(i-1-M)\eta)\big]}.
\end{multline}
By applying Identity~\ref{id-4-N=L} to 
$\mathcal{A}_{\{x_1,\ldots,x_M\}}\big[ f_{\{a\},\{z_1,\ldots,z_M\}}\big]$
and by computing the successive limits, we obtain \eqref{id-A-B-2}.
\end{proof}

\subsection{The case $M<N$}

In the case $n_a=2$, we obtain an identity which is the analog of Identity~2  of \cite{KitMNT17}, and which can be shown similarly:

\begin{identity}\label{id-2-M<N}
Let $n_a=2$ and $M<N$. Then
\begin{multline}\label{2-M<N}
   \mathcal{A}_{\{x_1,\ldots,x_N\}}[f_{\{a\},\{z_1,\ldots,z_M\}}]=
   (-1)^M \prod_{j=0}^{N-M-1}\sinh(a_1+a_2+j\eta) \\
   \times \mathcal{A}_{\{z_1,\ldots,z_M\}}[f_{\{\eta/2-a\},\{x_1,\ldots,x_N\}}].
\end{multline}
\end{identity}

The case $n_a=4$ is unfortunately much more complicated. In that case, we obtain the following result:

\begin{identity}\label{id-4-M<N}
Let $n_a=4$ and $M<N$. Then
\begin{multline}\label{4-M<N}
   \mathcal{A}_{\{x_1,\ldots,x_N\}}[f_{\{a\},\{z_1,\ldots,z_M\}}]=
   (-1)^M \prod_{j=0}^{N-M-1}\sinh\Big(j\eta+\sum_\ell a_\ell\Big)\\
   \times
    \mathcal{A}_{\{z_1,\ldots,z_M\}}[f_{\{\eta/2-a\},\{x\}}, \hat{g}^{(M)}_{\{a\},\{x\}}],
\end{multline}
where $\hat{g}^{(L)}_{\{a\},\{x\}}$ is defined by induction for $L\le N$ as
\begin{align}
  &\hat{g}^{(N)}_{\{a\},\{x\}}(z)=\sinh(a_1+a_2+a_3+a_4-\eta)\, \prod_{\ell=1}^N[\varsigma(z)-\varsigma(x_\ell)],
  \label{g-N} \\
  &\hat{g}^{(L)}_{\{a\},\{x\}}(z)=\frac{\bar{f}^{(L)}_{\{\eta/2-a\},\{x\}}(z)}{\sinh((L+1-N)\eta-\sum_\ell a_\ell)}\cdot 
  \lim_{z'\to\infty}\frac{\bar{f}^{(L+1)}_{\{\eta/2-a\},\{x\}}(z')+\hat{g}^{(L+1)}_{\{a\},\{x\}}(z')}{\varsigma(z')^{L}}
  \nonumber\\
  &\hspace{4.5cm}
    - \bar{f}^{(L)}_{\{\eta/2-a\},\{x\}}(z)  
   - \bar{f}^{(L+1)}_{\{\eta/2-a\},\{x\}}(z)-\hat{g}^{(L+1)}_{\{a\},\{x\}}(z).
   \label{rec-g}
\end{align}
in terms of the function $\bar{f}^{(L)}_{\{\eta/2-a\},\{x\}}$ defined as in \eqref{f-j} in terms of the sets $\{\eta/2-a\}\equiv\{\eta/2-a_1,\ldots,\eta/2-a_{n_a}\}$ and $\{x\}\equiv\{x_1,\ldots,x_N\}$.
\end{identity}

\begin{proof}
Since $M<N$, we can write
\begin{equation}
   f_{\{a\},\{z_1,\ldots,z_M\}}(x)
   =  \lim_{z_{M+1}\to\infty}\ldots \lim_{z_N\to\infty}
   f_{\{a\},\{z_1,\ldots,z_N\}}(x),
\end{equation}
so that
\begin{align}
   &\mathcal{A}_{\{x_1,\ldots,x_N\}}[f_{\{a\},\{z_1,\ldots,z_M\}}]
   = \lim_{z_{M+1}\to\infty}\ldots \lim_{z_N\to\infty} \mathcal{A}_{\{x_1,\ldots,x_N\}}[f_{\{a\},\{z_1,\ldots,z_N\}}]
   \nonumber\\
   &\qquad
   =  (-1)^N \lim_{z_{M+1}\to\infty}\ldots \lim_{z_N\to\infty} \mathcal{A}_{\{z_1,\ldots,z_N\}}[f_{\{\eta/2-a\},\{x_1,\ldots,x_N\}}, \hat{g}^{(N)}_{\{a\},\{x_1,\ldots,x_N\}}],
   \label{ident-lim}
\end{align}
where we have used Identity~\ref{id-4-N=L}. We want now to show by induction that the successive limits in \eqref{ident-lim} gives the right hand side of \eqref{4-M<N}. 

Let us first note that
\begin{equation}\label{lim-fj}
    \bar{f}^{(j)}_{\{\eta/2-a\},\{x\}}(z) 
    \underset{z\to\infty}{\sim} 
    \sinh\bigg((j+1-N)\eta-\sum_\ell a_\ell \bigg)\, [\varsigma(z)]^j,
\end{equation}
so that, in particular,
\begin{equation}\label{f_N+g_N}
   \bar{f}^{(N)}_{\{\eta/2-a\},\{x\}}(z) +\hat{g}^{(N)}_{\{a\},\{x\}}(z)=O\left(\varsigma(z)^{N-1}\right).
\end{equation}

Let us now suppose that, for some $L\le N$, $\hat{g}^{(L)}$ is well defined by the above recursion \eqref{g-N}-\eqref{rec-g}, that
\begin{multline}
    \lim_{z_{L+1}\to\infty}\ldots \lim_{z_N\to\infty} \mathcal{A}_{\{z_1,\ldots,z_N\}}[f_{\{\eta/2-a\},\{x\}}, \hat{g}^{(N)}_{\{a\},\{x\}}]
    =(-1)^{N-L} \\
   \times
   \prod_{j=0}^{N-L-1} \sinh\Big(j\eta+\sum_\ell a_\ell\Big)\
   \mathcal{A}_{\{z_1,\ldots,z_{L}\}}[f_{\{\eta/2-a\},\{x\}}, \hat{g}^{({L})}_{\{a\},\{x\}}],
\end{multline}
and that $\bar{f}^{(L)}_{\{\eta/2-a\},\{x\}}(z)+\hat{g}^{(L)}_{\{a\},\{x\}}(z)=O\left(\varsigma(z)^{L-1}\right)$.
Writing explicitly
\begin{equation}
   \mathcal{A}_{\{z_1,\ldots,z_{L}\}}[f_{\{\eta/2-a\},\{x\}}, \hat{g}^{({L})}_{\{a\},\{x\}}]
   =\frac{\det_{1\le i,j\le L}\left[ \bar{f}^{(j)}_{\{\eta/2-a\},\{x\}}(z_i)+\delta_{j,L}\, \hat{g}^{({L})}_{\{a\},\{x\}}(z_i)\right]}{\widehat{V}(z_1,\ldots,z_L)},
\end{equation}
and decomposing the determinant in the numerator with respect to the last column, we get
\begin{multline}
   \mathcal{A}_{\{z_1,\ldots,z_{L}\}}[f_{\{\eta/2-a\},\{x\}}, \hat{g}^{({L})}_{\{a\},\{x\}}]
   =\sum_{l=1}^L\frac{(-1)^{L+l}}{\widehat{V}(z_1,\ldots,z_L)}\,
    \det_{\substack{i\not= l\\ j\not= L}}\left[  \bar{f}^{(j)}_{\{\eta/2-a\},\{x\}}(z_i)\right]
   \\
   \times
   \left[ \bar{f}^{(L)}_{\{\eta/2-a\},\{x\}}(z_l) + \hat{g}^{({L})}_{\{a\},\{x\}}(z_l) \right].
\end{multline}
Let us remark that, in each of the first $L-1$ terms, the only dependence in $z_L$ is contained in the determinant, whereas in the last term $l=L$ it is contained in the last factor $ \left[ \bar{f}^{(L)}_{\{\eta/2-a\},\{x\}}(z_L) + \hat{g}^{({L})}_{\{a\},\{x\}}(z_L) \right]$. Taking the respective limits $z_L\to \infty$ in all these terms, we obtain, 
\begin{multline}\label{det-dec}
   \mathcal{A}_{\{z_1,\ldots,z_{L}\}}[f_{\{\eta/2-a\},\{x\}}, \hat{g}^{({L})}_{\{a\},\{x\}}]
      \underset{z_L\to\infty}{\longrightarrow} 
     \sum_{l=1}^{L-1} (-1)^{L+l}
     \sinh\bigg((L-N)\eta-\sum_\ell a_\ell \bigg)
    \\
   \times
    \frac{\det_{\substack{i\not= l, i<L\\ j< L-1}}\left[  \bar{f}^{(j)}_{\{\eta/2-a\},\{x\}}(z_i)\right]}{\widehat{V}(z_1,\ldots,z_{L-1})} \,
    \left[ \bar{f}^{(L)}_{\{\eta/2-a\},\{x\}}(z_l) + \hat{g}^{({L})}_{\{a\},\{x\}}(z_l) \right] 
   \\
   +\frac{\det_{\substack{i<L\\ j< L}}\left[  \bar{f}^{(j)}_{\{\eta/2-a\},\{x\}}(z_i)\right]}{\widehat{V}(z_1,\ldots,z_{L-1})} \,
    \lim_{z'\to\infty}\frac{\bar{f}^{(L)}_{\{\eta/2-a\},\{x\}}(z')+\hat{g}^{(L)}_{\{a\},\{x\}}(z')}{\varsigma(z')^{L-1}},
\end{multline}
in which we have used  \eqref{lim-fj} for taking the limit in the first $L-1$ terms, and the fact that the remaining limit is well defined in the last term.
Hence, recomposing the determinant in \eqref{det-dec}, we obtain that
\begin{multline}
    \lim_{z_{L}\to\infty} \mathcal{A}_{\{z_1,\ldots,z_{L}\}}[f_{\{\eta/2-a\},\{x\}}, \hat{g}^{({L})}_{\{a\},\{x\}}]
    = -\sinh\bigg((N-L)\eta+\sum_\ell a_\ell \bigg)\\
    \times
    \mathcal{A}_{\{z_1,\ldots,z_{L}\}}[f_{\{\eta/2-a\},\{x\}}, \hat{g}^{({L-1})}_{\{a\},\{x\}}]
\end{multline}
with
\begin{multline}
   \bar{f}^{(L-1)}_{\{\eta/2-a\},\{x\}}(z)+\hat{g}^{(L-1)}_{\{a\},\{x\}}(z)
   =-\bar{f}^{(L)}_{\{\eta/2-a\},\{x\}}(z)-\hat{g}^{(L)}_{\{a\},\{x\}}(z)
   \\
   +\frac{\bar{f}^{(L-1)}_{\{\eta/2-a\},\{x\}}(z)}{\sinh((L-N)\eta-\sum_\ell a_\ell)}\
  \lim_{z'\to\infty}\frac{\bar{f}^{(L)}_{\{\eta/2-a\},\{x\}}(z')+\hat{g}^{(L)}_{\{a\},\{x\}}(z')}{\zeta(z')^{L-1}}.
\end{multline}
It remains to notice that, by construction
\begin{equation}
  \bar{f}^{(L-1)}_{\{\eta/2-a\},\{x\}}(z)+\hat{g}^{(L-1)}_{\{a\},\{x\}}(z)=O\left(\varsigma(z)^{L-2}\right),
\end{equation}
which ends the proof of the recursion.
\end{proof}

\section{Determinant identities: transformation into generalized Slavnov  determinants}
\label{app-Slavnov}

In this appendix, we explain how to transform quantities of the form \eqref{def-Afg}, for two arbitrary functions $f$ and $g$ and a set of arbitrary parameters $\{z_1,\ldots,z_N\}\equiv\{x_1,\ldots,x_{L_1}\}\cup\{y_1,\ldots,y_{L_2}\}$, into some generalization of the Slavnov determinant \cite{Sla89}.

Throughout this appendix we will use  the following shortcut notations:
\begin{align}
      &X(\lambda)=\prod_{\ell=1}^{L_1}\big[\varsigma(\la)-\varsigma(x_\ell)\big],\\
      &X_k(\lambda)=\prod_{\ell\neq k}\big[\varsigma(\la)-\varsigma(x_\ell)\big],
\end{align}
\begin{equation}
    \varphi_{\{x\}}(\lambda)=\frac{\sinh(2\lambda-\eta)}{\sinh(2\lambda+\eta)}
    \frac{X(\lambda+\eta)}{X(\lambda-\eta)},
\end{equation}
and
\begin{align}
   X^{g}_{f,k} =& - \frac{g(x_k)}{f(-x_k)\,\sinh(2x_k)\,\sinh\eta\,X_k(x_k)\, X_k(x_k-\eta)},\\
   =& \frac{g(x_k)\,\sinh(2x_k-\eta)}{f(-x_k)\,X'(x_k)\, X(x_k-\eta)},
\end{align}
with $\varsigma(\lambda)$ given as in \eqref{short-notation}.

\subsection{A simple case: $L_1=L_2$ with one on-shell set of parameters}

\begin{identity}
\label{slavnov_square}
We suppose that $L_1=L_2\equiv L$, and that the parameters  $x_1,\ldots,x_L$  are on-shell, i.e. that they satisfy the equations:
\begin{equation}
\label{Bethe_withf}
  f(-x_k)-f(x_k)\, \varphi_{\{x\}}(x_k)=0,\quad k=1,\dots,L.
\end{equation}
Then
\begin{multline}
    \mathcal{A}_{\{x\}\cup\{y\}}[f, g]
  = \pl_{j=1}^L\Big(\sinh\eta\,f(-x_j)\sinh2x_j\Big)\,\,\frac{\widehat{V}(x_1-\frac\eta2,\dots, x_L-\frac\eta2)}{\widehat{V}(x_1+\frac\eta2,\dots, x_L+\frac\eta2)}
  \\
  \times
  \left( 1+\sum_{k=1}^LX^{g}_{f,k}\right)\ 
  \frac{\det_{L}\left[\sum_{\epsilon\in\{+,-\}} 
      f(\epsilon y_i)\, 
      \frac{X_k(y_i+\epsilon\eta)}{\varsigma(y_i)-\varsigma(x_k)} \right]}
     {\widehat{V}(x_L,\ldots,x_1)\, \widehat{V}(y_1,\ldots,y_L)} .
     \label{id_slavnov_square}
 \end{multline}

\end{identity}

\begin{proof}
 We introduce an auxiliary $2L\times 2L$ matrix $\widetilde{\mathcal{D}}$ with coefficients $\widetilde{\mathcal{D}}_{j,k}$ given by the following relations:
\begin{alignat}{2}
   &X_k^{(+)}(\lambda)\ X_k^{(-)}(\lambda)=\sum_{j=1}^{2L} \widetilde{\mathcal{D}}_{j,k}\ \varsigma(\lambda)^{j-1},
    & \qquad &1\le k\le \! L,  \nonumber\\
   &X^{(+)}(\lambda)\ X_k^{(-)}(\lambda)=\sum_{j=1}^{2L} \widetilde{\mathcal{D}}_{j, L+k}\ \varsigma(\lambda)^{j-1},    
     & \qquad &1\le k\le \! L.
\end{alignat}
Here we have defined $X^{(\pm)}(\lambda)$,  $X_k^{(\pm)}(\lambda)$ as the following polynomials  in $\varsigma(\lambda)$:
\begin{align}
     &X^{(\pm)}(\lambda)=\prod_{\ell=1}^L\big[\varsigma(\la)-\varsigma(x_\ell\pm \eta/2)\big],
   \label{X_pm}  \\  \label{X_k_pm}
     &X^{(\pm)}_k(\lambda)=\prod_{\substack{\ell=1 \\ \ell\not= k}}^L\big[\varsigma(\la)-\varsigma(x_\ell\pm \eta/2)\big],
     \quad 1\le k\le L.
     \end{align}
%

The determinant of this matrix can be computed as in \cite{KitMNT16,KitMNT17}. We obtain
\begin{equation}
\det_{2L}  \widetilde{ \mathcal{D}}
= \widehat{V}\Big(x_L+\frac{\eta}{2},\ldots,x_1+\frac{\eta}{2}\Big)\, \widehat{V}\Big(x_L-\frac{\eta}{2},\ldots,x_1-\frac{\eta}{2}\Big)\,\pl_{k=1}^L  X_k^{(-)}\Big(x_k+\frac{\eta}{2}\Big). 
\end{equation}

Computing the product of $\mathcal{A}_{\{x\}\cup\{y\}}[f, g]$ with the determinant of the matrix $\widetilde{ \mathcal{D}}$, we obtain
 \begin{equation}
    \mathcal{A}_{\{x\}\cup\{y\}}[f, g] \cdot  \det_{2L} \widetilde{\mathcal{D}}
  = \frac{\det_{2L}\mathcal{G}}{\widehat{V}(x_1,\ldots,x_L,y_1,\ldots,y_L)},
\end{equation}
where $\mathcal{G}$ is a block matrix:
\begin{equation}\label{mat-G-blocs}
   \mathcal{G}=   \det_{2L}\begin{pmatrix}
 \mathcal{G}^{(1,1)}&\mathcal{G}^{(1,2)}\\
 \mathcal{G}^{(2,1)}&\mathcal{G}^{(2,2)}
       \end{pmatrix}.
\end{equation}

The blocks  $ \mathcal{G}^{(a,b)}$ in \eqref{mat-G-blocs} are $L\times L$ matrices.
More precisely, the first block $ \mathcal{G}^{(1,1)}$ is the following diagonal matrix: 
\begin{align}
  \mathcal{G}^{(1,1)}_{i,k}
  = \delta_{i,k} \Bigg(  f(- x_i)\    & X_i^{(+)}\Big(x_i-\frac{\eta}{2}\Big)\ X_i^{(-)}\Big(x_i-\frac{\eta}{2}\Big)  \nonumber\\
 & + f( x_i)\
     X_i^{(+)}\Big(x_i+\frac{\eta}{2}\Big)\ X_i^{(-)}\Big(x_i+\frac{\eta}{2}\Big)\Bigg).
\end{align}
Using simple trigonometric relation,
\begin{equation*}
\big[\varsigma(\la\pm \eta/2)-\varsigma(x_\ell+ \eta/2)\big]\big[\varsigma(\la\pm \eta/2)-\varsigma(x_\ell- \eta/2)\big]
=\big[\varsigma(\la)-\varsigma(x_\ell)\big]\big[\varsigma(\la\pm \eta)-\varsigma(x_\ell)\big],
\end{equation*}
we can rewrite this matrix as
\begin{equation*}
  \mathcal{G}^{(1,1)}_{i,k}
  = -\delta_{i,k} X_i(x_i)\left(  f(- x_i)\     \frac{X(x_i-\eta)}{\sinh (2x_i-\eta)\sinh\eta} 
  - f( x_i)\    \frac{X(x_i+\eta)}{\sinh (2x_i+\eta)\sinh\eta} 
   \right),
\end{equation*}
which vanishes due to the equations \eqref{Bethe_withf}. 
It means in particular that the explicit form of the matrix  $ \mathcal{G}^{(2,2)}$ is irrelevant and that
\begin{equation}
\det_{2L}\mathcal{G}=(-1)^L\,\det_L \mathcal{G}^{(1,2)}\, \det_L  \mathcal{G}^{(2,1)}.
\end{equation}
The matrix $ \mathcal{G}^{(1,2)}$ is  a diagonal matrix with a rank 1 addition:
\begin{equation}
 \mathcal{G}^{(1,2)}_{i,k}
  = \delta_{i,k} \,  f(- x_i)\,     X^{(+)}\Big(x_i-\frac{\eta}{2}\Big)\, X_i^{(-)}\Big(x_i-\frac{\eta}{2}\Big) +g(x_i). 
\end{equation}
Its determinant can easily be computed:
\begin{align}
\label{det_with_g}
\det_L \mathcal{G}^{(1,2)}_{i,k}
  = &\pl_{k=1}^L   \Big(-\sinh\eta\ \sinh2x_k\ f(- x_k)\     X_k(x_k)\ X_k(x_k-\eta)\Big)\nonumber\\
  &\times\left( 1-\sum_{j=1}^L\frac{g(x_j)}{f(-x_j)\sinh 2x_j\sinh\eta\, X_j(x_j)X_j(x_j-\eta)}\right). 
\end{align}
Finally, the matrix  $ \mathcal{G}^{(2,1)}$ is a Slavnov-type matrix:
\begin{align}
  \mathcal{G}^{(2,1)}_{i,k}
  =& \sum_{\epsilon\in\{+,-\}} \!\!
      f(\epsilon y_i)\  X_k^{(+)} \Big(y_i+\epsilon\frac\eta 2\Big)X_k^{(-)} \Big(y_i+\epsilon\frac\eta 2\Big)\nonumber\\
      =& \sum_{\epsilon\in\{+,-\}} \!\!
      f(\epsilon y_i)\  X_k (y_i)X_k (y_i+\epsilon\eta ).
\end{align}
Now combining all the terms together we obtain the identity \eqref{id_slavnov_square}.
\end{proof}

%
%
 
Let us mention that, if $f$ coincides with the function $f_{\boldsymbol{\varepsilon},\boldsymbol{\varepsilon}}$ \eqref{feps}, then the system of equations \eqref{Bethe_withf} coincides with the Bethe equations resulting from the homogeneous T-Q functional equation \eqref{hom-TQ}.

We would also like to stress that in \eqref{id_slavnov_square} the function $g(\la)$ appears only in the irrelevant normalization coefficients. 

\subsection{The case of two arbitrary sets of parameters}\label{SP-general case}

Let us now turn to the most general situation in which the two sets of variables, as well as the functions $f$ and $g$, are arbitrary. We will distinguish two different cases according to whether the cardinality of the two sets of variables are equal ($L_1=L_2\equiv L$) or different ($L_1<L_2$).

\begin{identity}\label{id-Slav_gen=}
Let $L_1=L_2\equiv L$. Then
\begin{multline}
    \mathcal{A}_{\{x\}\cup\{y\}}[f, g]
  =  \frac{\widehat{V}(x_1-\frac\eta2,\dots, x_L-\frac\eta2)}{\widehat{V}(x_1+\frac\eta2,\dots, x_L+\frac\eta2)}\ 
  \left( 1+\sum_{k=1}^LX^{g}_{f,k}\right)
  \\
  \times
  \frac{\det_{L}  \bar{\mathcal{S}}_{\mathbf{x},\mathbf{y}}[f,g]}
  {\widehat{V}(x_L,\ldots,x_1)\, \widehat{V}(y_1,\ldots,y_L)},
  \label{id_slavnov_gen_square}
\end{multline}
where the $L\times L$ matrix  $\bar{\mathcal{S}}_{\mathbf{x},\mathbf{y}}$ is given by
\begin{multline}\label{mat-Slav-N=L}
   \big[\, \bar{\mathcal{S}}_{\mathbf{x},\mathbf{y}}[f,g]\, \big]_{i,k}
   =\!\!\sum_{\epsilon\in\{+,-\}} \!\!
     f(\epsilon y_i)
     \  X(y_i+\epsilon\eta)\
     \Bigg[ \frac{f(-x_k)}{\varsigma(y_i+\epsilon\frac{\eta}{2})-\varsigma(x_k+\frac{\eta}{2})}
     \\
     -
     \frac{f(x_k)\, \varphi_{\{x\}}(x_k)}
     {\varsigma(y_i+\epsilon\frac{\eta}{2})-\varsigma(x_k-\frac{\eta}{2})}
     +\frac{f(-x_k)-f(x_k)\, \varphi_{\{x\}}(x_k)}{1+\sum_{\ell=1}^L X^{g}_{f,\ell}}
     \sum_{j=1}^L\frac{X^{g}_{f,j}}{\varsigma(y_i+\epsilon\frac{\eta}{2})-\varsigma(x_j-\frac{\eta}{2})}
     \Bigg]\\
     +
     \frac{g(y_i)}{X(y_i)}
     \frac{f(-x_k)-f(x_k)\, \varphi_{\{x\}}(x_k)}{1+\sum_{\ell=1}^L X^{g}_{f,\ell}}.
\end{multline}
%
%
%
\end{identity}

\begin{rem}
All the extra terms with respect to Identity \ref{slavnov_square} are proportional to the quantity $f(-x_k)-f(x_k)\, \varphi_{\{x\}}(x_k)$ which vanishes if the equations \eqref{Bethe_withf} are satisfied. 
\end{rem}

\begin{proof}
The proof follows the same lines as the proof of Identity \ref{slavnov_square} with a slightly different auxiliary $(2L)\times (2L)$ matrix $\widetilde{\mathcal{C}}$ with coefficients $\widetilde{\mathcal{C}}_{j,k}$ given by the following relations:
\begin{alignat}{2}
   &X^{(+)}(\lambda)\ X_k^{(-)}(\lambda)=\sum_{j=1}^{2L} \widetilde{\mathcal{C}}_{j,k}\ \varsigma(\lambda)^{j-1},
    & \qquad &1\le k\le \! L,  \nonumber\\
   &X_k^{(+)}(\lambda)\ X^{(-)}(\lambda)=\sum_{j=1}^{2L} \widetilde{\mathcal{C}}_{j, L+k}\ \varsigma(\lambda)^{j-1},    
     & \qquad &1\le k\le \! L.
\end{alignat}
Here the polynomials $X^{(\pm)}(\lambda)$ and  $X_k^{(\pm)}(\lambda)$ are defined by (\ref{X_pm}) and (\ref{X_k_pm}). %
The determinant of the matrix $\widetilde{\mathcal{C}}$ can easily be computed:
\begin{equation}
\det_{2L}  \widetilde{ \mathcal{C}}
= \widehat{V}\Big(x_L+\frac{\eta}{2},\ldots,x_1+\frac{\eta}{2}\Big)\, \widehat{V}\Big(x_L-\frac{\eta}{2},\ldots,x_1-\frac{\eta}{2}\Big)\,\pl_{k=1}^L  X^{(+)}\Big(x_k-\frac{\eta}{2}\Big). 
\end{equation}
Similarly as in  the proof of Identity~\ref{slavnov_square}, we can write
 \begin{equation}
    \mathcal{A}_{\{x\}\cup\{y\}}[f, g] \cdot  \det_{2L} \widetilde{\mathcal{C}}
  = \frac{\det_{2L}\mathcal{G}}{\widehat{V}(x_1,\ldots,x_L,y_1,\ldots,y_L)},
\end{equation}
with
\begin{equation}
   \mathcal{G}=\begin{pmatrix}
 \widetilde{ \mathcal{G}}^{(1,1)}& \widetilde{\mathcal{G}}^{(1,2)}\\
 \widetilde{ \mathcal{G}}^{(2,1)}& \widetilde{\mathcal{G}}^{(2,2)}
 \end{pmatrix}.
\end{equation}
The matrices  $\widetilde{\mathcal{G}}^{(a,b)}$ can here be written in the following form:
\begin{equation*}
    \widetilde{\mathcal{G}}_{i,k}^{(1,b)}=\,\mathcal{G}^{(1,b)}_{i,k} + g(x_i),
        \qquad\ \widetilde{\mathcal{G}}_{i,k}^{(2,b)}=\,\mathcal{G}^{(2,b)}_{i,k} + g(y_i), \quad b=1,2,
\end{equation*}
where  $ \mathcal{G}^{(1,1)}$ and $\mathcal{G}^{(1,2)}$ are  $L\times L$ diagonal matrices with elements
\begin{align}
  &\mathcal{G}^{(1,1)}_{i,k}
  = \delta_{i,k}\  f(- x_i)\     X^{(+)}\Big(x_i-\frac{\eta}{2}\Big)\ X_i^{(-)}\Big(x_i-\frac{\eta}{2}\Big),\label{G11}
  \\
  &\mathcal{G}^{(1,2)}_{i,k}
  =\delta_{i,k}\  f( x_i)\
     X_i^{(+)}\Big(x_i+\frac{\eta}{2}\Big)\ X^{(-)}\Big(x_i+\frac{\eta}{2}\Big),
     \label{G12}
\end{align}
and $\mathcal{G}^{(2,1)}$ and $\mathcal{G}^{(2,2)}$ are $L\times L$ matrices with elements
\begin{align}
   &\mathcal{G}^{(2,1)}_{i,k}
   =   \!\sum_{\epsilon\in\{+,-\}}\!\!
      f(\epsilon y_i)\
     \frac{X^{(+)}(y_i+\epsilon\frac{\eta}{2})\ X^{(-)}(y_i+\epsilon\frac{\eta}{2})}
             {\varsigma(y_i+\epsilon\frac{\eta}{2})-\varsigma(x_k-\frac{\eta}{2})},
\label{G21}    
    \\
   &\mathcal{G}^{(2,2)}_{i,k}
   =
     \!\sum_{\epsilon\in\{+,-\}}\!\!
      f(\epsilon y_i)\
     \frac{X^{(+)}(y_i+\epsilon\frac{\eta}{2})\ X^{(-)}(y_i+\epsilon\frac{\eta}{2})}
             {\varsigma(y_i+\epsilon\frac{\eta}{2})-\varsigma(x_k+\frac{\eta}{2})}.
   \label{G22}
\end{align}
%
%
We now use the following formula for the determinant of block matrix:
\begin{equation}
  \det_{2L}\begin{pmatrix}
 \widetilde{ \mathcal{G}}^{(1,1)}& \widetilde{\mathcal{G}}^{(1,2)}\\
      \widetilde{\mathcal{G}}^{(2,1)}& \widetilde{\mathcal{G}}^{(2,2)}
 \end{pmatrix}= \det_L  \widetilde{\mathcal{G}}^{(1,1)} \det_{L}\left(  \widetilde{\mathcal{G}}^{(2,2)} - \widetilde{\mathcal{G}}^{(2,1)}\left.{\widetilde{\mathcal{G}}^{(1,1)}}\right.^{-1} \widetilde{\mathcal{G}}^{(1,2)}\right).
\end{equation}
Since   $ \widetilde{\mathcal{G}}^{(1,1)}$ is the sum of a diagonal invertible  matrix with a rank 1 matrix it is possible to compute its determinant
\begin{equation}
\det  \widetilde{\mathcal{G}}^{(1,1)}=\left(1+\sum_{j=1}^L\left(\mathcal{G}^{(1,1)}_{j,j}\right)^{-1}\, g(x_j)\right)\,\det \mathcal{G}^{(1,1)},
\end{equation}
and its inverse   by the Sherman-Morrison formula:
\begin{equation}
   \left( \widetilde{\mathcal{G}}^{(1,1)}\right)^{-1}_{i,k}= \delta_{i,k}\left( \mathcal{G}^{(1,1)}_{i,i}\right)^{-1}
   -\frac{\left(\mathcal{G}^{(1,1)}_{i,i}\right)^{-1}\, g(x_i)\, \left(\mathcal{G}^{(1,1)}_{k,k}\right)^{-1}}{1+\sum_{j=1}^L\left(\mathcal{G}^{(1,1)}_{j,j}\right)^{-1}\, g(x_j)}.
\end{equation}
Then after straightforward but cumbersome computations we obtain the expression \eqref{id_slavnov_gen_square}.
\end{proof}

\begin{identity}\label{id-Slav_gen>}
Let $L_1<L_2$. Then
\begin{equation}
    \mathcal{A}_{\{x\}\cup\{y\}}[f, g]
  =  \frac{\widehat{V}(x_1-\frac\eta2,\dots, x_{L_1}-\frac\eta2)}{\widehat{V}(x_1+\frac\eta2,\dots, x_{L_1}+\frac\eta2)}\
    \frac{\det_{L_2}  \widetilde{\mathcal{S}}_{\mathbf{x},\mathbf{y}}[f,g]}
  {\widehat{V}(x_{L_1},\ldots,x_1)\, \widehat{V}(y_1,\ldots,y_{L_2})},
\label{id_slavnov_gen_rect}
\end{equation}
where
\begin{multline}
   \big[\, \widetilde{\mathcal{S}}_{\mathbf{x},\mathbf{y}}[f,g]\, \big]_{i,k}
   =\!\!\sum_{\epsilon\in\{+,-\}} \!\!
      f(\epsilon y_i)\ 
     X(y_i+\epsilon\eta)
     \\
     \times
     \left[ \frac{f(-x_k)}{\varsigma(y_i+\epsilon\frac{\eta}{2})-\varsigma(x_k+\frac{\eta}{2})}
     -
     \frac{f(x_k)\, \varphi_{\{x\}}(x_k)}
     {\varsigma(y_i+\epsilon\frac{\eta}{2})-\varsigma(x_k-\frac{\eta}{2})}\right]
     \quad \text{if } k\le L_1,
\end{multline}
and
\begin{multline}
   \big[\, \widetilde{\mathcal{S}}_{\mathbf{x},\mathbf{y}}[f,g]\, \big]_{i,k}
   =\!\!\sum_{\epsilon\in\{+,-\}} \!\!
      f(\epsilon y_i)\ 
     X(y_i+\epsilon\eta)\  
     \Bigg\{ \varsigma(y_i+\epsilon{\eta}/{2})^{k-L_1-1}
     \\
     -\delta_{k,L_2}\, \sum_{j=1}^{L_1}\frac{X^{g}_{f,j}}{\varsigma(y_i+\epsilon\frac{\eta}{2})-\varsigma(x_j-\frac{\eta}{2})} \Bigg\}
     +\delta_{k,L_2}\,  \frac{g(y_i)}{X(y_i)}
     \qquad \text{if } k >L_1.
\end{multline}

\end{identity}

\begin{proof}
Once again, the proof follows the lines of \cite{KitMNT17}. As before we introduce an auxiliary $(L_1+L_2)\times (L_1+L_2)$ matrix $\widetilde{\mathcal{C}}$ with coefficients $\widetilde{\mathcal{C}}_{j,k}$ given by the following relations:
\begin{alignat}{2}
   &X^{(+)}(\lambda)\ X_k^{(-)}(\lambda)=\sum_{j=1}^{L_1+L_2} \widetilde{\mathcal{C}}_{j,k}\ \varsigma(\lambda)^{j-1},
    & \qquad &1\le k\le \! L_1,  \nonumber\\
   &X_k^{(+)}(\lambda)\ X^{(-)}(\lambda)=\sum_{j=1}^{L_1+L_2} \widetilde{\mathcal{C}}_{j, L_1+k}\ \varsigma(\lambda)^{j-1},    
     & \qquad &1\le k\le \! L_1, \nonumber\\
   &X^{(+)}(\lambda)\ X^{(-)}(\lambda)\ W_k(\lambda)=\sum_{j=1}^{L_1+L_2} \widetilde{\mathcal{C}}_{j,2L_1+k}\ \varsigma(\lambda)^{j-1},
   & \qquad &1\le k\le L_2-L_1.
\end{alignat}
Here the polynomials $X^{(\pm)}(\lambda)$   are defined by (\ref{X_pm}) while $X_k^{(\pm)}(\lambda)$ by (\ref{X_k_pm}) and $W_k(\lambda)$ is the following polynomials  in $\varsigma(\lambda)$:
\begin{align}
     &
     W_k(\lambda)=\prod_{\substack{\ell=1 \\ \ell\not= k}}^S\big[\varsigma(\lambda)-\varsigma(w_\ell)\big],
     \quad 1\le k\le S=L_2-L_1,
\end{align}
where $w_1,\ldots,w_S$  are arbitrary pairwise distinct auxiliary variables.
Note in particular that, for $k\le 2L_1$,  $\widetilde{\mathcal{C}}_{j,k}=0$ if $j>2L_1$ and $\widetilde{\mathcal{C}}_{2L_1,k}=1$. Note also that $\widetilde{\mathcal{C}}_{L_1+L_2,k}=1$ for $k>2L_1$.
The determinant of this matrix can be computed as in \cite{KitMNT16,KitMNT17}. We obtain
\begin{equation}
\det_{L_1+L_2}  \widetilde{ \mathcal{C}}
= \widehat{V}(w_S,\ldots,w_1)\, 
   \widehat{V}\Big(x_{L_1}+\frac{\eta}{2},\ldots,x_1+\frac{\eta}{2},x_{L_1}-\frac{\eta}{2},\ldots,x_1-\frac{\eta}{2}\Big). 
\end{equation}

Computing the product of $  \mathcal{A}_{\{x\}\cup\{y\}}[f, g] $
with the determinant of the matrix $\widetilde{ \mathcal{C}}$, we obtain:
 \begin{equation}
      \mathcal{A}_{\{x\}\cup\{y\}}[f, g] 
 \cdot  \det_{L_1+L_2} \widetilde{\mathcal{C}}
  = \frac{\det_{L_1+L_2}\mathcal{G}}{\widehat{V}(x_1,\ldots,x_{L_1},y_1,\ldots,y_{L_2})} ,
 \end{equation}
 where  $\mathcal{G}$ is given as the following block matrix:
 \begin{equation}
 \mathcal{G}
 =\begin{pmatrix}
 \mathcal{G}^{(1,1)}&\mathcal{G}^{(1,2)}&\mathcal{G}^{(1,3)}\\
 \mathcal{G}^{(2,1)}&\mathcal{G}^{(2,2)}&\mathcal{G}^{(2,3)}
 \end{pmatrix}.
 \end{equation}
In this expression, $ \mathcal{G}^{(1,1)}$ and $\mathcal{G}^{(1,2)}$ are  $L_1\times L_1$ matrices with elements given by \eqref{G11} and \eqref{G12}
for $1\le i,k\le L_1$, whereas $\mathcal{G}^{(2,1)}$ and $\mathcal{G}^{(2,2)}$ are $L_2\times L_1$ matrices with elements given by \eqref{G21} and \eqref{G22}
%
for $1\le i\le L_2$ and $1\le k\le L_1$.
Finally,  we have extra blocks $\mathcal{G}^{(1,3)}$ and $\mathcal{G}^{(2,3)}$.
$\mathcal{G}^{(1,3)}$ is a $L_1\times S$ rank 1 matrix with elements:
\begin{align}
  \mathcal{G}^{(1,3)}_{i,k}= g(x_i) = \sum_{j=1}^S \mathcal{C}_{j,k}^W\, \big\{ \delta_{j,S}\, g(x_i) \big\},
\end{align}
for $1\le i\le L_1$ and $1\le k \le S$, whereas $\mathcal{G}^{(2,3)}$ is a $L_2\times S$ matrix with elements
\begin{align}
     \mathcal{G}^{(2,3)}_{i,k}
    &=
     \!\sum_{\epsilon\in\{+,-\}}\!\!
      f(\epsilon y_i)\      X^{(+)}\Big(y_i+\epsilon\frac{\eta}{2}\Big)\ X^{(-)}\Big(y_i+\epsilon\frac{\eta}{2}\Big)\
     W_k\Big(y_i+\epsilon\frac{\eta}{2}\Big)
     + g(y_i)
     \nonumber\\
    &=\sum_{j=1}^S \mathcal{C}_{j,k}^W \  \Bigg\{
     \sum_{\epsilon\in\{+,-\}} \!\!
     f(\epsilon y_i)\,     X^{(+)}\Big(y_i+\epsilon\frac{\eta}{2}\Big)\, X^{(-)}\Big(y_i+\epsilon\frac{\eta}{2}\Big)\
      \varsigma(y_i+\epsilon\eta/2)^{ j-1}
      \nonumber\\
      &\hspace{5cm}
     + \delta_{j,S}\, g(y_i)\Bigg\}.
\end{align}
In these expressions, $\mathcal{C}^W$ is the $S\times S$ matrix with determinant $\widehat{V}(w_S,\ldots,w_1)$ defined from the set of variables $\{w_1,\ldots,w_S\}$ by the relations
\begin{equation}
W_k(\lambda)=\sul_{j=1}^S \mathcal{C}^W_{j,k}\ \varsigma(\lambda)^{j-1}.
\end{equation}
Hence we obtain
\begin{equation}\label{det-G-1}
    \det_{L_1+L_2}\mathcal{G}
    =\det_{L_1+L_2}\begin{pmatrix}
 \mathcal{G}^{(1,1)}&\mathcal{G}^{(1,2)}& \widetilde{\mathcal{G}}^{(1,3)}\\
 \mathcal{G}^{(2,1)}&\mathcal{G}^{(2,2)}&\widetilde{\mathcal{G}}^{(2,3)}
       \end{pmatrix}
       \cdot \widehat{V}(w_S,\ldots,w_1),
\end{equation}
where the elements of the $L_1\times S$ block  $\widetilde{\mathcal{G}}^{(1,3)}$ and of the $L_2\times S$ block $\widetilde{\mathcal{G}}^{(2,3)}$ are respectively given by
\begin{align*}
 &\widetilde{ \mathcal{G}}^{(1,3)}_{i,k}= \delta_{k,S}\, g(x_i),
  \\
 &\widetilde{ \mathcal{G}}^{(2,3)}_{i,k}=\!\!\sum_{\epsilon\in\{+,-\}} \!\!
      f(\epsilon y_i)\,
        X^{(+)}\Big(y_i+\epsilon\frac{\eta}{2}\Big)\, X^{(-)}\Big(y_i+\epsilon\frac{\eta}{2}\Big)\
     \varsigma(y_i+\epsilon {\eta}/{2})^{k-1}
     + \delta_{k,S}\, g(y_i).
\end{align*}
%

 It is easy to compute by blocks the determinant in \eqref{det-G-1} using the fact that $\mathcal{G}^{(1,1)}$ is a diagonal invertible matrix. It leads to the final expression \eqref{id_slavnov_gen_rect}.
%
%
%
%
%
\end{proof}


\end{document}